\documentclass[showpacs,amsmath,amssymb,twocolumn,superscriptaddress,notitlepage,preprintnumbers,longbibliography,pra]{revtex4-1}

\usepackage{braket}
\usepackage{physics}
\usepackage{mathtools}
\usepackage{diagbox}
\usepackage[utf8]{inputenc}
\usepackage{algorithm}
\usepackage{color}
\usepackage{amssymb}
\usepackage{amsmath}%
\usepackage{algpseudocode}
\usepackage{algorithm}
\usepackage{qcircuit}
\usepackage{orcidlink}
\usepackage{titlesec}
\usepackage{cprotect}
\usepackage{graphicx} 
\usepackage{tikz-network} 
\usepackage[dvipsnames]{xcolor}
\usepackage{amsthm}

\newcommand{\BB}{\mathcal{B}}

\renewcommand{\tr}{\text{Tr}}

\newcommand{\OO}{\mathcal{O}}

\newcommand{\LL}{\mathcal{L}}
\newcommand{\II}{\boldsymbol{1}}

\newcommand{\bos}{\boldsymbol}
\newcommand{\shortxymatrix}[1]{\xymatrix@1@C=.6cm{#1}}
\newcommand{\EZero}{\shortxymatrix{\ar@{-}[r]&}}

\newcommand{\EE}{\mathcal{E}}

\newcommand{\CC}{\mathcal{C}}

\newcommand{\NN}{\mathcal{N}}
\newcommand{\ii}{\mathrm{i}}
\newcommand{\im}{\mathrm{i}}
\renewcommand{\SS}{\mathcal{S}}
\newcommand{\ZZ}{\mathbb{Z}}

\usepackage{hyperref}
\hypersetup{colorlinks=true, linkcolor=blue, citecolor=blue, urlcolor=magenta }

\newtheorem{proposition}{Proposition}

\usepackage[author=cxx,color=blue]{changes}

\makeatletter
\renewcommand{\@fnsymbol}[1]{%
  \ifcase#1\or
    $*$, $\dagger$\or
    $*$\or
    $\ddagger$\or
    $\S$\or
    \P\or
    \|\or
    **\or
    \dagger\dagger
  \fi
}
\makeatother

\begin{document}

\title{Phase transitions in parametrized quantum circuits}

\author{Xiaoyang Wang\orcidlink{0000-0002-2667-1879}}
\thanks{These authors contributed equally to the work,\\ \href{xiaoyang.wang@riken.jp}{xiaoyang.wang@riken.jp}}

\affiliation{RIKEN Center for Interdisciplinary Theoretical and Mathematical Sciences (iTHEMS), Wako 351-0198, Japan}
\affiliation{RIKEN Center for Computational Science (R-CCS), Kobe 650-0047, Japan}
\author{Han Xu}
\thanks{These authors contributed equally to the work.}
\affiliation{RIKEN Center for Computational Science (R-CCS), Kobe 650-0047, Japan}
\author{Lukas Broers}
\thanks{These authors contributed equally to the work.}
\affiliation{RIKEN Center for Computational Science (R-CCS), Kobe 650-0047, Japan}
\author{Tomonori Shirakawa\orcidlink{0000-0001-7923-216X}}
\affiliation{RIKEN Center for Computational Science (R-CCS), Kobe 650-0047, Japan}
\affiliation{RIKEN Center for Quantum Computing (RQC), Wako 351-0198, Japan}
\affiliation{RIKEN Center for Interdisciplinary Theoretical and Mathematical Sciences (iTHEMS), Wako 351-0198, Japan}
\affiliation{RIKEN Pioneering Research Institute (PRI), Wako 351-0198, Japan}
\author{Seiji Yunoki}
\email{yunoki@riken.jp}
\affiliation{RIKEN Center for Computational Science (R-CCS), Kobe 650-0047, Japan}
\affiliation{RIKEN Center for Quantum Computing (RQC), Wako 351-0198, Japan}
\affiliation{RIKEN Cluster for Pioneering Research (CPR), Wako 351-0198, Japan}
\affiliation{RIKEN Center for Emergent Matter Science (CEMS), Wako 351-0198, Japan}

\begin{abstract}
    Phase transitions are among the most intriguing phenomena in physical systems, yet their behavior near criticality remain challenging to study using classical algorithms. Parameterized quantum circuits (PQCs) offer a promising approach to investigating such regimes on practical quantum computers. However, in order to use it to probe critical behavior, a PQC itself should be non-trivial and exhibit a phase transition and non-analyticity~---~a property that has not yet been clearly identified. In this work, we identify a mechanism for generating non-analyticities intrinsically in PQCs. As a concrete realization, we construct a class of sequential PQCs whose observable expectation value is a non-analytic function of the circuit parameter in the infinite volume limit, showing that the prepared PQC states undergo a phase transition at the non-analytic points. The entanglement and the identified order parameter have distinct behaviors in different phases, revealing a phase diagram of the PQC state. We show that classical simulation of this PQC based on tensor networks and Pauli propagation gets less efficient in the vicinity of the phase transition point, indicating a physically motivated route towards practical quantum advantage using PQCs with phase transitions. 
\end{abstract}

\date{\today}
\maketitle

Phase transitions are fundamental phenomena in various physical systems, characterized by the emergence of long-range correlations and non-analyticity in macroscopic properties as the external parameters, such as temperature or applied field strength, are varied~\cite{huang2008statistical,Sachdev2011,negele2018quantum,brankov2000theory,guenther2021overviewqcdphasediagram,Belkin_2019,nakkiran2019deepdoubledescentbigger,kaplan2020scalinglawsneurallanguage}. 
These correlated effects make phase transitions notoriously difficult to study using classical algorithms like tensor networks and Monte-Carlo~\cite{ORUS2014117,WOLFF199093}.

Quantum computers provide an opportunity for investigating long-range correlated systems beyond the reach of classical simulation. In particular, parametrized quantum circuits (PQCs) are promising to prepare highly entangled quantum states and to demonstrate the practical utility of currently available quantum computers~\cite{Kandala2017,McArdle_20,TILLY2022,Di_Meglio_2024,Abbas_2024,Guo2024,Moreno25}. 
One use case of PQCs is to prepare ground or thermal states of quantum systems within the variational quantum eigensolver (VQE) framework~\cite{Peruzzo:2014,Wecker_2015,Cerezo_vqa,Wu_2019}. Recent numerical studies have demonstrated that PQCs can exhibit signatures of phase transitions in various model systems~\cite{zfdt-1k63,Wang23,guo2023performancevqephasetransition,Takis_2025}. However, due to the limited variational-state-preparation accuracy~\cite{farhi2014quantum,PhysRevA.103.042612,Bravyi_2020}, it remains unclear to what extent PQCs are capable of generating strong correlations and demonstrating the phase transition behavior, especially in the infinite volume limit. Additionally, the utility of PQCs beyond classical algorithms is challenged by classical simulation algorithms, including tensor network and Pauli propagation methods~\cite{Berezutskii_25,doi:10.1126/science.ado6285,Cerezo2025,10.1145/3564246.3585234,PhysRevLett.133.120603,Fontana_2025,Angrisani2025,lerch2024efficientquantumenhancedclassicalsimulation,rudolph2025paulipropagationcomputationalframework,broers2024exclusiveorencodedalgebraicstructure,orqa}. 

\begin{figure}
    \centering
    \includegraphics[width=0.47\textwidth]{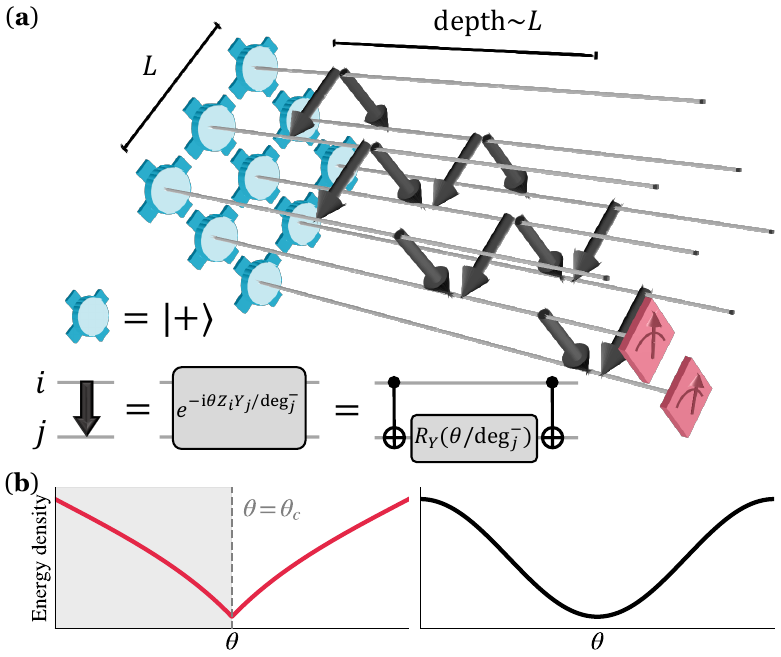}
    \caption{(\textbf{a}) Illustration of a sequential parametrized quantum circuit (PQC) that exhibits a phase transition. Qubits on a 2-d lattice are initialized in a $\ket{+}$ product state. Black arrows denote $ZY$ rotations applied sequentially to the qubits. (\textbf{b}) Energy landscape of PQCs with (left panel) and without (right panel) phase transitions. The energy cusp singularity at $\theta=\theta_c$ signals a phase transition dividing the shaded and blank parameter regions in the left panel.}
    \label{fig:intro-figure}
\end{figure}

In this work, we provide positive evidence on the capacity of PQCs to generate long-range correlated quantum states by showing the existence of phase transitions in a class of PQCs, which consists of a ladder-like sequence of parameterized gates, applied to an initial product state, as illustrated in Fig.~\ref{fig:intro-figure}(\textbf{a}). This sequential structure plays an important role in the state preparation using tensor network~\cite{PhysRevLett.95.110503,PhysRevA.75.032311,PhysRevLett.128.010607} and quantum algorithms~\cite{PRXQuantum.2.010342,PhysRevLett.101.180506,Chen_2024,wang2025performanceguaranteeslightconevariational,PhysRevB.107.L041109,Cochran_2025}. 
Using the language of Pauli propagation, we rigorously prove the emergence of energy cusp singularities, as illustrated in the left panel of Fig.~\ref{fig:intro-figure}(\textbf{b}). This shows that despite being constructed from analytic gates, sequential PQCs can exhibit non-analytic expectation values in the infinite volume limit, closely paralleling conventional phase transitions in statistical physics~\cite{PhysRev.87.404,PhysRev.87.410,PhysRevB.111.045139}. We further confirm that the cusp singularities divide the circuit states into distinct phases characterized by sharply different entanglement properties and order parameter values, revealing a phase diagram at the circuit level.
\\\\
\textbf{Non-analyticity in PQC}~---~We first show how the non-analytic behavior arises in PQCs using the language of Pauli propagation~\cite{10.1145/3564246.3585234,PhysRevLett.133.120603,Fontana_2025,Angrisani2025,lerch2024efficientquantumenhancedclassicalsimulation,rudolph2025paulipropagationcomputationalframework,broers2024exclusiveorencodedalgebraicstructure,orqa}. The physical properties of a quantum system are quantified by local observables $\hat{O}$ represented in the Pauli basis as $\hat{O}=\sum_I c_I P_I$, with $c_I\in\mathbb{R}$ and $P_I$ denoting Pauli strings.
We denote the quantum state produced by a PQC as $\ket{\phi(\theta)} = V(\theta)\ket{\phi_0}$. $V(\theta)\equiv  U_K(\theta)\cdots U_1(\theta)$ is composed of $K$ non-commuting layers of quantum gates. The Pauli string expectation is defined as
\begin{align}
    \langle P_I\rangle_{\theta} \equiv \bra{\phi(\theta)}P_I\ket{\phi(\theta)}=\bra{\phi_0}V^{\dagger}(\theta)P_I V(\theta)\ket{\phi_0}.\nonumber
\end{align}
Here we consider circuits in which all layers share the same parameter $\theta$. Using Pauli propagation methods, the expectation value is computed layer-by-layer: after the first layer $U_K^{\dagger}P_I U_K$, the initial Pauli string $P_I$ is transformed into a weighted sum of many Pauli strings, each of which is then acted on by the next layer $U_{K-1}$, generating more strings, and so on. Finally, the expectation $\langle P_I\rangle_{\theta}$ becomes the weighted sum of expectation values of individual Pauli strings with respect to $\ket{\phi_0}$.

\begin{figure}
    \centering
    \includegraphics[width=0.48\textwidth]{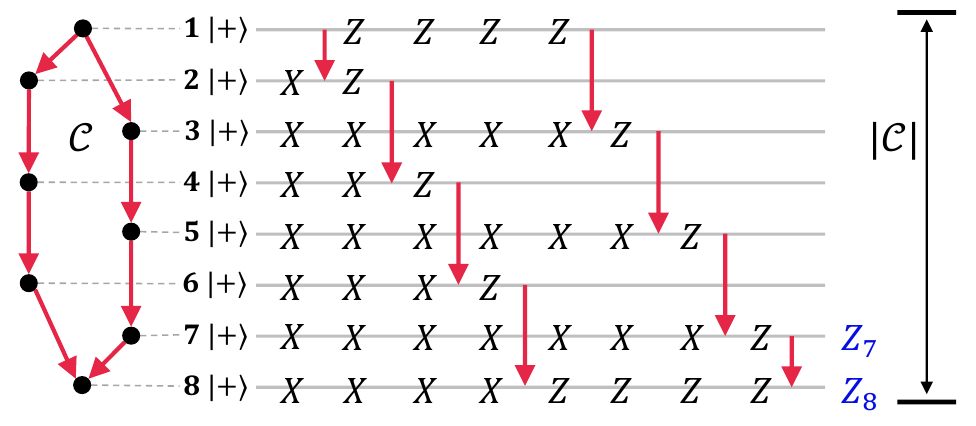}
    \caption{Pauli propagation of the observable $Z_7Z_8$ in a PQC that consists of a $ZY$ rotation cycle $\CC$. Each red arrow denotes a $ZY$ rotation $e^{-\im\theta Z_iY_j/2}$ specified in Fig.~\ref{fig:intro-figure}(\textbf{a}). The vertically arranged Pauli strings together form a Pauli path contributing to the expectation $\langle Z_7Z_8\rangle_{\theta}$.}
    \label{fig:non_local_path}
\end{figure}

We start with a minimal one-dimensional example of a PQC composed of Pauli-$ZY$ rotations $e^{-\im\theta Z_iY_j/2}$, acting on the initial product state $\ket{\phi_0}=\ket{+}^{\otimes N}$, which is the $+1$ eigenstate of Pauli-$X$ strings.
The $ZY$ rotations are arranged sequentially along a closed cycle $\CC$ of length $|\CC|=N$, as illustrated in Fig.~\ref{fig:non_local_path}. Consider the expectation of a local observable $\hat{O} = Z_7Z_8$ located at the bottom of the cycle. Under successive application of the $ZY$ rotations, the observable is deformed into a sequence of vertically arranged Pauli strings in Fig.~\ref{fig:non_local_path}, which we refer to as a \emph{Pauli path}. Each transformation along the path contributes a factor of either $\sin\theta$ or $\cos\theta$, depending on whether the Pauli string is deformed or not. After all $ZY$ rotations have been performed, the expectation takes a closed-form expression
\begin{align}
    \langle\hat{O}\rangle_\theta = -\sin\theta\cos\theta+(-\sin\theta)^{|\CC|-1}\cos\theta,
\end{align}
where the second term comes from the non-local Pauli path in Fig.~\ref{fig:non_local_path}. 

In this PQC, non-analyticity arises from the proliferation of non-local Pauli paths. 
For a single closed cycle, the $k_0$th derivative of the observable expectation is
\begin{align}
    \partial^{k_0}_{\theta}\langle \hat{O}\rangle_\theta=|\CC|^{k_0}(-\sin\theta)^{|\CC|-k_0-1}\cos^{k_0+1}\theta+\OO(|\CC|^{k_0-1}).\nonumber
\end{align}
The key factor is the number $N_{\mathcal C}$ of distinct non-local cycles with respect to the cycle length $|\mathcal C|$. If $N_{\CC}$ grows only polynomially with $|\CC|$, i.e., $N_{\CC}=\mathcal{O}(\mathrm{poly}(|\CC|))$, then the total contribution $N_{\CC}\,\partial^{k_0}_{\theta}\langle \hat{O}\rangle_\theta$ remains well-behaved. In contrast, if $N_\CC$ proliferates exponentially as $N_{\CC}\sim\OO({b^{|\CC|}}/{|\CC|^{k}})$, then the total contribution in the $|\CC|\to\infty$ limit scales as
\begin{align}
   N_{\CC}\,\partial^{k_0}_{\theta}\langle \hat{O}\rangle_\theta\sim\OO(|\CC|^{k_0-k}(b\sin\theta)^{|\CC|}),
   \label{eq:divergent-derivative}
\end{align}
where non-analyticity occurs whenever $b\sin\theta\ge 1$ and the polynomial factor $|\CC|^{k_0-k}$ fails to suppress the growth. For example, if $b\sin\theta=1, k<2$, the 2nd derivative is divergent at the point $\theta=\arcsin(1/b)$. A divergent 2nd derivative corresponds to the discontinuity in a 1st derivative, and thus a cusp singularity in the observable expectation value, as illustrated in Fig.~\ref{fig:intro-figure}(\textbf{b}).

The Pauli paths are non-local only if the depth of the PQC grows with the system size, which is satisfied in sequential PQCs. Additionally, the sequential circuit has finite local depth that is provably barren-plateau-free~\cite{zhang_absence_2024}, and therefore well-suited for efficient optimization. In the following content, we prove that sequential PQCs can have exponentially many non-local Pauli paths and thus non-analytic observable expectations.
\\\\
\textbf{Ising whip circuit}~---~We extend the above example to a sequential PQC on a 2-d $L \times L$ lattice, where $N=L^2$ qubits are initialized to $\ket{\phi_0}=\ket{+}^{\otimes N}$, and Pauli-$ZY$ rotations act layer-by-layer from the top to the bottom of the lattice, as shown in Fig.~\ref{fig:intro-figure}(\textbf{a}). To account for boundary effects, each $ZY$ rotation applied to a qubit pair $(i,j)$ is parametrized by $\theta/\deg^-_j$, where $\deg^-_j$ represents the number of $ZY$ rotations in which the generators acts on $j$ with $Y_j$. The action of the circuit can be understood readily at the two-site level.
A Pauli-$ZY$ rotation acting on qubit pair $(1,2)$ takes the form $e^{-\im \theta Z_1Y_2} = \ket{0}\bra{0}_1\otimes e^{-\im \theta Y_2} +\ket{1}\bra{1}_1\otimes e^{\im \theta Y_2}$, 
indicating that the second qubit is whipped around the $Y$-axis of the Bloch sphere in the direction controlled by the state of the first qubit.
Taking $\theta=-\pi/4$ leads to the Bell state $(\ket{00}+\ket{11})/\sqrt{2}$, as illustrated in Fig.~\ref{fig:whip_circuit}.
For the full PQC, the effective control qubit is the top qubit in Fig.~\ref{fig:intro-figure}(\textbf{a}), and the remaining qubits are whipped, such that the PQC state reads
\begin{equation}
    \begin{aligned}
    \ket{\phi_{\text{w}}(\theta)}&\equiv \prod_{\langle i,j\rangle} \left. e^{-\im\theta Z_i Y_j/\deg^-_j} \ket{+}^{\otimes N}\right|_{\theta =-\pi/4}\\
    &=\frac{1}{\sqrt{2}}(\ket{0}^{\otimes N}+\ket{1}^{\otimes N}).
    \label{eq:GHZ-state}
\end{aligned}
\end{equation}
This state is the ground state of the $2$-d ferromagnetic Ising Hamiltonian $H= -\frac{1}{|\EE|}\sum_{\langle i,j\rangle }Z_iZ_j$. The normalization factor $|\EE|\equiv2L(L-1)$ ensures that the expectation $\langle H \rangle_{\theta}\equiv\bra{\phi_{\text{w}}(\theta)}H\ket{\phi_{\text{w}}(\theta)}$ represents the energy density of the Ising system. We refer to this sequential PQCs as 2-d \textit{whip circuits}. Generalizations to arbitrary spatial dimensions are discussed in the Supplemental Material~\cite{supp}. The $ZY$ rotation is derived from the imaginary-time evolution of the Ising Hamiltonian~\cite{PhysRevA.111.032612}. In the same way, the sequential PQC can be constructed for other models like $\ZZ_2$ gauge theory, as detailed in the End Note.

\begin{figure}
    \centering
    \includegraphics[width=0.48\textwidth]{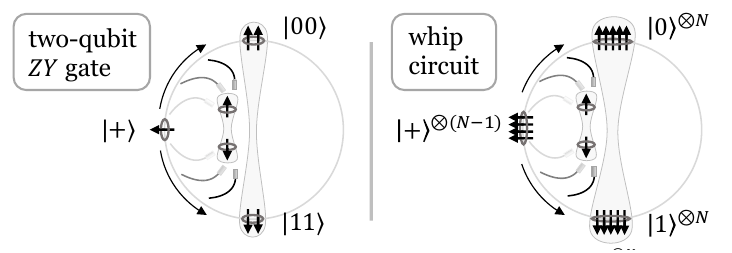}
    \caption{Whip circuit preparing the Ising ground state. The control qubit at the center of the circle whips the remaining qubits to align with it. The circle is the Bloch sphere in the $x$-$z$ plane. The gray dumbbells indicate superposition.}
    \label{fig:whip_circuit}
\end{figure}
\begin{figure}
    \centering
    \includegraphics[width=0.48\textwidth]{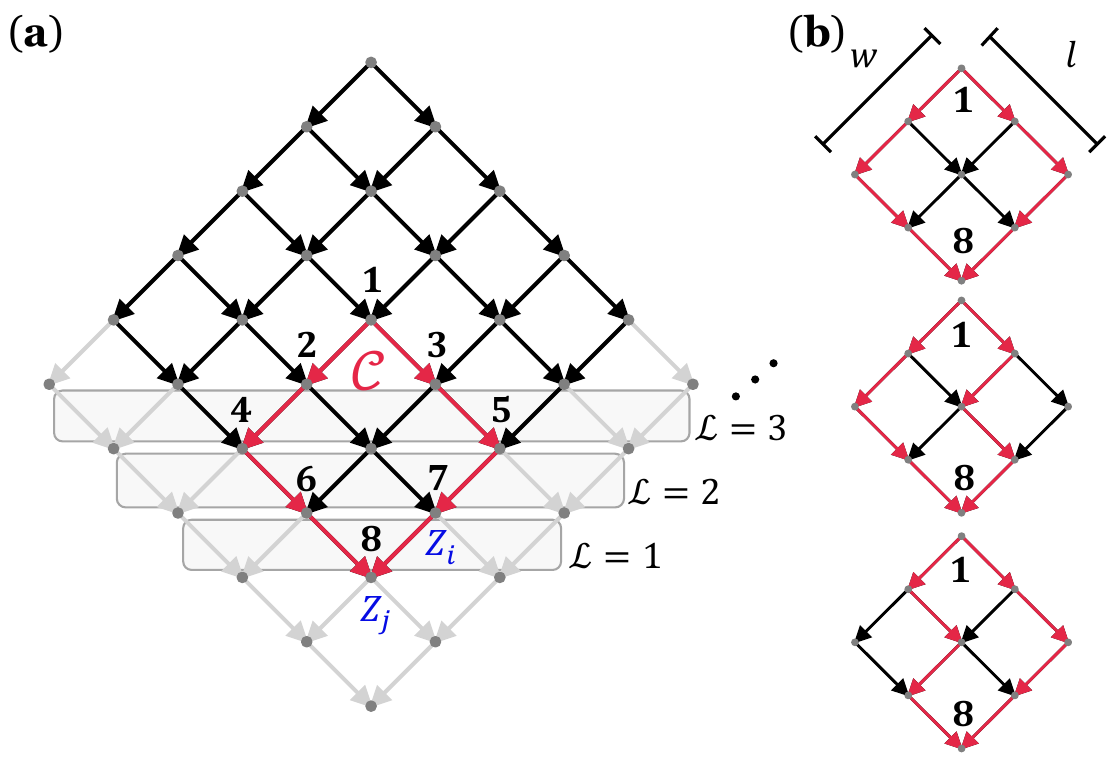}
    \caption{Non-local Pauli paths of the 2-d whip circuit. (\textbf{a}) Black arrows denote $ZY$ rotations that act non-trivially on the observable $Z_iZ_j$. The red-highlighted cycle corresponds to a non-local Pauli path similar to that in Fig.~\ref{fig:non_local_path}. (\textbf{b}) Three non-local Pauli paths in the length-$l$ width-$w$ rectangle contributing to the expectation $\langle Z_i Z_j\rangle_{\theta}$.}
    \label{fig:Pauli-paths}
\end{figure}

\begin{figure*}
    \centering
    \includegraphics[width=1\textwidth]{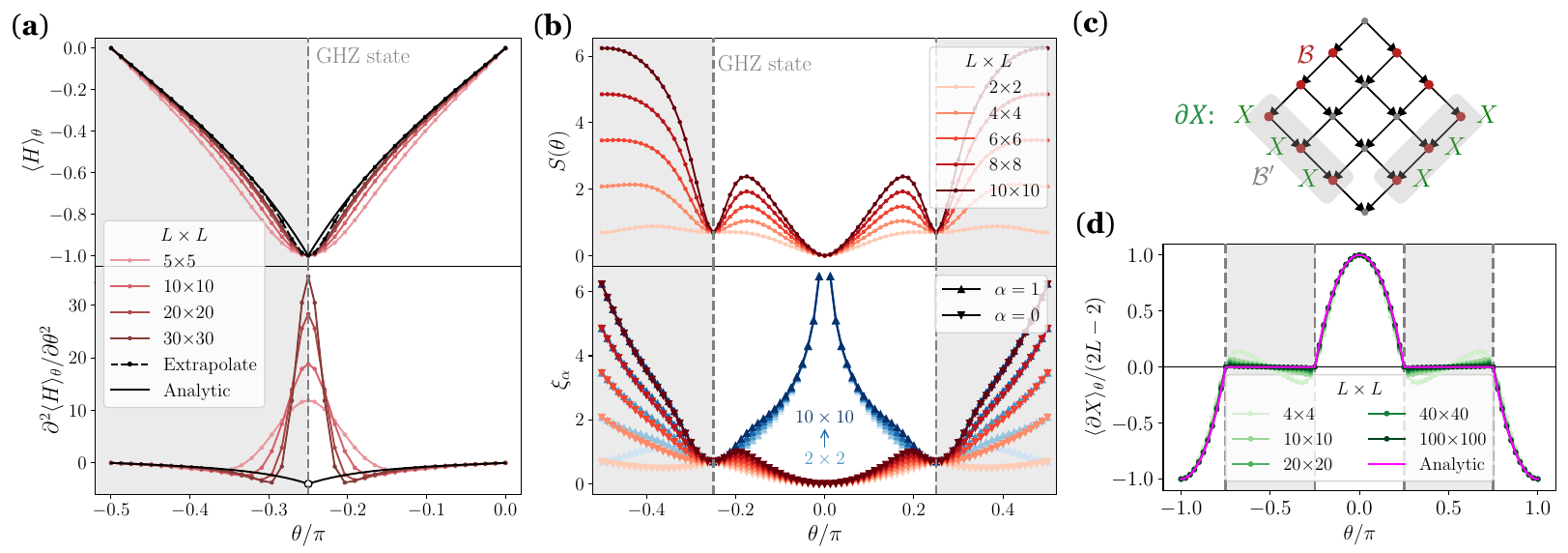}
    \caption{(\textbf{a}) Energy density of the Ising whip circuit and its 2nd derivative as a function of the whip angle $\theta$. Numerical curves on $L\times L$ lattices and their infinite volume extrapolated curve converge to the analytic one. $\theta_c=-\pi/4$ is the non-analytic point with the divergent derivative. (\textbf{b}) Entanglement entropy and entanglement spectrum of the Ising whip circuit on lattices from $2\times 2$ to $10\times 10$. They have distinguished behavior in sectors divided by the non-analytic points $\theta_c=\pm\pi/4$. (\textbf{c}) Boundary ($\mathcal{B}$) and lower boundary ($\mathcal{B}'$) of the 2-d lattice. The order parameter operator is the summation of Pauli-$X$ operators at the lower boundary. (\textbf{d}) Order parameter of the Ising whip circuit as a function of the whip angle. The transitions from zero to non-zero at $\theta_c=\pi/4+m\pi/2, m\in\ZZ$ signal spontaneous symmetry breaking.}
    \label{fig:energy-landscape}
\end{figure*}

\begin{figure}[b]
    \centering
    \includegraphics[width=0.36\textwidth]{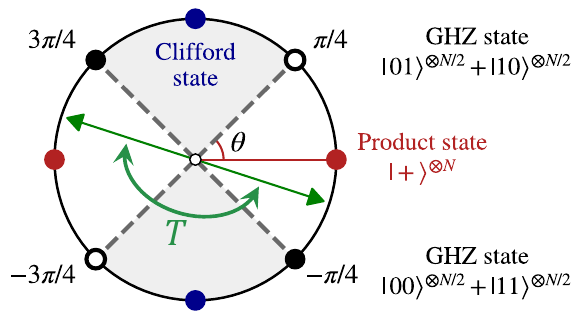}
    \caption{Phase diagram of the Ising whip circuit with the parameter $\theta\in(-\pi,\pi]$. The blank and shaded regions are two phases divided by the non-analytic points $\theta_c$. $T:\theta\rightarrow \theta+\pi$ denotes the symmetry transformation.}
    \label{fig:phase}
\end{figure}

The 2-d whip circuit generates exponentially many non-local Pauli paths as required by the PQC non-analyticity. The local observable $Z_iZ_j$ is transformed non-trivially by $ZY$ rotations denoted by black arrows in Fig.~\ref{fig:Pauli-paths}(\textbf{a}). The expectation value $\langle Z_iZ_j\rangle_{\theta}$ is a linear combination of non-local Pauli paths. Similar to Fig.~\ref{fig:non_local_path}, each Pauli path corresponds to a red-highlighted oriented cycle supported on a rectangular region with length $l$ and width $w$, as illustrated in Fig.~\ref{fig:Pauli-paths}(\textbf{b}). These cycles have the same length $|\CC|=2(l+w)$. The number of cycles supported on this rectangle is counted using the Lindstr\"om-Gessel-Viennot (LGV) lemma~\cite{Gessel1989DeterminantsPA}
\begin{align}
    N_{\CC}(l,w) = \frac{lw}{(l+w-1)(l+w)^2}\binom{l+w}{l}^2.
    \label{eq:main-LGV-lemma}
\end{align}
Using the Stirling formula, for example, consider the special case $l=w\to\infty$,
\begin{align}
    \lim_{l=w\to\infty}N_{\CC}(l,w) = \frac{2}{\pi|\CC|^2}2^{|\CC|}.
    \label{eq:n-pauli-path-scaling}
\end{align}
Thus, the exponentially growing $N_{\CC}(l,w)$ with respect to $|\CC|$ satisfies the criterion for divergent derivatives in Eq.~\eqref{eq:divergent-derivative}.

Using Eq.~\eqref{eq:main-LGV-lemma}, in the infinite volume limit $L\to \infty$, the expectation value $\langle Z_iZ_j\rangle_{\theta}$ of the whip circuit can be expressed as an infinite series~\cite{supp}
\begin{align}
    \langle Z_iZ_j\rangle_{\theta}=-\sum_{\LL=1}^{\infty}\frac{4^{\LL-1}\Gamma(\LL-\frac{1}{2})}{\sqrt{\pi}\LL !}(\cos\theta\sin\theta)^{2\LL-1},
    \label{eq:ZZ-series}
\end{align}
where $\LL$ labels the $ZY$-rotation layers [see gray boxes in Fig.~\ref{fig:Pauli-paths}(\textbf{a})], and $\Gamma(x)$ is the Gamma function. 
This series has a closed form, yielding the Ising energy density
\begin{align}
\lim_{L\to \infty}\langle H\rangle_{\theta}= \langle Z_iZ_j\rangle_{\theta}=\frac{1-|\cos 2\theta|}{\sin 2\theta},
    \label{eq:analytic-formula}
\end{align}
which has cusp singularities at $\theta_c= \pi/4+m\pi/2, m\in\mathbb{Z}$. Numerical results in Fig.~\ref{fig:energy-landscape}(\textbf{a}) confirm convergence to the analytic prediction as the system size increases, with the second derivative $\partial_\theta^2\langle H\rangle_\theta$ growing rapidly near $\theta=-\pi/4$.
\\\\
\textbf{Phase diagram}~---~The PQC state $\ket{\phi_{\text{w}}(\theta)}$ has distinct behavior in two phases separated by the non-analytic points $\theta_c$. We illustrate this by studying the entanglement properties and order parameter.

Figure~\ref{fig:energy-landscape}(\textbf{b}) shows the entanglement entropy $S(\theta)=-\tr (\rho_A\log\rho_A)$, where $\rho_A$ is the reduced density matrix of half of the 2-d lattice, together with the entanglement spectrum $\xi_{\alpha}(\theta)$ from the Schmidt decomposition $\ket{\phi_{\text{w}}(\theta)}=\sum_{\alpha} e^{-\xi_{\alpha}(\theta)/2}\ket{\phi_{\alpha}^A}\otimes |\phi_{\alpha}^{\bar{A}}\rangle$. We find that $S(\theta)$ exhibits different scaling behaviors in sectors divided by $\theta_c=\pm\pi/4$:
$S(0)$ is constant to the system size, whereas $S(\pm\pi/2)$ increase linearly with $L$, indicating an area-law entanglement. 
The low-lying entanglement spectrum changes its degeneracy across $\theta_c$, from non-degenerate within $[-\pi/4,\pi/4]$ to twofold degenerate outside this interval, which is typically observed during quantum phase transitions~\cite{PhysRevLett.108.227201,PhysRevA.78.032329,PhysRevB.83.245134,PhysRevLett.109.237208,PhysRevB.106.014306,Lepori2013}.

The order parameter of the Ising whip circuit can be identified from the symmetry transformation of $\ket{\phi_{\text{w}}(\theta)}$. The energy expectation $\langle H\rangle_{\theta}$ in Eq.~\eqref{eq:analytic-formula} is invariant under a discrete transformation $T:\theta\rightarrow \theta+\pi$. In the End Note, we prove that this transformation corresponds to a symmetry operator $\hat{T}\equiv \prod_{i\in \mathcal{B}}Z_i$, such that 
\begin{align}
    \ket{\phi_{\text{w}}(\theta+\pi)} = e^{i\phi}\hat{T} \ket{\phi_{\text{w}}(\theta)},
    \label{eq:symmetry-transformation}
\end{align}
where $e^{i\phi}$ is a global phase. Here, $\mathcal{B}$ includes the boundary sites of the 2-d lattice as illustrated in Fig.~\ref{fig:energy-landscape}(\textbf{c}). To signal the symmetry breaking, the order parameter operator should anti-commute with the symmetry operator. Thus, we choose the order parameter operator to be the summation of the Pauli-$X$ operators at the \textit{lower} boundary $\BB'$ of the 2-d lattice illustrated in Fig.~\ref{fig:energy-landscape}(\textbf{c})~\footnote{The upper boundary is omitted. Because local observables at sites $i\in \BB/\BB'$ do not commute with $ZY$ rotations only at the 1-d upper boundary of the whip circuit, such that the 2-d structure of the whip circuit is not ``seen'' by these observables. }
\begin{align}
    \partial X \equiv \sum_{i\in\BB'}X_i.
\end{align}
This operator includes $2L-2$ sites and satisfies the anti-commutation relation $\{\hat{T},\partial X\}=0$.

In the infinite volume limit, we derive the analytic expression of the order parameter~\cite{supp}
\begin{align}
    \lim_{L\to\infty}\frac{\langle \partial X\rangle_{\theta}}{2L-2}=\frac{\cos 2\theta +|\cos 2\theta|}{ 2\cos\theta},
\end{align}
which also has cusp singularities at $\theta_c=\pi/4+m\pi/2,m\in\ZZ$. 
As plotted in Fig.~\ref{fig:energy-landscape}(\textbf{d}), the order parameter is zero in the symmetry preserving region $\theta\in[-3\pi/4,-\pi/4]\cup [\pi/4,3\pi/4]$, and non-zero in the symmetry breaking region $\theta\in(-\pi/4,\pi/4)\cup (3\pi/4,-\pi/4)$. The numerical results of 2-d projected entangled pair state (PEPS) converge to the analytic prediction as the lattice size increases from $4\times 4$ to $100\times 100$.

The resulting phase diagram of the Ising whip circuit as a function of $\theta$ is shown in Fig.~\ref{fig:phase}, which implicitly admits a period of $2\pi$ in $\theta$. More details of the phase diagram and discussion on the correlation length are presented in Supplemental Material~\cite{supp}.
\\\\
\textbf{Classical simulability}~---~
Phase transitions can indicate increased classical simulation complexity of PQC. For the Pauli propagation method, the expectation $\langle Z_iZ_j\rangle_{\theta}$ in Eq.~\eqref{eq:ZZ-series} has an asymptotic form for $\LL\gg1$
\begin{align}
    \langle Z_iZ_j\rangle_{\theta}\sim-\sum_{\LL}\frac{1}{\LL^{3/2}}(\sin2\theta)^{2\LL-1}, 
    \label{eq:ZZ-series-scaling}
\end{align}
where the addends transit from exponential to power-law decay with respect to $\LL$ at $\theta=\theta_c$. Since $\LL$ is proportional to the weight of Pauli strings contributing to $\langle Z_iZ_j\rangle_{\theta}$, the transition makes the Pauli propagation method using the weight-based truncation ineffective. Specifically, we prove that naive Pauli propagation displays a time complexity of $T_{\text{naive PP}}\sim \OO(4^{\text{poly}(L)})$ on an $L\times L$ lattice to evaluate $\langle Z_iZ_j\rangle_{\theta}$ within an error $\epsilon \sim \OO(1/\text{poly} (L))$. The exponential scaling is reached at and only at the phase transition point $\theta=\theta_c$. In contrast, on a quantum computer, the same task displays a time complexity of $T_{\text{QC}}\sim \OO(\text{poly}(L))$~\cite{supp}.

To overcome this classical inefficiency, we introduce an early-evaluation strategy for Pauli propagation methods, which exploits the 2-d PQC structure and the initial product state by tracing out qubits mid-circuit when possible. This allows us to classically evaluate $\langle Z_iZ_j\rangle_{\theta}$ on an $L\times L$ lattice exactly using time $T_{\text{PP}}\sim\OO(L^4)$~\cite{supp}. We implement this early-evaluation strategy, and numerically verify the time scaling, using the or-represented quantum algebra (ORQA) framework~\cite{orqa}. However, this polynomial scaling is not straightforwardly generalized to three- and higher-dimensional Ising whip circuits~\cite{supp}.

For the tensor-network method, the whip circuit simulation using matrix product state (MPS) is not tractable due to the area-law entanglement shown in Fig.~\ref{fig:energy-landscape}(\textbf{b}). Instead, we parametrize the Ising whip circuit by PEPS, and numerically find that evaluating $\langle Z_iZ_j\rangle_{\theta}$ with constant truncation cutoffs requires time $T_{\text{TN}}\sim\OO(L^2)$, and the largest evaluation error occurs in the vicinity of the phase transition points~\cite{supp}. However, the efficient PEPS simulation is also not easily applicable to the higher-dimensional whip circuits. Therefore, we conjecture the classical intractability of simulating the high-dimensional Ising whip circuits.

As an example of a higher-dimensional PQC with phase transition, the sequential PQC of $\ZZ_2$ gauge theory in the End Note is constructed on a 3-d lattice, where a cusp singularity similar to that in the Ising whip circuit is observed. 
\\\\
\textbf{Conclusion}~---~This work presents a phase transition that occurs intrinsically in a class of sequential PQCs. This circuit-based phase transition is characterized by  non-analytic points originating from the exponential proliferation of non-local Pauli paths. We introduce an exemplary Ising-type PQC whose energy density exhibits cusp singularities. The PQC states on the two sides of the cusp singularity are in different phases with distinct entanglement properties and order parameter values, such that the phase diagram of the Ising PQC state is obtained. Similar non-analyticity is observed in other models, including the $\ZZ_2$ gauge theory. Moreover, we discuss the difficulty of classically simulating the Ising PQC near criticality using Pauli propagation and tensor network methods, indicating that high-dimensional PQCs with phase transitions are classically intractable and thus promising to demonstrate quantum advantage using practical quantum computers. However, in practice, strong hardware noise has an impact on the non-local Pauli paths of the sequential PQC~\cite{PhysRevLett.133.120603}, and thus on the potential quantum advantage. The hardware demonstration of the phase transition and the study of the noise effects are left for future work.
\\\\
\textbf{Acknowledgment}~---~We thank Tetsuo Hatsuda, Yasumichi Aoki, Lingxiao Wang, Tongyang Li, and Xiao Yuan for insightful discussions and valuable comments. 
This work was partially supported by Project No.~JPNP20017, funded by the New Energy and Industrial Technology Development Organization (NEDO), Japan. 
We acknowledge support from JSPS KAKENHI (Grant No.~JP21H04446) 
from the Ministry of Education, Culture, Sports, Science and Technology (MEXT), Japan.
Additional funding was provided by 
JST COI-NEXT (Grant No.~JPMJPF2221) and by 
the Program for Promoting Research of the Supercomputer Fugaku (Grant No.~MXP1020230411) from MEXT, Japan. 
We further acknowledge support from 
the RIKEN TRIP initiative (RIKEN Quantum), the UTokyo Quantum Initiative, and
the COE research grant in computational science from Hyogo Prefecture and Kobe City through the Foundation for Computational Science. 
\\\\
\textbf{Author contributions}~---~X.W. and H.X. conceived the original study. H.X., L.B., and X.W. carried out the numerical calculations to support the study. L.B. and X.W. performed the analytical derivations. All authors X.W., H.X., L.B., T.S., and S.Y. discussed the results of the manuscript, and all authors contributed to the writing of the manuscript.

\bibliography{main.bib}

@misc{supp,
howpublished = {See Supplemental Material at \url{Supplemental/Material/website} for constructing Ising whip circuit on arbitrary dimensions, properties of 1-d and 2-d Ising whip circuits, details on the whip circuit phase diagram, construction of the 3-d WALA, and the classical simulation time complexity of the Ising whip circuits, which includes Ref. [64,65]} }

@inproceedings{Gessel1989DeterminantsPA,
  title={Determinants, Paths, and Plane Partitions},
  author={Gessel, Ira M. and Viennot, X. G.},
  year={1989},
  url={https://people.brandeis.edu/~gessel/homepage/papers/pp.pdf}
}

@article{ORUS2014117,
	author = {Rom{\'a}n Or{\'u}s},
	doi = {https://doi.org/10.1016/j.aop.2014.06.013},
	issn = {0003-4916},
	journal = {Annals of Physics},
	pages = {117-158},
	title = {A practical introduction to tensor networks: Matrix product states and projected entangled pair states},
	url = {https://www.sciencedirect.com/science/article/pii/S0003491614001596},
	volume = {349},
	year = {2014}}

@article{Berezutskii_25,
	author = {Berezutskii, Aleksandr and Liu, Minzhao and Acharya, Atithi and Ellerbrock, Roman and Gray, Johnnie and Haghshenas, Reza and He, Zichang and Khan, Abid and Kuzmin, Viacheslav and Lyakh, Dmitry and Lykov, Danylo and Mandr{\`a}, Salvatore and Mansell, Christopher and Melnikov, Alexey and Melnikov, Artem and Mironov, Vladimir and Morozov, Dmitry and Neukart, Florian and Nocera, Alberto and Perlin, Michael A. and Perelshtein, Michael and Steinberg, Matthew and Shaydulin, Ruslan and Villalonga, Benjamin and Pflitsch, Markus and Pistoia, Marco and Vinokur, Valerii and Alexeev, Yuri},
	da = {2025/10/01},
	doi = {10.1038/s42254-025-00853-1},
	id = {Berezutskii2025},
	isbn = {2522-5820},
	journal = {Nature Reviews Physics},
	number = {10},
	pages = {581--593},
	title = {Tensor networks for quantum computing},
	ty = {JOUR},
	url = {https://doi.org/10.1038/s42254-025-00853-1},
	volume = {7},
	year = {2025}}

@article{KITAEV20032,
	author = {A.Yu. Kitaev},
	doi = {https://doi.org/10.1016/S0003-4916(02)00018-0},
	issn = {0003-4916},
	journal = {Annals of Physics},
	number = {1},
	pages = {2-30},
	title = {Fault-tolerant quantum computation by anyons},
	url = {https://www.sciencedirect.com/science/article/pii/S0003491602000180},
	volume = {303},
	year = {2003}}

@article{PhysRevLett.101.180506,
  title = {Sequential Implementation of Global Quantum Operations},
  author = {Lamata, L. and Le\'on, J. and P\'erez-Garc\'{\i}a, D. and Salgado, D. and Solano, E.},
  journal = {Phys. Rev. Lett.},
  volume = {101},
  issue = {18},
  pages = {180506},
  numpages = {4},
  year = {2008},
  month = {Oct},
  publisher = {American Physical Society},
  doi = {10.1103/PhysRevLett.101.180506},
  url = {https://link.aps.org/doi/10.1103/PhysRevLett.101.180506}
}

@article{PRXQuantum.2.010342,
  title = {Real- and Imaginary-Time Evolution with Compressed Quantum Circuits},
  author = {Lin, Sheng-Hsuan and Dilip, Rohit and Green, Andrew G. and Smith, Adam and Pollmann, Frank},
  journal = {PRX Quantum},
  volume = {2},
  issue = {1},
  pages = {010342},
  numpages = {15},
  year = {2021},
  month = {Mar},
  publisher = {American Physical Society},
  doi = {10.1103/PRXQuantum.2.010342},
  url = {https://link.aps.org/doi/10.1103/PRXQuantum.2.010342}
}

@article{PhysRevA.75.032311,
  title = {Sequential generation of matrix-product states in cavity QED},
  author = {Sch\"on, C. and Hammerer, K. and Wolf, M. M. and Cirac, J. I. and Solano, E.},
  journal = {Phys. Rev. A},
  volume = {75},
  issue = {3},
  pages = {032311},
  numpages = {10},
  year = {2007},
  month = {Mar},
  publisher = {American Physical Society},
  doi = {10.1103/PhysRevA.75.032311},
  url = {https://link.aps.org/doi/10.1103/PhysRevA.75.032311}
}

@article{PhysRevLett.95.110503,
  title = {Sequential Generation of Entangled Multiqubit States},
  author = {Sch\"on, C. and Solano, E. and Verstraete, F. and Cirac, J. I. and Wolf, M. M.},
  journal = {Phys. Rev. Lett.},
  volume = {95},
  issue = {11},
  pages = {110503},
  numpages = {4},
  year = {2005},
  month = {Sep},
  publisher = {American Physical Society},
  doi = {10.1103/PhysRevLett.95.110503},
  url = {https://link.aps.org/doi/10.1103/PhysRevLett.95.110503}
}

@article{Takis_2025,
	author = {Angelides, Takis and Naredi, Pranay and Crippa, Arianna and Jansen, Karl and K{\"u}hn, Stefan and Tavernelli, Ivano and Wang, Derek S.},
	da = {2025/01/18},
	doi = {10.1038/s41534-024-00950-6},
	id = {Angelides2025},
	isbn = {2056-6387},
	journal = {npj Quantum Information},
	number = {1},
	pages = {6},
	title = {First-order phase transition of the Schwinger model with a quantum computer},
	ty = {JOUR},
	url = {https://doi.org/10.1038/s41534-024-00950-6},
	volume = {11},
	year = {2025}}

@misc{guo2023performancevqephasetransition,
      title={The Performance of VQE across a phase transition point in the $J_1$-$J_2$ model on kagome lattice}, 
      author={Yuheng Guo and Mingpu Qin},
      year={2023},
      eprint={2306.04851},
      archivePrefix={arXiv},
      primaryClass={cond-mat.str-el},
      url={https://arxiv.org/abs/2306.04851}, 
}

@article{zfdt-1k63,
  title = {Variational simulation of quantum phase transitions induced by boundary fields},
  author = {Duriez, Alan and Saguia, Andreia and Sarandy, Marcelo S.},
  journal = {Phys. Rev. B},
  volume = {112},
  issue = {9},
  pages = {094401},
  numpages = {10},
  year = {2025},
  month = {Sep},
  publisher = {American Physical Society},
  doi = {10.1103/zfdt-1k63},
  url = {https://link.aps.org/doi/10.1103/zfdt-1k63}
}

@article{Moreno25,
	author = {Javier Robledo-Moreno and Mario Motta and Holger Haas and Ali Javadi-Abhari and Petar Jurcevic and William Kirby and Simon Martiel and Kunal Sharma and Sandeep Sharma and Tomonori Shirakawa and Iskandar Sitdikov and Rong-Yang Sun and Kevin J. Sung and Maika Takita and Minh C. Tran and Seiji Yunoki and Antonio Mezzacapo},
	doi = {10.1126/sciadv.adu9991},
	eprint = {https://www.science.org/doi/pdf/10.1126/sciadv.adu9991},
	journal = {Science Advances},
	number = {25},
	pages = {eadu9991},
	title = {Chemistry beyond the scale of exact diagonalization on a quantum-centric supercomputer},
	url = {https://www.science.org/doi/abs/10.1126/sciadv.adu9991},
	volume = {11},
	year = {2025}}

@article{Guo2024,
	author = {Guo, Shaojun and Sun, Jinzhao and Qian, Haoran and Gong, Ming and Zhang, Yukun and Chen, Fusheng and Ye, Yangsen and Wu, Yulin and Cao, Sirui and Liu, Kun and Zha, Chen and Ying, Chong and Zhu, Qingling and Huang, He-Liang and Zhao, Youwei and Li, Shaowei and Wang, Shiyu and Yu, Jiale and Fan, Daojin and Wu, Dachao and Su, Hong and Deng, Hui and Rong, Hao and Li, Yuan and Zhang, Kaili and Chung, Tung-Hsun and Liang, Futian and Lin, Jin and Xu, Yu and Sun, Lihua and Guo, Cheng and Li, Na and Huo, Yong-Heng and Peng, Cheng-Zhi and Lu, Chao-Yang and Yuan, Xiao and Zhu, Xiaobo and Pan, Jian-Wei},
	da = {2024/08/01},
	doi = {10.1038/s41567-024-02530-z},
	id = {Guo2024},
	isbn = {1745-2481},
	journal = {Nature Physics},
	number = {8},
	pages = {1240--1246},
	title = {Experimental quantum computational chemistry with optimized unitary coupled cluster ansatz},
	ty = {JOUR},
	url = {https://doi.org/10.1038/s41567-024-02530-z},
	volume = {20},
	year = {2024}}

@misc{guenther2021overviewqcdphasediagram,
      title={Overview of the QCD phase diagram -- Recent progress from the lattice}, 
      author={Jana N. Guenther},
      year={2021},
      eprint={2010.15503},
      archivePrefix={arXiv},
      primaryClass={hep-lat},
      url={https://arxiv.org/abs/2010.15503}, 
}

@misc{kaplan2020scalinglawsneurallanguage,
      title={Scaling Laws for Neural Language Models}, 
      author={Jared Kaplan and Sam McCandlish and Tom Henighan and Tom B. Brown and Benjamin Chess and Rewon Child and Scott Gray and Alec Radford and Jeffrey Wu and Dario Amodei},
      year={2020},
      eprint={2001.08361},
      archivePrefix={arXiv},
      primaryClass={cs.LG},
      url={https://arxiv.org/abs/2001.08361}, 
}

@misc{nakkiran2019deepdoubledescentbigger,
      title={Deep Double Descent: Where Bigger Models and More Data Hurt}, 
      author={Preetum Nakkiran and Gal Kaplun and Yamini Bansal and Tristan Yang and Boaz Barak and Ilya Sutskever},
      year={2019},
      eprint={1912.02292},
      archivePrefix={arXiv},
      primaryClass={cs.LG},
      url={https://arxiv.org/abs/1912.02292}, 
}

@article{Belkin_2019,
   title={Reconciling modern machine-learning practice and the classical bias–variance trade-off},
   volume={116},
   ISSN={1091-6490},
   url={http://dx.doi.org/10.1073/pnas.1903070116},
   DOI={10.1073/pnas.1903070116},
   number={32},
   journal={Proceedings of the National Academy of Sciences},
   publisher={Proceedings of the National Academy of Sciences},
   author={Belkin, Mikhail and Hsu, Daniel and Ma, Siyuan and Mandal, Soumik},
   year={2019},
   month=jul, pages={15849–15854} }

@misc{wang2025performanceguaranteeslightconevariational,
      title={Performance guarantees of light-cone variational quantum algorithms for the maximum cut problem}, 
      author={Xiaoyang Wang and Yuexin Su and Tongyang Li},
      year={2025},
      eprint={2504.12896},
      archivePrefix={arXiv},
      primaryClass={quant-ph},
      url={https://arxiv.org/abs/2504.12896}, 
}

@article{PhysRevB.107.L041109,
  title = {Parametrized quantum circuit for weight-adjustable quantum loop gas},
  author = {Sun, Rong-Yang and Shirakawa, Tomonori and Yunoki, Seiji},
  journal = {Phys. Rev. B},
  volume = {107},
  issue = {4},
  pages = {L041109},
  numpages = {5},
  year = {2023},
  month = {Jan},
  publisher = {American Physical Society},
  doi = {10.1103/PhysRevB.107.L041109},
  url = {https://link.aps.org/doi/10.1103/PhysRevB.107.L041109}
}

@article{Cochran_2025,
	author = {Cochran, T. A. and Jobst, B. and Rosenberg, E. and Lensky, Y. D. and Gyawali, G. and Eassa, N. and Will, M. and Szasz, A. and Abanin, D. and Acharya, R. and Aghababaie Beni, L. and Andersen, T. I. and Ansmann, M. and Arute, F. and Arya, K. and Asfaw, A. and Atalaya, J. and Babbush, R. and Ballard, B. and Bardin, J. C. and Bengtsson, A. and Bilmes, A. and Bourassa, A. and Bovaird, J. and Broughton, M. and Browne, D. A. and Buchea, B. and Buckley, B. B. and Burger, T. and Burkett, B. and Bushnell, N. and Cabrera, A. and Campero, J. and Chang, H. -S. and Chen, Z. and Chiaro, B. and Claes, J. and Cleland, A. Y. and Cogan, J. and Collins, R. and Conner, P. and Courtney, W. and Crook, A. L. and Curtin, B. and Das, S. and Demura, S. and De Lorenzo, L. and Di Paolo, A. and Donohoe, P. and Drozdov, I. and Dunsworth, A. and Eickbusch, A. and Elbag, A. Moshe and Elzouka, M. and Erickson, C. and Ferreira, V. S. and Burgos, L. Flores and Forati, E. and Fowler, A. G. and Foxen, B. and Ganjam, S. and Gasca, R. and Genois, {\'E}. and Giang, W. and Gilboa, D. and Gosula, R. and Grajales Dau, A. and Graumann, D. and Greene, A. and Gross, J. A. and Habegger, S. and Hansen, M. and Harrigan, M. P. and Harrington, S. D. and Heu, P. and Higgott, O. and Hilton, J. and Huang, H. -Y. and Huff, A. and Huggins, W. and Jeffrey, E. and Jiang, Z. and Jones, C. and Joshi, C. and Juhas, P. and Kafri, D. and Kang, H. and Karamlou, A. H. and Kechedzhi, K. and Khaire, T. and Khattar, T. and Khezri, M. and Kim, S. and Klimov, P. and Kobrin, B. and Korotkov, A. and Kostritsa, F. and Kreikebaum, J. and Kurilovich, V. and Landhuis, D. and Lange-Dei, T. and Langley, B. and Lau, K. -M. and Ledford, J. and Lee, K. and Lester, B. and Le Guevel, L. and Li, W. and Lill, A. T. and Livingston, W. and Locharla, A. and Lundahl, D. and Lunt, A. and Madhuk, S. and Maloney, A. and Mandr{\`a}, S. and Martin, L. and Martin, O. and Maxfield, C. and McClean, J. and McEwen, M. and Meeks, S. and Megrant, A. and Miao, K. and Molavi, R. and Molina, S. and Montazeri, S. and Movassagh, R. and Neill, C. and Newman, M. and Nguyen, A. and Nguyen, M. and Ni, C. -H. and Ottosson, K. and Pizzuto, A. and Potter, R. and Pritchard, O. and Quintana, C. and Ramachandran, G. and Reagor, M. and Rhodes, D. and Roberts, G. and Sankaragomathi, K. and Satzinger, K. and Schurkus, H. and Shearn, M. and Shorter, A. and Shutty, N. and Shvarts, V. and Sivak, V. and Small, S. and Smith, W. C. and Springer, S. and Sterling, G. and Suchard, J. and Sztein, A. and Thor, D. and Torunbalci, M. and Vaishnav, A. and Vargas, J. and Vdovichev, S. and Vidal, G. and Vollgraff Heidweiller, C. and Waltman, S. and Wang, S. X. and Ware, B. and White, T. and Wong, K. and Woo, B. W. K. and Xing, C. and Yao, Z. Jamie and Yeh, P. and Ying, B. and Yoo, J. and Yosri, N. and Young, G. and Zalcman, A. and Zhang, Y. and Zhu, N. and Zobrist, N. and Boixo, S. and Kelly, J. and Lucero, E. and Chen, Y. and Smelyanskiy, V. and Neven, H. and Gammon-Smith, A. and Pollmann, F. and Knap, M. and Roushan, P.},
	da = {2025/06/01},
	doi = {10.1038/s41586-025-08999-9},
	id = {Cochran2025},
	isbn = {1476-4687},
	journal = {Nature},
	number = {8067},
	pages = {315--320},
	title = {Visualizing dynamics of charges and strings in (2 + 1)D lattice gauge theories},
	ty = {JOUR},
	url = {https://doi.org/10.1038/s41586-025-08999-9},
	volume = {642},
	year = {2025}}

@article{PhysRevB.106.014306,
  title = {Entanglement spectrum and quantum phase diagram of the long-range XXZ chain},
  author = {Schneider, J. T. and Thomson, S. J. and Sanchez-Palencia, L.},
  journal = {Phys. Rev. B},
  volume = {106},
  issue = {1},
  pages = {014306},
  numpages = {14},
  year = {2022},
  month = {Jul},
  publisher = {American Physical Society},
  doi = {10.1103/PhysRevB.106.014306},
  url = {https://link.aps.org/doi/10.1103/PhysRevB.106.014306}
}

@article{PhysRevLett.109.237208,
  title = {Entanglement Spectrum, Critical Exponents, and Order Parameters in Quantum Spin Chains},
  author = {De Chiara, G. and Lepori, L. and Lewenstein, M. and Sanpera, A.},
  journal = {Phys. Rev. Lett.},
  volume = {109},
  issue = {23},
  pages = {237208},
  numpages = {5},
  year = {2012},
  month = {Dec},
  publisher = {American Physical Society},
  doi = {10.1103/PhysRevLett.109.237208},
  url = {https://link.aps.org/doi/10.1103/PhysRevLett.109.237208}
}

@article{PhysRevB.83.245134,
  title = {Entanglement spectrum and boundary theories with projected entangled-pair states},
  author = {Cirac, J. Ignacio and Poilblanc, Didier and Schuch, Norbert and Verstraete, Frank},
  journal = {Phys. Rev. B},
  volume = {83},
  issue = {24},
  pages = {245134},
  numpages = {12},
  year = {2011},
  month = {Jun},
  publisher = {American Physical Society},
  doi = {10.1103/PhysRevB.83.245134},
  url = {https://link.aps.org/doi/10.1103/PhysRevB.83.245134}
}

@article{PhysRevA.78.032329,
  title = {Entanglement spectrum in one-dimensional systems},
  author = {Calabrese, Pasquale and Lefevre, Alexandre},
  journal = {Phys. Rev. A},
  volume = {78},
  issue = {3},
  pages = {032329},
  numpages = {4},
  year = {2008},
  month = {Sep},
  publisher = {American Physical Society},
  doi = {10.1103/PhysRevA.78.032329},
  url = {https://link.aps.org/doi/10.1103/PhysRevA.78.032329}
}

@article{PhysRevLett.108.227201,
  title = {Boundary-Locality and Perturbative Structure of Entanglement Spectra in Gapped Systems},
  author = {Alba, Vincenzo and Haque, Masudul and L\"auchli, Andreas M.},
  journal = {Phys. Rev. Lett.},
  volume = {108},
  issue = {22},
  pages = {227201},
  numpages = {5},
  year = {2012},
  month = {May},
  publisher = {American Physical Society},
  doi = {10.1103/PhysRevLett.108.227201},
  url = {https://link.aps.org/doi/10.1103/PhysRevLett.108.227201}
}

@article{Fontana_2025,
	author = {Fontana, Enrico and Rudolph, Manuel S. and Duncan, Ross and Rungger, Ivan and C{\^\i}rstoiu, Cristina},
	da = {2025/05/22},
	doi = {10.1038/s41534-024-00955-1},
	id = {Fontana2025},
	isbn = {2056-6387},
	journal = {npj Quantum Information},
	number = {1},
	pages = {84},
	title = {Classical simulations of noisy variational quantum circuits},
	ty = {JOUR},
	url = {https://doi.org/10.1038/s41534-024-00955-1},
	volume = {11},
	year = {2025}}

@misc{rudolph2025paulipropagationcomputationalframework,
      title={Pauli Propagation: A Computational Framework for Simulating Quantum Systems}, 
      author={Manuel S. Rudolph and Tyson Jones and Yanting Teng and Armando Angrisani and Zoë Holmes},
      year={2025},
      eprint={2505.21606},
      archivePrefix={arXiv},
      primaryClass={quant-ph},
      url={https://arxiv.org/abs/2505.21606}, 
}

@misc{lerch2024efficientquantumenhancedclassicalsimulation,
      title={Efficient quantum-enhanced classical simulation for patches of quantum landscapes}, 
      author={Sacha Lerch and Ricard Puig and Manuel S. Rudolph and Armando Angrisani and Tyson Jones and M. Cerezo and Supanut Thanasilp and Zoë Holmes},
      year={2024},
      eprint={2411.19896},
      archivePrefix={arXiv},
      primaryClass={quant-ph},
      url={https://arxiv.org/abs/2411.19896}, 
}

@misc{broers2024exclusiveorencodedalgebraicstructure,
      title={Exclusive-or encoded algebraic structure for efficient quantum dynamics}, 
      author={Lukas Broers and Ludwig Mathey},
      year={2024},
      eprint={2404.09312},
      archivePrefix={arXiv},
      primaryClass={cond-mat.other},
      url={https://arxiv.org/abs/2404.09312}, 
}

@article{doi:10.1126/science.ado6285,
	author = {Andrew D. King and Alberto Nocera and Marek M. Rams and Jacek Dziarmaga and Roeland Wiersema and William Bernoudy and Jack Raymond and Nitin Kaushal and Niclas Heinsdorf and Richard Harris and Kelly Boothby and Fabio Altomare and Mohsen Asad and Andrew J. Berkley and Martin Boschnak and Kevin Chern and Holly Christiani and Samantha Cibere and Jake Connor and Martin H. Dehn and Rahul Deshpande and Sara Ejtemaee and Pau Farre and Kelsey Hamer and Emile Hoskinson and Shuiyuan Huang and Mark W. Johnson and Samuel Kortas and Eric Ladizinsky and Trevor Lanting and Tony Lai and Ryan Li and Allison J. R. MacDonald and Gaelen Marsden and Catherine C. McGeoch and Reza Molavi and Travis Oh and Richard Neufeld and Mana Norouzpour and Joel Pasvolsky and Patrick Poitras and Gabriel Poulin-Lamarre and Thomas Prescott and Mauricio Reis and Chris Rich and Mohammad Samani and Benjamin Sheldan and Anatoly Smirnov and Edward Sterpka and Berta Trullas Clavera and Nicholas Tsai and Mark Volkmann and Alexander M. Whiticar and Jed D. Whittaker and Warren Wilkinson and Jason Yao and T.J. Yi and Anders W. Sandvik and Gonzalo Alvarez and Roger G. Melko and Juan Carrasquilla and Marcel Franz and Mohammad H. Amin},
	doi = {10.1126/science.ado6285},
	eprint = {https://www.science.org/doi/pdf/10.1126/science.ado6285},
	journal = {Science},
	number = {0},
	pages = {eado6285},
	title = {Beyond-classical computation in quantum simulation},
	url = {https://www.science.org/doi/abs/10.1126/science.ado6285},
	volume = {0},
        year = {2025}}

@article{Abbas_2024,
	author = {Abbas, Amira and Ambainis, Andris and Augustino, Brandon and B{\"a}rtschi, Andreas and Buhrman, Harry and Coffrin, Carleton and Cortiana, Giorgio and Dunjko, Vedran and Egger, Daniel J. and Elmegreen, Bruce G. and Franco, Nicola and Fratini, Filippo and Fuller, Bryce and Gacon, Julien and Gonciulea, Constantin and Gribling, Sander and Gupta, Swati and Hadfield, Stuart and Heese, Raoul and Kircher, Gerhard and Kleinert, Thomas and Koch, Thorsten and Korpas, Georgios and Lenk, Steve and Marecek, Jakub and Markov, Vanio and Mazzola, Guglielmo and Mensa, Stefano and Mohseni, Naeimeh and Nannicini, Giacomo and O'Meara, Corey and Tapia, Elena Pe{\~n}a and Pokutta, Sebastian and Proissl, Manuel and Rebentrost, Patrick and Sahin, Emre and Symons, Benjamin C. B. and Tornow, Sabine and Valls, V{\'\i}ctor and Woerner, Stefan and Wolf-Bauwens, Mira L. and Yard, Jon and Yarkoni, Sheir and Zechiel, Dirk and Zhuk, Sergiy and Zoufal, Christa},
	da = {2024/12/01},
	doi = {10.1038/s42254-024-00770-9},
	id = {Abbas2024},
	isbn = {2522-5820},
	journal = {Nature Reviews Physics},
	number = {12},
	pages = {718--735},
	title = {Challenges and opportunities in quantum optimization},
	ty = {JOUR},
	url = {https://doi.org/10.1038/s42254-024-00770-9},
	volume = {6},
	year = {2024}}

@article{PhysRevLett.128.010607,
  title = {Sequential Generation of Projected Entangled-Pair States},
  author = {Wei, Zhi-Yuan and Malz, Daniel and Cirac, J. Ignacio},
  journal = {Phys. Rev. Lett.},
  volume = {128},
  issue = {1},
  pages = {010607},
  numpages = {6},
  year = {2022},
  month = {Jan},
  publisher = {American Physical Society},
  doi = {10.1103/PhysRevLett.128.010607},
  url = {https://link.aps.org/doi/10.1103/PhysRevLett.128.010607}
}

@inproceedings{10.1145/3564246.3585234,
author = {Aharonov, Dorit and Gao, Xun and Landau, Zeph and Liu, Yunchao and Vazirani, Umesh},
title = {A Polynomial-Time Classical Algorithm for Noisy Random Circuit Sampling},
year = {2023},
isbn = {9781450399135},
publisher = {Association for Computing Machinery},
address = {New York, NY, USA},
url = {https://doi.org/10.1145/3564246.3585234},
doi = {10.1145/3564246.3585234},
booktitle = {Proceedings of the 55th Annual ACM Symposium on Theory of Computing},
pages = {945–957},
numpages = {13},
location = {Orlando, FL, USA},
series = {STOC 2023}
}

@article{Di_Meglio_2024,
   title={Quantum Computing for High-Energy Physics: State of the Art and Challenges},
   volume={5},
   ISSN={2691-3399},
   url={http://dx.doi.org/10.1103/PRXQuantum.5.037001},
   DOI={10.1103/prxquantum.5.037001},
   number={3},
   journal={PRX Quantum},
   publisher={American Physical Society (APS)},
   author={Di Meglio, Alberto and Jansen, Karl and Tavernelli, Ivano and Alexandrou, Constantia and Arunachalam, Srinivasan and Bauer, Christian W. and Borras, Kerstin and Carrazza, Stefano and Crippa, Arianna and Croft, Vincent and de Putter, Roland and Delgado, Andrea and Dunjko, Vedran and Egger, Daniel J. and Fernández-Combarro, Elias and Fuchs, Elina and Funcke, Lena and González-Cuadra, Daniel and Grossi, Michele and Halimeh, Jad C. and Holmes, Zoë and Kühn, Stefan and Lacroix, Denis and Lewis, Randy and Lucchesi, Donatella and Martinez, Miriam Lucio and Meloni, Federico and Mezzacapo, Antonio and Montangero, Simone and Nagano, Lento and Pascuzzi, Vincent R. and Radescu, Voica and Ortega, Enrique Rico and Roggero, Alessandro and Schuhmacher, Julian and Seixas, Joao and Silvi, Pietro and Spentzouris, Panagiotis and Tacchino, Francesco and Temme, Kristan and Terashi, Koji and Tura, Jordi and Tüysüz, Cenk and Vallecorsa, Sofia and Wiese, Uwe-Jens and Yoo, Shinjae and Zhang, Jinglei},
   year={2024},
   month=aug }

@article{PhysRevB.111.045139,
  title = {Yang-Lee zeros in quantum phase transitions: An entanglement perspective},
  author = {Li, Hongchao},
  journal = {Phys. Rev. B},
  volume = {111},
  issue = {4},
  pages = {045139},
  numpages = {12},
  year = {2025},
  month = {Jan},
  publisher = {American Physical Society},
  doi = {10.1103/PhysRevB.111.045139},
  url = {https://link.aps.org/doi/10.1103/PhysRevB.111.045139}
}

@article{PhysRev.87.410,
  title = {Statistical Theory of Equations of State and Phase Transitions. II. Lattice Gas and Ising Model},
  author = {Lee, T. D. and Yang, C. N.},
  journal = {Phys. Rev.},
  volume = {87},
  issue = {3},
  pages = {410--419},
  numpages = {0},
  year = {1952},
  month = {Aug},
  publisher = {American Physical Society},
  doi = {10.1103/PhysRev.87.410},
  url = {https://link.aps.org/doi/10.1103/PhysRev.87.410}
}

@article{PhysRev.87.404,
  title = {Statistical Theory of Equations of State and Phase Transitions. I. Theory of Condensation},
  author = {Yang, C. N. and Lee, T. D.},
  journal = {Phys. Rev.},
  volume = {87},
  issue = {3},
  pages = {404--409},
  numpages = {0},
  year = {1952},
  month = {Aug},
  publisher = {American Physical Society},
  doi = {10.1103/PhysRev.87.404},
  url = {https://link.aps.org/doi/10.1103/PhysRev.87.404}
}

@article{PhysRevLett.133.120603,
  title = {Simulating Noisy Variational Quantum Algorithms: A Polynomial Approach},
  author = {Shao, Yuguo and Wei, Fuchuan and Cheng, Song and Liu, Zhengwei},
  journal = {Phys. Rev. Lett.},
  volume = {133},
  issue = {12},
  pages = {120603},
  numpages = {7},
  year = {2024},
  month = {Sep},
  publisher = {American Physical Society},
  doi = {10.1103/PhysRevLett.133.120603},
  url = {https://link.aps.org/doi/10.1103/PhysRevLett.133.120603}
}

@article{PhysRevA.111.032612,
  title = {Imaginary Hamiltonian variational Ansatz for combinatorial optimization problems},
  author = {Wang, Xiaoyang and Chai, Yahui and Feng, Xu and Guo, Yibin and Jansen, Karl and T\"uys\"uz, Cenk},
  journal = {Phys. Rev. A},
  volume = {111},
  issue = {3},
  pages = {032612},
  numpages = {20},
  year = {2025},
  month = {Mar},
  publisher = {American Physical Society},
  doi = {10.1103/PhysRevA.111.032612},
  url = {https://link.aps.org/doi/10.1103/PhysRevA.111.032612}
}

@article{Cerezo_vqa,
	author = {Cerezo, M. and Arrasmith, Andrew and Babbush, Ryan and Benjamin, Simon C. and Endo, Suguru and Fujii, Keisuke and McClean, Jarrod R. and Mitarai, Kosuke and Yuan, Xiao and Cincio, Lukasz and Coles, Patrick J.},
	da = {2021/09/01},
	doi = {10.1038/s42254-021-00348-9},
	id = {Cerezo2021},
	isbn = {2522-5820},
	journal = {Nature Reviews Physics},
	number = {9},
	pages = {625--644},
	title = {Variational quantum algorithms},
	ty = {JOUR},
	url = {https://doi.org/10.1038/s42254-021-00348-9},
	volume = {3},
	year = {2021}}

@article{PhysRevA.103.042612,
  title = {MaxCut quantum approximate optimization algorithm performance guarantees for $p>1$},
  author = {Wurtz, Jonathan and Love, Peter},
  journal = {Phys. Rev. A},
  volume = {103},
  issue = {4},
  pages = {042612},
  numpages = {15},
  year = {2021},
  month = {Apr},
  publisher = {American Physical Society},
  doi = {10.1103/PhysRevA.103.042612},
  url = {https://link.aps.org/doi/10.1103/PhysRevA.103.042612}
}

@article{Bravyi_2020,
  title = {Obstacles to Variational Quantum Optimization from Symmetry Protection},
  author = {Bravyi, S. and Kliesch, A. and Koenig, R. and Tang, E.},
  journal = {Phys. Rev. Lett.},
  volume = {125},
  issue = {26},
  pages = {260505},
  numpages = {6},
  year = {2020},
  month = {Dec},
  publisher = {American Physical Society},
  doi = {10.1103/PhysRevLett.125.260505},
  url = {https://link.aps.org/doi/10.1103/PhysRevLett.125.260505}
}

@article{TILLY2022,
	author = {Jules Tilly and Hongxiang Chen and Shuxiang Cao and Dario Picozzi and Kanav Setia and Ying Li and Edward Grant and Leonard Wossnig and Ivan Rungger and George H. Booth and Jonathan Tennyson},
	journal = {Physics Reports},
	pages = {1-128},
	title = {The Variational Quantum Eigensolver: A review of methods and best practices},
	volume = {986},
	year = {2022}
}

@article{Kandala2017,
	author = {Kandala, A. and Mezzacapo, A. and Temme, K. and Takita, M. and Brink, M. and Chow, Jerry M. and Gambetta, J. M.},
	da = {2017/09/01},
	doi = {10.1038/nature23879},
	id = {Kandala2017},
	isbn = {1476-4687},
	journal = {Nature},
	number = {7671},
	pages = {242--246},
	title = {Hardware-efficient variational quantum eigensolver for small molecules and quantum magnets},
	ty = {JOUR},
	url = {https://doi.org/10.1038/nature23879},
	volume = {549},
	year = {2017}}

@article{Chen_2024,
  title = {Sequential quantum circuits as maps between gapped phases},
  author = {Chen, Xie and Dua, Arpit and Hermele, Michael and Stephen, David T. and Tantivasadakarn, Nathanan and Vanhove, Robijn and Zhao, Jing-Yu},
  journal = {Phys. Rev. B},
  volume = {109},
  issue = {7},
  pages = {075116},
  numpages = {21},
  year = {2024},
  month = {Feb},
  publisher = {American Physical Society},
  doi = {10.1103/PhysRevB.109.075116},
  url = {https://link.aps.org/doi/10.1103/PhysRevB.109.075116}
}

@article{Wecker_2015,
  title = {Solving strongly correlated electron models on a quantum computer},
  author = {Wecker, D. and Hastings, M. B. and Wiebe, N. and Clark, B. K. and Nayak, C. and Troyer, M.},
  journal = {Phys. Rev. A},
  volume = {92},
  issue = {6},
  pages = {062318},
  numpages = {24},
  year = {2015},
  month = {Dec},
  publisher = {American Physical Society},
  doi = {10.1103/PhysRevA.92.062318},
  url = {https://link.aps.org/doi/10.1103/PhysRevA.92.062318}
}

@article{Wang23,
  title = {Critical behavior of the Ising model by preparing the thermal state on a quantum computer},
  author = {Wang, Xiaoyang and Feng, Xu and Hartung, Tobias and Jansen, Karl and Stornati, Paolo},
  journal = {Phys. Rev. A},
  volume = {108},
  issue = {2},
  pages = {022612},
  numpages = {12},
  year = {2023},
  month = {Aug},
  publisher = {American Physical Society},
  doi = {10.1103/PhysRevA.108.022612},
  url = {https://link.aps.org/doi/10.1103/PhysRevA.108.022612}
}

@article{WOLFF199093,
title = {Critical slowing down},
journal = {Nuclear Physics B - Proceedings Supplements},
volume = {17},
pages = {93-102},
year = {1990},
issn = {0920-5632},
doi = {https://doi.org/10.1016/0920-5632(90)90224-I},
url = {https://www.sciencedirect.com/science/article/pii/092056329090224I},
author = {Ulli Wolff}
}

@article{Peruzzo:2014,
  doi = {10.1038/ncomms5213},
  url = {https://doi.org/10.1038/ncomms5213},
  title = {A variational eigenvalue solver on a photonic quantum processor},
  author = {Peruzzo, A. and McClean, J. and Shadbolt, P. and Yung, M. and Zhou, X. and Love, P. J. and Aspuru-Guzik, A. and O’Brien, J. L.},
  journal = {{Nat. Comm.}},
  volume = {5},
  pages = {4213},
  year = {2014}
}

@misc{lin2022lecturenotesquantumalgorithms,
      title={Lecture Notes on Quantum Algorithms for Scientific Computation}, 
      author={Lin Lin},
      year={2022},
      eprint={2201.08309},
      archivePrefix={arXiv},
      primaryClass={quant-ph},
        page={18},
      url={https://arxiv.org/abs/2201.08309}, 
}

@article{McArdle_20,
  title = {Quantum computational chemistry},
  author = {McArdle, S. and Endo, S. and Aspuru-Guzik, A. and Benjamin, S. C. and Yuan, X.},
  journal = {Rev. Mod. Phys.},
  volume = {92},
  issue = {1},
  pages = {015003},
  numpages = {51},
  year = {2020},
  month = {Mar},
  publisher = {American Physical Society},
  doi = {10.1103/RevModPhys.92.015003},
  url = {https://link.aps.org/doi/10.1103/RevModPhys.92.015003}
}

@article{Angrisani2025,
  title = {Classically Estimating Observables of Noiseless Quantum Circuits},
  author = {Angrisani, Armando and Schmidhuber, Alexander and Rudolph, Manuel S. and Cerezo, M. and Holmes, Zo\"e and Huang, Hsin-Yuan},
  journal = {Phys. Rev. Lett.},
  volume = {135},
  issue = {17},
  pages = {170602},
  numpages = {10},
  year = {2025},
  month = {Oct},
  publisher = {American Physical Society},
  doi = {10.1103/lh6x-7rc3},
  url = {https://link.aps.org/doi/10.1103/lh6x-7rc3}
}

@article{Cerezo2025,
	author = {Cerezo, M. and Larocca, Martin and Garc{\'\i}a-Mart{\'\i}n, Diego and Diaz, N. L. and Braccia, Paolo and Fontana, Enrico and Rudolph, Manuel S. and Bermejo, Pablo and Ijaz, Aroosa and Thanasilp, Supanut and Anschuetz, Eric R. and Holmes, Zo{\"e}},
	da = {2025/08/25},
	doi = {10.1038/s41467-025-63099-6},
	id = {Cerezo2025},
	isbn = {2041-1723},
	journal = {Nature Communications},
	number = {1},
	pages = {7907},
	title = {Does provable absence of barren plateaus imply classical simulability?},
	ty = {JOUR},
	url = {https://doi.org/10.1038/s41467-025-63099-6},
	volume = {16},
	year = {2025}}

@article{farhi2014quantum,
  title={A Quantum Approximate Optimization Algorithm}, 
  author={Farhi, E. and Goldstone, J. and Gutmann, S.},
   journal={arXiv preprint arXiv:1411.4028},
   url={https://arxiv.org/abs/1411.4028},
   year={2014}
}

@article{Wu_2019,
  title = {Variational Thermal Quantum Simulation via Thermofield Double States},
  author = {Wu, Jingxiang and Hsieh, Timothy H.},
  journal = {Phys. Rev. Lett.},
  volume = {123},
  issue = {22},
  pages = {220502},
  numpages = {6},
  year = {2019},
  month = {Nov},
  publisher = {American Physical Society},
  doi = {10.1103/PhysRevLett.123.220502},
  url = {https://link.aps.org/doi/10.1103/PhysRevLett.123.220502}
}

@article{zhang_absence_2024,
	title = {Absence of {Barren} {Plateaus} in {Finite} {Local}-{Depth} {Circuits} with {Long}-{Range} {Entanglement}},
	volume = {132},
	url = {https://link.aps.org/doi/10.1103/PhysRevLett.132.150603},
	doi = {10.1103/PhysRevLett.132.150603},
	number = {15},
	urldate = {2024-05-05},
	journal = {Physical Review Letters},
	author = {Zhang, Hao-Kai and Liu, Shuo and Zhang, Shi-Xin},
	month = apr,
	year = {2024},
	note = {Publisher: American Physical Society},
	pages = {150603},
}

@article{ITensor,
	title={{The ITensor Software Library for Tensor Network Calculations}},
	author={Matthew Fishman and Steven R. White and E. Miles Stoudenmire},
	journal={SciPost Phys. Codebases},
	pages={4},
	year={2022},
	publisher={SciPost},
	doi={10.21468/SciPostPhysCodeb.4},
	url={https://scipost.org/10.21468/SciPostPhysCodeb.4},
}

@misc{orqa,
      title={Scalable Simulation of Quantum Many-Body Dynamics with Or-Represented Quantum Algebra}, 
      author={Lukas Broers and Rong-Yang Sun and Seiji Yunoki},
      year={2025},
      eprint={2506.13241},
      archivePrefix={arXiv},
      primaryClass={quant-ph},
      url={https://arxiv.org/abs/2506.13241}, 
}

@book{huang2008statistical,
  title={Statistical mechanics},
  author={Huang, Kerson},
  year={2008},
  publisher={John Wiley \& Sons}
}

@book{brankov2000theory,
  title={Theory of critical phenomena in finite-size systems: scaling and quantum effects},
  author={Brankov, Jordan G and Danchev, Daniel M and Tonchev, Nicholai S},
  year={2000},
  publisher={World Scientific}
}

@book{Sachdev2011,
  place={Cambridge}, 
  edition={2}, 
  title={Quantum Phase Transitions}, 
  publisher={Cambridge University Press}, 
  author={Sachdev, Subir}, 
  year={2011}
}

@book{negele2018quantum,
  title={Quantum many-particle systems},
  author={Negele, John W and Orland, Henri},
  year={2018},
  publisher={CRC Press}
}

@article{Lepori2013,
  title = {Scaling of the entanglement spectrum near quantum phase transitions},
  author = {Lepori, L. and De Chiara, G. and Sanpera, A.},
  journal = {Phys. Rev. B},
  volume = {87},
  issue = {23},
  pages = {235107},
  numpages = {10},
  year = {2013},
  month = {Jun},
  publisher = {American Physical Society},
  doi = {10.1103/PhysRevB.87.235107},
  url = {https://link.aps.org/doi/10.1103/PhysRevB.87.235107}
}
\bibliographystyle{IEEEtran} 

\appendix
\begin{figure*}
    \centering
    \includegraphics[width=0.95\textwidth]{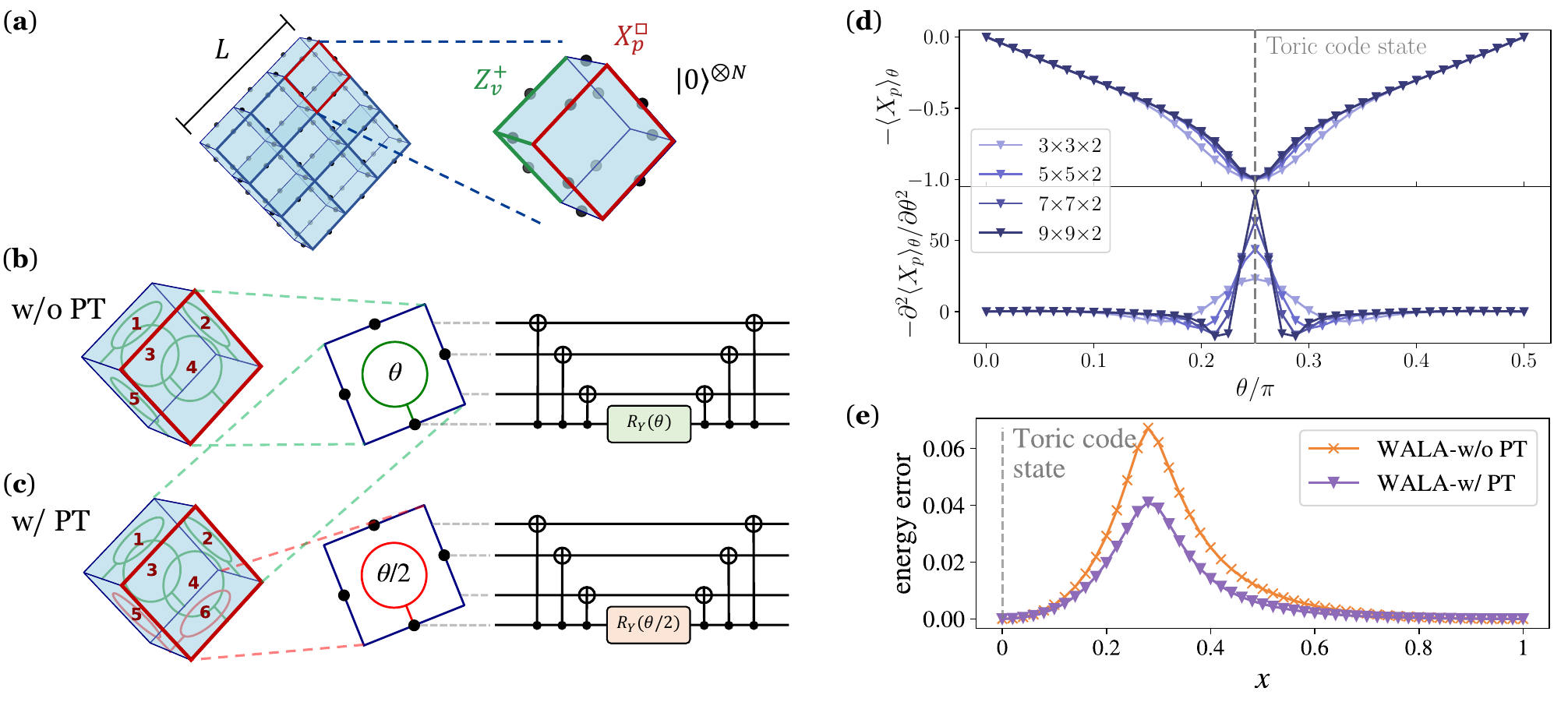}
    \caption{(\textbf{a}) Illustration of an $L\times L\times 2$ lattice of the $\ZZ_2$ gauge theory, with the qubits denoted by black dots initialized to the state $\ket{0}^{\otimes N}$. (\textbf{b},\textbf{c}) Weight-adjustable loop ansatz (WALA) on a cube without and with phase transition (PT). The plaquette rotation $e^{-\im\theta XXXY}$ is applied on each surface of the cube sequentially by the order $\mathbf{1}$-$\mathbf{5}$ and $\mathbf{1}$-$\mathbf{6}$ in (\textbf{b}) and (\textbf{c}), respectively. (\textbf{d}) Plaquette expectation and its 2nd derivative as a function of the WALA angle $\theta$. $\theta=\pi/4$ corresponds to the exact toric code state, which is also the ground state of $H_{\ZZ_2}(0)$. (\textbf{e}) Relative error of the ground state energy $\left|(E_{\text{VQE}}-E_{\text{ED}}\right)/E_{\text{ED}}|$ obtained by VQE and exact diagonalization (ED) as a function of the coupling strength $x$ in $H_{\ZZ_2}(x)$, using WALA with and without phase transition as heuristic ans\"atze.}
    \label{fig:toric}
\end{figure*}

\section*{End Note}

\noindent \textit{Appendix A: Phase transition in the weight-adjustable loop ansatz}~---~The weight-adjustable loop ansatz (WALA) is proposed to prepare the ground state of $\mathbb{Z}_2$ gauge theory with the Hamiltonian
\begin{align}
    H_{\ZZ_2}(x)\equiv -(1-x)\sum_{\text{plaquette}} X_{p}^{\square}-x\sum_{\text{links}} Z_l-\sum_{\text{vertex}} Z_{v}^{+},\nonumber
\end{align}
where the plaquette operators $X_p^{\square} = \prod_{i\in p}X_i$ are products of Pauli-$X$ operators on links of a plaquette, the vertex operators $Z_{v}^{+}=\prod_{i\in v}Z_i$ are products of Pauli-$Z$ operators emanating from a vertex $v$. $x\in[0,1]$ tunes the strength of an external magnetic field along the $z$ direction. The ground state of the two limiting cases $H_{\ZZ_2}(0)$ and $H_{\ZZ_2}(1)$ are the toric code state~\cite{KITAEV20032} and the product state $\ket{0}^{\otimes N}$, respectively.

Consider WALA on a three-dimensional $L\times L\times 2$ lattice consisting of $(L-1)^2$ cubes. For one cube shown in the right panel of Fig.~\ref{fig:toric}(\textbf{a}), each link is associated with one qubit initialized in the $\ket{0}$ state. For each plaquette of the cube, a parametrized four-qubit gate $e^{-i\theta XXXY}$ or $e^{-i\theta XXXY/2}$ is applied, denoted respectively by an oriented green or red circle in Fig.~\ref{fig:toric}(\textbf{b}) and (\textbf{c}). These four-qubit gates are applied sequentially following the order $\mathbf{1}$-$\mathbf{5}$ and $\mathbf{1}$-$\mathbf{6}$ for the two WALAs in Fig.~\ref{fig:toric}(\textbf{b}) and (\textbf{c}), labeled by without (w/o) and with (w/) phase transition (PT) respectively. These two WALAs are different on the bottom two plaquettes of each cube, but both exactly prepare the ground state of $H_{\ZZ_2}(0)$ with $\theta=\pi/4$ and the ground state of $H_{\ZZ_2}(1)$ with $\theta=0$. The WALA in Fig.~\ref{fig:toric}(\textbf{b}) is identical to the original WALA proposed in literature~\cite{PhysRevB.107.L041109,Cochran_2025}, which has no non-analyticity and no phase transition~\cite{supp}.

In contrast, our numerical results suggest that the modified WALA in Fig.~\ref{fig:toric}(\textbf{c}) exhibits non-analyticity and phase transition in the infinite volume limit $L\to \infty$. Fig.~\ref{fig:toric}(\textbf{d}) shows the numerical results of the plaquette expectation $\langle X_p^{\square}\rangle_{\theta}$ of the modified WALA as a function of the WALA angle $\theta$. On lattices from $3\times 3\times 2$ to $9\times 9\times 2$, the numerical results indicate a cusp of the expectation $-\langle X_p^{\square}\rangle_{\theta}$ and a divergent 2nd derivative $-\partial^2\langle X_p^{\square}\rangle_{\theta}/\partial\theta^2$ at $\theta=\pi/4$, indicating that the modified WALA has a phase transition at $\theta_c=\pi/4$.

Compared with the WALA without phase transition, the modified WALA with phase transition as a heuristic VQE ansatz has improved state-preparation accuracy. To see this, on a single cube with $12$ qubits, the WALA angle $\theta$ in both WALA with and without phase transition are optimized to the minimum energy $\langle H_{\ZZ_2}(x)\rangle _{\theta}$ for each $x\in[0,1]$, and the corresponding energy $E_{\text{VQE}}\equiv\langle H_{\ZZ_2}(x)\rangle _{\theta^{*}}$ is obtained. The resulting energy relative error is plotted in Fig.~\ref{fig:toric}(\textbf{e}). We see that WALA with phase transition consistently achieves lower relative error for all $x\in[0,1]$, with the worst-case relative error reduced from 0.067 to 0.041, corresponding to a 38.8\% reduction in error. Similar error reduction is also observed for the Ising whip circuit~\cite{supp}. These numerical results demonstrate the generality of phase transition phenomena in sequential PQCs and the capability of PQCs with phase transition to improve state-preparation accuracy. 
\\\\
\noindent \textit{Appendix B: Symmetry operator}~---~In this section, we prove that the symmetry operator $\hat{T}=\prod_{i\in\BB}Z_i$ performs the transformation $T:\theta\rightarrow \theta+\pi$ to the PQC state $\ket{\phi_{\text{w}}(\theta)}$ given by Eq.~\eqref{eq:symmetry-transformation}. 

First, $T$ transforms $ZY$ rotations in the bulk of the lattice and at the boundary of the lattice as
\begin{equation}
    \begin{aligned}
    e^{-\ii\theta ZY/2}&\rightarrow e^{-\ii(\theta+\pi) ZY/2}=-\ii ZY e^{-\ii\theta ZY/2},\\
    e^{-\ii\theta ZY}&\rightarrow e^{-\ii(\theta+\pi) ZY}=-e^{-\ii\theta ZY}.
    \label{eq:ZYgate-transformation}
\end{aligned}
\end{equation}
We see that only the $ZY$ rotation in the bulk of the lattice gives an additional $ZY$ operator. Then, for each bulk site, the product of these $ZY$ operators leads to an identity
\begin{center}
    \includegraphics[width=0.4\textwidth]{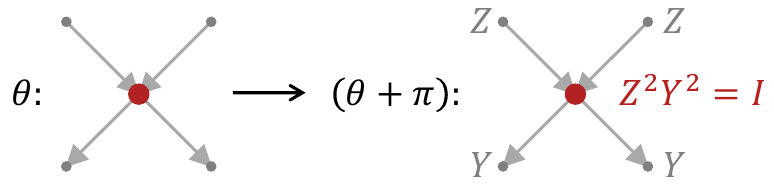}
\end{center}
where each gray arrow denotes a bulk $ZY$ rotation $e^{-\ii \theta ZY/2}$. For the site at the lower boundary $\BB'$ or at the upper boundary $\BB/\BB'$, on the other hand, the product leads to a Pauli-$Z$ operator
\begin{center}
    \includegraphics[width=0.4\textwidth]{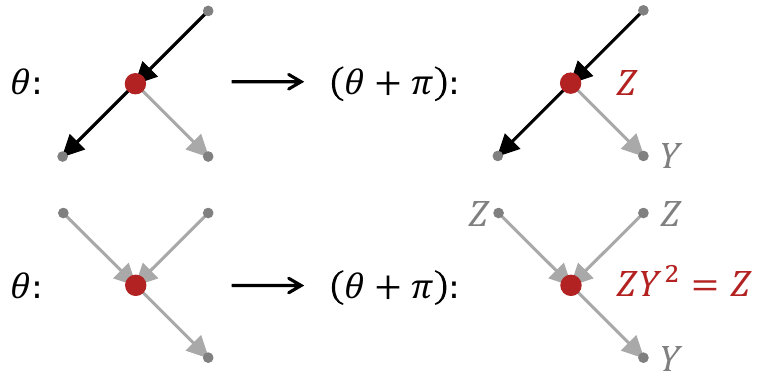}
\end{center}
where each black arrow denotes a boundary $ZY$ rotation $e^{-\ii \theta ZY}$. Therefore, the product of all Pauli-$Z$ operators at the boundary $\BB$ gives the symmetry operator $\hat{T}$. The global phases in Eq.~\eqref{eq:ZYgate-transformation} lead to the global phase $e^{i\phi}$ in Eq.~\eqref{eq:symmetry-transformation}, which is irrelevant and not considered during the derivation above.

\onecolumngrid

\newpage

\

\newpage

\setcounter{page}{1}
\renewcommand{\thefigure}{S\arabic{figure}}
\renewcommand{\theequation}{S\arabic{equation}}


\titleformat{\section}{\normalfont\large\bfseries}{Supplementary Note \thesection:}{1em}{}

\section{Ising whip circuit on arbitrary dimensions}\label{sec:Ising whip circuit on arbitrary dimensions}

\begin{figure}[b]
    \centering
    \includegraphics[width=0.6\textwidth]{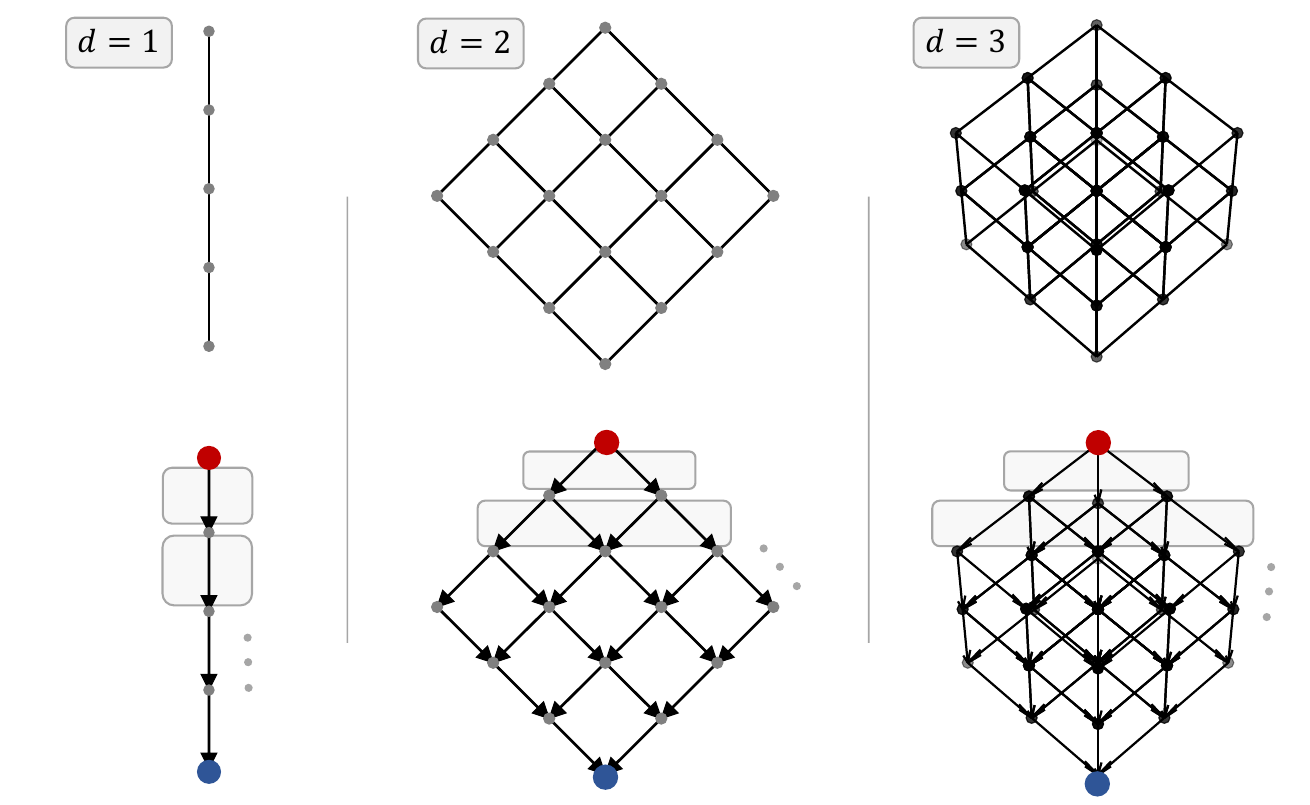}
    \caption{Ising lattices with dimension $d=1,2,3$ and the corresponding directed acyclic graph (DAG). Each DAG has a single source and a single sink labeled by the red and blue dots, respectively. A gray rectangle boxes one layer of the $ZY$ rotations.}
    \label{fig:Ising-lattice}
\end{figure}

In this note, we introduce the whip circuit to prepare the ground state of the Ising model in arbitrary dimensions. The Ising model has the Hamiltonian 
\begin{align}
    H_J= J\sum_{\langle i,j\rangle }Z_iZ_j,
    \label{eq:Ising-Hamiltonian}
\end{align}
where the cases $J=1$ and $J=-1$ correspond to anti-ferromagnetic and ferromagnetic spin systems, respectively. We consider a $d$-dimensional square lattices with open boundary conditions (OBC) and a side length $L$.

In the main text, the whip circuit is constructed for the 2-d Ising model. For an arbitrary $d$-dimensional lattice, the whip circuit is uniquely determined by a directed acyclic graph (DAG) with a single source and a single sink adapted from the $d$-dimensional lattice, as illustrated in Fig.~\ref{fig:Ising-lattice}, where each directed edge denotes a parametrized $ZY$ rotation. The arrow tail and head indicate the qubits that are acted by the $Z$ and $Y$ operators, respectively. Therefore, the Ising whip state reads
\begin{align}
    \ket{\phi_{\text{w},d}(\theta)} =\prod_{\langle i,j\rangle} e^{-\im K_{d,j}\theta Z_i Y_j/2} \ket{+}^{\otimes N},
\end{align}
where $N=L^{d}$ is the total number of qubits, $K_{d,j}\equiv d/\deg^-_j$ is a prefactor depending on $\deg^-_j$~---~the in-degree of the node $j$ in the DAG. The order of applying the $ZY$ rotations, i.e, the order of the product $\prod_{\langle i,j\rangle}$, follows a sequential order from top to bottom of the DAG: all $ZY$ rotations are divided into \textit{layers} boxed by the gray rectangles in Fig.~\ref{fig:Ising-lattice}, and $ZY$ rotations in the same layer can be applied in an arbitrary order, due to the commutation relations $[e^{-\im K_{d,j}\theta Z_i Y_j}, e^{-\im K_{d,j}\theta Z_k Y_j}]=[e^{-\im K_{d,j}\theta Z_i Y_j}, e^{-\im K_k\theta Z_i Y_k}]=0$ for arbitrary $i\neq j\neq k$. The 2-d Ising whip state considered in the main text corresponds to $\ket{\phi_{\text{w}}(\theta)}=\ket{\phi_{\text{w},2}(\theta)}$.

The Ising whip circuit is not limited to the square lattice. As long as the undirected lattice graph is biconnected~---~the graph remains connected if any one node of the graph were to be removed~---~ the DAG with a single source and a single sink can be efficiently obtained, and the whip circuit can be constructed~\cite{wang2025performanceguaranteeslightconevariational}. Thus, the Ising whip circuit can be generalized to triangular lattices, Kagome lattices, and even random lattices used in condensed matter and statistical physics.

The Ising whip circuit exactly prepares the ground state of the Ising Hamiltonian $H_J$. The sequential action of the $ZY$ rotations acting on the initial state $\ket{+}^{\otimes N}$ gives
\begin{equation}
    \begin{aligned}
    \ket{\phi_{\text{w},d}(\theta=\theta_d)} &= \frac{1}{\sqrt{2}}(\ket{0101\cdots 01}+\ket{1010\cdots 10}),\\
    \ket{\phi_{\text{w},d}(\theta=-\theta_d)} &= \frac{1}{\sqrt{2}}(\ket{0000\cdots 00}+\ket{1111\cdots 11})
    \label{eq:ground-states},
\end{aligned}
\end{equation}
where $\theta_d\equiv \pi/(2d)$. These two states are the ground state of $H_J$ with $J=1$ and $J=-1$, respectively. Denoting these two states uniformly as $\ket{\phi_{\text{w},d}(J\theta_d)}$, the ground state energy $E_0$ of these two states reads
\begin{align}
    H_J\ket{\phi_{\text{w},d}(J\theta_d)} = E_0\ket{\phi_{\text{w},d}(J\theta_d)} =-d(L-1)L^{d-1}\ket{\phi_{\text{w},d}(J\theta_d)},
\end{align}
which equals to the number of edges of the $d$-dimensional lattice. In the following two sections, we explore the properties of the Ising whip circuit on one- and two-dimensional lattices.

\section{Whip circuit on a 1-d lattice}

In this section, we study the Ising whip circuit on a 1-d lattice with both open boundary condition (OBC) and periodic boundary condition (PBC), and show that the 1-d circuit has no phase transition according to its Hamiltonian expectation and entanglement entropy in the following two subsections.
\begin{figure}
    \centering
    \includegraphics[width=0.4\textwidth]{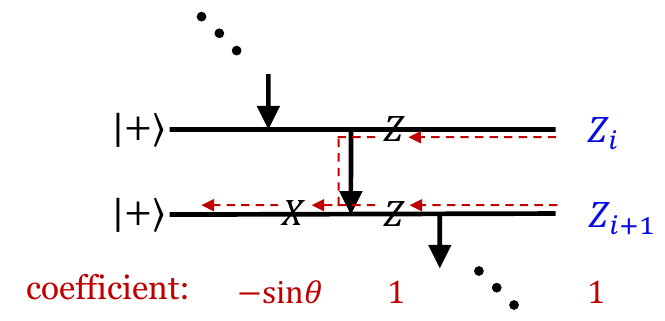}
    \caption{Pauli propagation of the local observable $Z_iZ_{i+1}$ in the whip circuit on the $1$-d chain. Only one Pauli path contributes to the expectation $\bra{\phi_{\text{w},1}(\theta)} Z_iZ_{i+1}\ket{\phi_{\text{w},1}(\theta)}$ given by the Pauli strings $Z_{i}Z_{i+1}$ and $X_{i+1}$. Black arrows denote $ZY$ rotations $e^{-\im \theta ZY/2}$.}
    \label{fig:whip_chain}
\end{figure}

\subsection{Hamiltonian expectation}
\noindent \textbf{Open boundary condition}~---~The Hamiltonian expectation of the Ising model on a $1$-d chain can be evaluated using the language of Pauli propagation. For one $Z_iZ_{i+1}$ term in the Hamiltonian $H_J = J\sum_{i=1}^{L-1}Z_iZ_{i+1}$, only one Pauli path has non-zero contribution to its expectation, as shown in Fig.~\ref{fig:whip_chain}. By the $ZY$ rotation $e^{-\im \theta Z_{i}Y_{i+1}/2}$, $Z_iZ_{i+1}$ is transformed into two terms
\begin{align}
    e^{\im \theta Z_{i}Y_{i+1}/2} Z_iZ_{i+1}e^{-\im \theta Z_{i}Y_{i+1}/2}= \cos\theta Z_iZ_{i+1}-\sin\theta X_{i+1}.
\end{align}
For the first term $Z_iZ_{i+1}$, since no further gates acting on the qubit $i+1$, and the initial $\ket{+}$ state has $\bra{+}Z_{i+1}\ket{+}=0$, the first term has no contribution to the $Z_iZ_{i+1}$ expectation value. On the other hand, the second term with $\bra{+}X_{i+1}\ket{+}=1$ contributes a coefficient $(-\sin\theta)$. Thus, we have 
\begin{align}
    \bra{\phi_{\text{w},1}(\theta)} H_J\ket{\phi_{\text{w},1}(\theta)} = -J(L-1)\sin\theta.
\end{align}
The energy density $\bra{\phi_{\text{w},1}(\theta)} H_J\ket{\phi_{\text{w},1}(\theta)}/(L-1)=-J\sin\theta$ is an analytic function of $\theta$, indicating that the $1$-d whip circuit has no non-analyticity, and thus no phase transition. 

In this example, we see that the $1$-d whip circuit has no non-analyticity if only local Pauli path contributes to the observable expectation. We call the expectation contribution from $\OO(1)$-weight Pauli path as \textit{local contribution}. In contrast, we show in the following content that the non-analyticity comes from Pauli paths whose weights grow with the system size. We call the expectation contribution from these Pauli paths as \textit{non-local contribution}.
\\\\
\noindent \textbf{Periodic boundary condition}~---~Here, we study the Ising whip circuit with PBC on a circular lattice $\CC$ with $L$ nodes and $L$ edges. The Hamiltonian $H_{J}^{\CC} =J \sum_{i=1}^{L}Z_iZ_{i+1}$ has the PBC $Z_{L+1}=Z_1$. The corresponding whip circuit with $L=8$ is illustrated in Fig.~\ref{fig:bipolar-ZY-for-2-regular}. Different from the circuit illustrated in Fig.~2 of the main text, here two kinds of $ZY$ rotations $e^{-\im \theta Z_i Y_j}$ and $e^{-\im \theta Z_i Y_j/2}$ are denoted by black and gray arrows, respectively. For each local observable $Z_iZ_{i+1}$, similar to the case of 1-d chain, it only has local contribution if $i\neq L-1,L$. In contrast, for $i=L-1$ and $i=L$, i.e., $Z_{L-1}Z_L$ and $Z_{L}Z_1$ have not only the local contribution $(-\cos\theta\sin\theta)$, but also the non-local one. The non-local contribution comes from the cycle structure of the whip circuit, as illustrated by the Pauli path in Fig.~\ref{fig:bipolar-ZY-for-2-regular}. Thus, each of $Z_{L-1}Z_L$ and $Z_{L}Z_1$ has a cycle contribution of $(-\sin(2\theta))^{L-2}\sin\theta\cos\theta$. Combining all local contributions, the Hamiltonian expectation reads
\begin{equation}
    \begin{aligned}
    \bra{\phi_{\text{w},\CC}(\theta)}H_{J}^{\CC}\ket{\phi_{\text{w},\CC}(\theta)}
    = -J[\underbrace{(L-2)\sin(2\theta) +2\sin\theta\cos\theta}_{\text{local contribution}}+\underbrace{2(-\sin(2\theta))^{L-2}\sin\theta\cos\theta}_{\text{cycle contribution}})].
    \label{eq:exact-expected-cut-number-2-regular}
\end{aligned}
\end{equation}
We see that the cycle contribution is proportional to $\sin^{L-1}(2\theta)$. In the infinite volume limit $L\to \infty$, this term has a singularity at $\theta=\pi/4+m\pi/2,m\in\mathbb{Z}$. 

\begin{figure}
    \centering
    \includegraphics[width=0.4\textwidth]{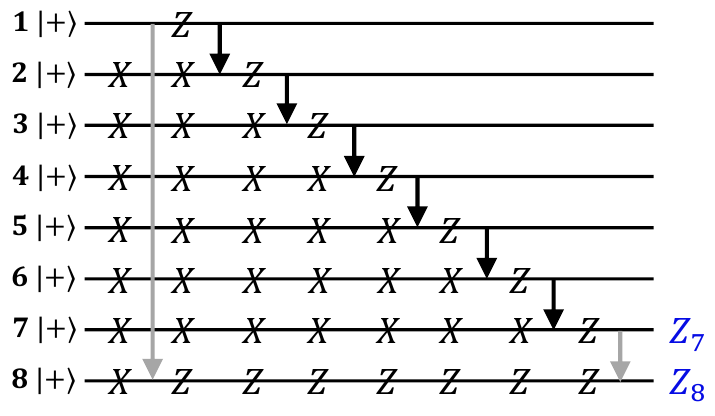}
    \caption{Whip circuit on a cycle with $8$ qubits. Two kinds of $ZY$ rotations $e^{-\im \theta Z_i Y_j}$ and $e^{-\im \theta Z_i Y_j/2}$ are denoted by black and gray arrows, respectively.}
    \label{fig:bipolar-ZY-for-2-regular}
\end{figure}

However, in the energy density $\langle \overline{H_J}\rangle_{\theta}=\bra{\phi_{\text{w},\CC}(\theta)}H_J\ket{\phi_{\text{w},\CC}(\theta)}/L$, the singularity vanishes in the infinite volume limit
\begin{align}
    \lim_{L\to\infty}\left. \sin^{L-1}(2\theta)/L\right|_{\theta=\pi/4}=0.
\end{align}
Thus, the singularity does not lead to non-analyticity in the energy density. The non-local Pauli paths do not lead to non-analyticity, since the Hamiltonian has only $\OO(1)$ non-local terms. As introduced in Sec. II of the main text, the non-analyticity requires \textit{exponentially many} non-local Pauli paths. In contrast, 2-d lattice has numerous cycles such that the non-local Pauli paths are proliferated, leading to a non-analyticity in the energy density.

\subsection{Entanglement entropy}
The claim that the $1$-d whip circuit has no phase transition can further be verified by analytically and numerically calculating its entanglement entropy. Consider the matrix-product-operator (MPO) decomposition of the two-qubit $ZY$ rotation in the computational basis
\begin{equation}
    e^{-\im\theta Z_iY_j/2} = \sum_{x_i,x_i'}\sum_{x_j,x_j'} o_{x_i'x_j'}^{x_ix_j}\ket{x_i'}\bra{x_i}\otimes \ket{x_j'}\bra{x_j},
\end{equation}
where the coefficients $o_{i'j'}^{ij}$ can be factorized as
the following tensor-network notation
\begin{center}
\begin{tikzpicture}
    \SetVertexStyle[FillColor=lightgray,MinSize=0.4\DefaultUnit,FillOpacity=1.0]
\Vertex[y=1,Pseudo]{i_A}
\Vertex[y=-1,Pseudo]{ii_A}
\Vertex[x=0,y=0,color=Cyan]{A}
\Vertex[x=2,y=0,color=Cyan]{B}
\Vertex[x=1,y=0,color=Salmon,shape=diamond]{Link1}
\Vertex[x=2,y=1,Pseudo]{i_B}
\Vertex[x=2,y=-1,Pseudo]{ii_B}
\Edge(A)(i_A)
\Edge(A)(ii_A)
\Edge(A)(Link1)
\Edge(B)(Link1)
\Edge(B)(i_B)
\Edge(B)(ii_B)
\Text[x=3]{$\equiv$}
\Vertex[x=4,y=1,Pseudo]{i_A2}
\Vertex[x=4,y=-1,Pseudo]{ii_A2}
\Vertex[x=4,y=0]{A2}
\Vertex[x=5,y=1,Pseudo]{i_B2}
\Vertex[x=5,y=-1,Pseudo]{ii_B2}
\Vertex[x=5,color=black]{B2}
\Edge(A2)(i_A2)
\Edge(A2)(ii_A2)
\Edge(A2)(B2)
\Edge(B2)(i_B2)
\Edge(B2)(ii_B2)
\end{tikzpicture}\\
\end{center}
Then, according to the Euler identity
\begin{equation}\label{eq:tn-zy-decomp}
    e^{-\im\theta Z_iY_j/2} = \cos(\theta/2)I - \im\sin(\theta/2)Z_iY_j,
\end{equation}
the nonzero tensor elements can be straightforwardly obtained:
\begin{center}
\begin{tikzpicture}
\SetVertexStyle[FillColor=lightgray,MinSize=0.4\DefaultUnit,FillOpacity=1.0]
\Vertex[x=1,y=1,Pseudo]{i_A}
\Vertex[x=1,y=-1,Pseudo]{ii_A}
\Vertex[x=1,y=0]{A}
\Vertex[x=2,y=0,Pseudo]{Link1}
\Edge[label=$i$](A)(i_A)
\Edge[label=$i'$](A)(ii_A)
\Edge[label=$1$](A)(Link1)
\Text[x=1.8,y=0,position=right]{$=\delta_{i'}^{i}{}$,}
\Vertex[x=3,y=1,Pseudo]{i_B}
\Vertex[x=3,y=-1,Pseudo]{ii_B}
\Vertex[x=3,y=0]{B}
\Vertex[x=4,y=0,Pseudo]{Link2}
\Edge[label=$i$](B)(i_B)
\Edge[label=$i'$](B)(ii_B)
\Edge[label=$2$](B)(Link2)
\Text[x=3.8,y=0,position=right]{$={\hat Z}_{i'}^{i}{}$,}
\Vertex[x=6.0,y=1,Pseudo]{i_C}
\Vertex[x=6.0,y=-1,Pseudo]{ii_C}
\Vertex[x=6.0,y=0,color=black]{C}
\Vertex[x=5.0,y=0,Pseudo]{Link3}
\Edge[label=$i$](C)(i_C)
\Edge[label=$i'$](C)(ii_C)
\Edge[label=$1$](C)(Link3)
\Text[x=6.3,y=0,position=right]{$=\delta_{i'}^{i}{\cos(\frac{\theta}{2})}$,}
\Vertex[x=9.5,y=1,Pseudo]{i_D}
\Vertex[x=9.5,y=-1,Pseudo]{ii_D}
\Vertex[x=9.5,y=0,color=black]{D}
\Vertex[x=8.5,y=0,Pseudo]{Link4}
\Edge[label=$i$](D)(i_D)
\Edge[label=$i'$](D)(ii_D)
\Edge[label=$2$](D)(Link4)
\Text[x=9.8,y=0,position=right]{$=-\im{\hat Y}_{i'}^{i}{\sin(\frac{\theta}{2})}$}
\end{tikzpicture}
\end{center}
The initial state $\ket{+}^{\otimes L}$ can also be written as
\begin{equation}
    \ket{+}^{\otimes L} = \frac{1}{\sqrt{2^{L}}}\sum_{x_1,\dots,x_L} \ket{x_1\dots x_L}
    = \sum_{x_1,\dots,x_L} A_{x_1}A_{x_2}\cdots A_{x_L}\ket{x_1}\ket{x_2}\cdots \ket{x_L},
\end{equation}
where $A_{x_i} = (1/\sqrt{2},1/\sqrt{2})_{x_i}^T$, 
i.e., a matrix product state (MPS) with bond dimension one.

The 1-d whip circuit state with and without the boundary gates $e^{-\im Z_1 Y_L/2}$ are represented by the following tensor networks:
\begin{center}
\begin{tabular}{cc}
\begin{tikzpicture}
    \SetVertexStyle[FillColor=BrickRed,MinSize=0.4\DefaultUnit,FillOpacity=1.0]
\Vertex[x=0,y=0,label=$\ket{+}$,position=above]{A1}
\Vertex[x=1.0,y=0,label=$\ket{+}$,position=above]{A2}
\Vertex[x=2.0,y=0,label=$\ket{+}$,position=above]{A3}
\Vertex[x=3.0,y=0,label=$\ket{+}$,position=above]{A4}
\Vertex[x=4.0,y=0,label=$\ket{+}$,position=above]{A5}
\Vertex[x=5.0,y=0,label=$\ket{+}$,position=above]{A6}
\Vertex[x=0,y=-0.5,color=lightgray]{T1_L}
\Vertex[x=1.0,y=-0.5,color=black]{T1_R}
\Edge(T1_L)(T1_R)
\Vertex[x=1.0,y=-1,color=lightgray]{T2_L}
\Vertex[x=2.0,y=-1,color=black]{T2_R}
\Edge(T2_L)(T2_R)
\Vertex[x=2.0,y=-1.5,color=lightgray]{T3_L}
\Vertex[x=3.0,y=-1.5,color=black]{T3_R}
\Edge(T3_L)(T3_R)
\Vertex[x=3.0,y=-2.0,color=lightgray]{T4_L}
\Vertex[x=4.0,y=-2.0,color=black]{T4_R}
\Edge(T4_L)(T4_R)
\Vertex[x=4.0,y=-2.5,color=lightgray]{T5_L}
\Vertex[x=5.0,y=-2.5,color=black]{T5_R}
\Edge(T5_L)(T5_R)
\Vertex[x=0,y=-3,Pseudo]{E1}
\Vertex[x=1.0,y=-3,Pseudo]{E2}
\Vertex[x=2.0,y=-3,Pseudo]{E3}
\Vertex[x=3.0,y=-3,Pseudo]{E4}
\Vertex[x=4.0,y=-3,Pseudo]{E5}
\Vertex[x=5.0,y=-3,Pseudo]{E6}
\Edge(A1)(E1)
\Edge(A2)(E2)
\Edge(A3)(E3)
\Edge(A4)(E4)
\Edge(A5)(E5)
\Edge(A6)(E6)
\end{tikzpicture}
(OBC); &
\begin{tikzpicture}
    \SetVertexStyle[FillColor=BrickRed,MinSize=0.4\DefaultUnit,FillOpacity=1.0]
\Vertex[x=0,y=0,label=$\ket{+}$,position=above]{A1}
\Vertex[x=1.0,y=0,label=$\ket{+}$,position=above]{A2}
\Vertex[x=2.0,y=0,label=$\ket{+}$,position=above]{A3}
\Vertex[x=3.0,y=0,label=$\ket{+}$,position=above]{A4}
\Vertex[x=4.0,y=0,label=$\ket{+}$,position=above]{A5}
\Vertex[x=5.0,y=0,label=$\ket{+}$,position=above]{A6}
\Vertex[x=0,y=-0.5,color=lightgray]{T1_L}
\Vertex[x=1.0,y=-0.5,color=black]{T1_R}
\Edge(T1_L)(T1_R)
\Vertex[x=1.0,y=-1,color=lightgray]{T2_L}
\Vertex[x=2.0,y=-1,color=black]{T2_R}
\Edge(T2_L)(T2_R)
\Vertex[x=2.0,y=-1.5,color=lightgray]{T3_L}
\Vertex[x=3.0,y=-1.5,color=black]{T3_R}
\Edge(T3_L)(T3_R)
\Vertex[x=3.0,y=-2.0,color=lightgray]{T4_L}
\Vertex[x=4.0,y=-2.0,color=black]{T4_R}
\Edge(T4_L)(T4_R)
\Vertex[x=4.0,y=-2.5,color=lightgray]{T5_L}
\Vertex[x=5.0,y=-2.5,color=black]{T5_R}
\Edge(T5_L)(T5_R)
\Vertex[x=0.0,y=-3,color=lightgray]{T6_L}
\Vertex[x=1.0,y=-3,color=white]{T6_1}\Vertex[x=2.0,y=-3,color=white]{T6_2}
\Vertex[x=3.0,y=-3,color=white]{T6_3}\Vertex[x=4.0,y=-3,color=white]{T6_4}
\Vertex[x=5.0,y=-3,color=black]{T6_R}
\Edge(T6_L)(T6_R)
\Vertex[x=0,y=-3.8,Pseudo]{E1}
\Vertex[x=1.0,y=-3.8,Pseudo]{E2}
\Vertex[x=2.0,y=-3.8,Pseudo]{E3}
\Vertex[x=3.0,y=-3.8,Pseudo]{E4}
\Vertex[x=4.0,y=-3.8,Pseudo]{E5}
\Vertex[x=5.0,y=-3.8,Pseudo]{E6}
\Edge(A1)(E1)
\Edge(A2)(E2)
\Edge(A3)(E3)
\Edge(A4)(E4)
\Edge(A5)(E5)
\Edge(A6)(E6)
\end{tikzpicture}
(PBC),
\end{tabular}
\end{center}
where the non-local gate in the PBC case is incorporated by inserting identity tensors (shown as white circles).
Direct contraction of local tensors leads to a final-state MPS with bond dimension $\chi=2$ for the OBC circuit and bond dimension $\chi=4$ for the PBC case. 
In fact, it is natural to contract the bond tensors and site tensors connected by the curved lines together as shown below,
\begin{center}
\begin{tikzpicture}
    \SetVertexStyle[FillColor=Emerald,MinSize=0.4\DefaultUnit,FillOpacity=1.0]
\Vertex[x=0.0,y=0]{F1}
\Vertex[x=0.5,y=0,color=Salmon,shape=diamond]{L1}
\Vertex[x=1.0,y=0]{F2}
\Vertex[x=1.5,y=0,color=Salmon,shape=diamond]{L2}
\Vertex[x=2.0,y=0]{F3}
\Vertex[x=2.5,y=0,color=Salmon,shape=diamond]{L3}
\Vertex[x=3.0,y=0]{F4}
\Vertex[x=3.5,y=0,color=Salmon,shape=diamond]{L4}
\Vertex[x=4.0,y=0]{F5}
\Vertex[x=4.5,y=0,color=Salmon,shape=diamond]{L5}
\Vertex[x=5.0,y=0]{F6}
\Edge(F1)(F6)
\Vertex[x=0,y=-1,Pseudo,size=0]{G1}
\Vertex[x=1.0,y=-1,Pseudo,size=0]{G2}
\Vertex[x=2.0,y=-1,Pseudo,size=0]{G3}
\Vertex[x=3.0,y=-1,Pseudo,size=0]{G4}
\Vertex[x=4.0,y=-1,Pseudo,size=0]{G5}
\Vertex[x=5.0,y=-1,Pseudo,size=0]{G6}
\Edge(G1)(F1)
\Edge(G2)(F2)
\Edge(G3)(F3)
\Edge(G4)(F4)
\Edge(G5)(F5)
\Edge(G6)(F6)
\Edge[lw=1.0,bend=45,color=green](L1)(F2)
\Edge[lw=1.0,bend=45,color=green](L2)(F3)
\Edge[lw=1.0,bend=-45,color=green](L3)(F3)
\Edge[lw=1.0,bend=-45,color=green](L4)(F4)
\Edge[lw=1.0,bend=-45,color=green](L5)(F5)
\Vertex[x=2.75,y=-0.7,size=0,Pseudo]{C1}
\Vertex[x=2.75,y=0.3,size=0,Pseudo]{C2}
\Edge[lw=0.8,style={dashed},label=cut,position={below=1mm}](C1)(C2)
\end{tikzpicture}
\end{center}

For the OBC whip circuit, all local tensors become left- or right-orthogonal, such that the orthogonal center is placed at the middle of the chain.
Performing singular value decomposition (SVD) on the tensor closest to the cut yields the Schmidt decomposition
\begin{equation}
    \ket{\psi} = \sum_{\alpha=0}^{1} \sqrt{\lambda_\alpha} \ket{\phi_\alpha^L}\otimes\ket{\phi_\alpha^R},
\end{equation}
with singular values $\sqrt{\lambda_0}=|\cos(\theta/2)|,\sqrt{\lambda_1}=|\sin(\theta/2)|$, directly related to the bond tensor elements of Eq.~\eqref{eq:tn-zy-decomp}.
Therefore, the entanglement spectrum $\xi_{\alpha}=-\log \lambda_{\alpha}$  is degenerate only at discrete points $\pi/2+m\pi, m\in\ZZ$, and no degeneracy shift across these points. Additionally, the entanglement entropy reads
\begin{equation}
        S(\theta) = -\sum_{\alpha=0}^{1} \lambda_\alpha \log(\lambda_\alpha) = -\cos^2(\theta/2)\log(\cos^2(\theta/2)) - \sin^2(\theta/2)\log(\sin^2(\theta/2)).
\end{equation}
As plotted in Fig.~\ref{fig:energy-landscape-1D}, both the entanglement entropy and spectrum have no distinct behavior on the two sides of the whip angle $\theta=\pm\pi/2$, indicating that the 1-d OBC whip circuit has no phase transition at these angle values.

For the PBC whip circuit, the bulk tensors after contraction still satisfy the left/right orthogonality, but the edge tensors no longer do.
A QR decomposition is then recursively applied on each local tensor from both ends toward the center to restore orthogonality \cite{ITensor}.
The entanglement entropy and spectrum are calculated as above.
Numerical results for a cycle with $L=100$ are shown in Fig.~\ref{fig:energy-landscape-1D}.
In contrast to the 2-d case discussed in the main text, the spectrum degeneracy does not change across the whip angle $\theta=\pm\pi/2$, showing that the 1-d PBC whip circuit also has no phase transition at these angle values.

\begin{figure*}
    \centering
    \includegraphics[width=0.4\textwidth]{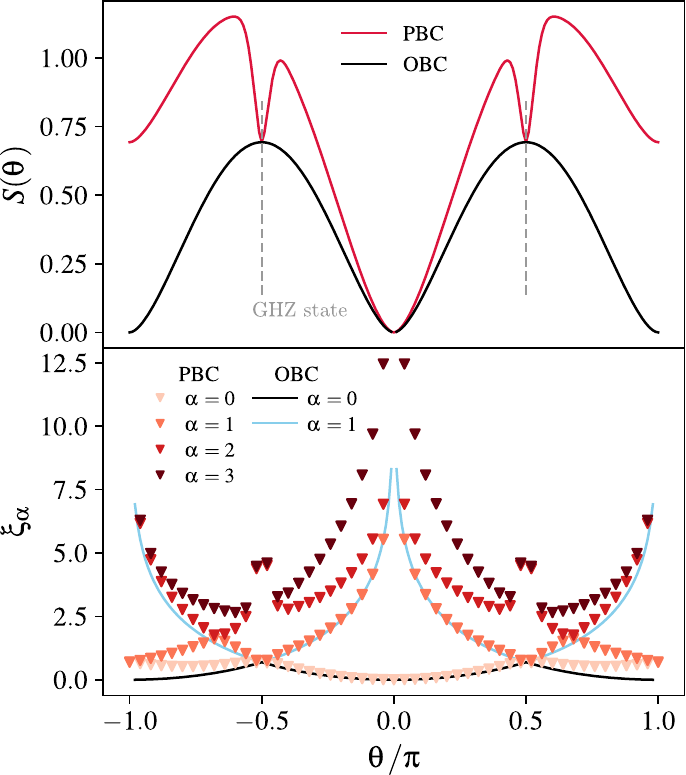}
    \caption{Entanglement entropy and entanglement spectrum of the 1-d Ising whip circuit as functions of the whip angle $\theta$. The OBC results are obtained from the analytic formula, and the PBC results are numerically obtained on a cycle with $L=100$.}
    \label{fig:energy-landscape-1D}
\end{figure*}

\section{Whip circuit on a 2-d lattice}
This note studies the properties of the Ising whip circuit on a 2-d lattice. In Sec.~\ref{sec:ZZ-expectation}, we give the rigorous proof on the expectation value $\langle Z_iZ_j\rangle_{\theta}$ in Eq.~(7) of the main text, and then give in Sec.~\ref{sec:Time complexity of evaluating $ZZ$ expectation} the time complexity of evaluating the expectation $\langle Z_iZ_j\rangle_{\theta}$ using quantum computers and naive Pauli propagation method. Sec.~\ref{sec:correlation-function} studies the correlation function of the 2-d whip circuit, whose critical exponent is analytically obtained. We study in Sec.~\ref{sec:order-parameter} the order parameter and prove Eq.~(10) of the main text. Sec.~\ref{sec:Whip circuits without phase transition} presents the 2-d whip circuit without phase transition. Finally, the VQE performances of the 2-d whip circuits with and without phase transition are compared in Sec.~\ref{sec:VQE-performance}.

\subsection{Hamiltonian expectation}\label{sec:ZZ-expectation}
In this section, we analytically derive the expectation of the Ising Hamiltonian local term $\langle Z_i Z_j\rangle_{\theta}$ of the 2-d whip circuit given in the main text. This result provides strong evidence that the $2$-d whip circuit hosts a continuous phase transition. Since the Ising Hamiltonian has a $ZZ$ term on each link, the $ZZ$ expectation value is equal to the Hamiltonian energy density averaged over all links. We therefore consider a single $ZZ$ term located in the bulk of the $2$-d lattice, as illustrated by the two blue points in Fig.~\ref{fig:loops_analytical_ZZ}(\textbf{a}). Without loss of generality, we construct an $x$-$y$ coordinate system with the origin at one of the $ZZ$ qubits.
We label $ZZ$ operators by their coordinates, e.g., $Z_{0,0}Z_{0,1}$. 
Each arrow in Fig.~\ref{fig:loops_analytical_ZZ}(\textbf{a}) denotes a $ZY$ rotation $e^{-\im\theta ZY/2}$ applied sequentially to the initial state $\ket{+}^{\otimes N}$ from top to bottom of Fig.~\ref{fig:loops_analytical_ZZ}(\textbf{a}). 
\\\\
\noindent \textbf{Local contribution}~---~The local contribution of the whip circuit to the $Z_{0,0}Z_{0,1}$ expectation is identical to $Z_{L-1}Z_L$ or $Z_{L}Z_1$ in Eq.~\eqref{eq:exact-expected-cut-number-2-regular}, where we have 
\begin{align}
    \langle Z_{0,0}Z_{0,1}\rangle_{\theta,\text{local}} = -\sin\theta\cos\theta.
    \label{eq:local-contribution}
\end{align}
At the GHZ state point $\theta=-\pi/4$, $\langle Z_{0,0}Z_{0,1}\rangle_{(-\pi/4),\text{local}}=1/2$. On the other hand, the full expectation of the GHZ state is $ \langle Z_{0,0}Z_{0,1}\rangle_{(-\pi/4)}=1$. This $1/2$ difference comes from non-local cycle contributions that we evaluate in the following content.
\\\\
\noindent \textbf{Cycle contribution}~---~ A non-local Pauli path contributing to the $ZZ$ expectation corresponds to a closed loop in the $2$-d bulk. This can be observed by analyzing the action of the whip circuit to $Z_{0,0}Z_{0,1}$ in the Heisenberg picture. The sequential structure of the whip circuit allows us to separate it into layers of $ZY$ rotations that act on the same diagonal cross-section of the system. As shown in Fig.~\ref{fig:loops_analytical_ZZ}(\textbf{a}), starting from $Z_{0,0}Z_{0,1}$, we denote these layers by the order $\LL$ of their actions on $Z_{0,0}Z_{0,1}$. Due to the commutation relation $[Z_i,e^{-\im\theta Z_i Y_j/2}]=0$, only $ZY$ rotations in layers $\LL\geq 1$ of Fig.~\ref{fig:loops_analytical_ZZ}(\textbf{a}) transform $Z_{0,0}Z_{0,1}$, whereas the gray $ZY$ rotations do not transform $Z_{0,0}Z_{0,1}$.

To see how the layered whip circuit in $\LL\geq 1$ generates non-local Pauli paths, we focus on the action of one layer $\LL$ to the qubit at $(x,y)$, as illustrated in Fig.~\ref{fig:loops_analytical_ZZ}(\textbf{b}).
The qubit is affected by two rotations within $\LL$. Due to the circuit structure, these rotations are also the last rotations throughout the full circuit applied to the qubit $(x,y)$. 
This is important, because during the Heisenberg time-evolution of $Z_{0,0}Z_{0,1}$, after applying the layer $\LL$, any Pauli strings containing $Z_{x,y}$ or $Y_{x,y}$ will not contribute to the final expectation value with respect to the initial state $\ket{+}$ at qubit $(x,y)$. 
This means, we restrict our analysis to Pauli strings that after the gate layer $\LL$, contain only $X_{x,y}$ or identity $I_{x,y}$. As an example, consider the action of the two gates in $\LL =1$ acting on the local $Z_{0,0}$
\begin{equation}
    \begin{aligned}
    &e^{\im\theta Y_{0,0}Z_{1,0}/2}e^{\im\theta Y_{0,0}Z_{0,1}/2}Z_{0,0}e^{-\im\theta Y_{0,0}Z_{1,0}/2}e^{-\im\theta Y_{0,0}Z_{0,1}/2} \\
    &= \cos^2(\theta) Z_{0,0} -  \cos(\theta)\sin(\theta) X_{0,0}Z_{0,1}  - \cos(\theta)\sin(\theta) X_{0,0}Z_{1,0} + \sin(\theta)^2 Z_{1,0}Z_{0,0}Z_{0,1}.
    \label{eq:one-propagation}
\end{aligned}
\end{equation}
As mentioned above, we can ignore the first and fourth term throughout the rest of the circuit of $\LL \geq 2$. Because they contain $Z_{0,0}$ which will not be transformed by any gates in $\LL \geq 2$ in the whip circuit, such that these terms vanish due to $\bra{+}Z\ket{+}=0$. The second and third terms contain $X_{0,0}$ and will therefore contribute to the expectation value since $\bra{+}X\ket{+}=1$. Additionally, these two terms contain one local $Z$ at each qubit of the next layer $\LL=2$, such that the above propagation of $\LL=1$ can be applied directly to $\LL=2$ and subsequent layers after that. Thus in this analysis, a single local observable $Z$ generates Pauli strings that always contain exactly one local $Z$ throughout the layered whip circuit. 

\begin{figure}
    \centering
    \includegraphics[width=0.95\linewidth]{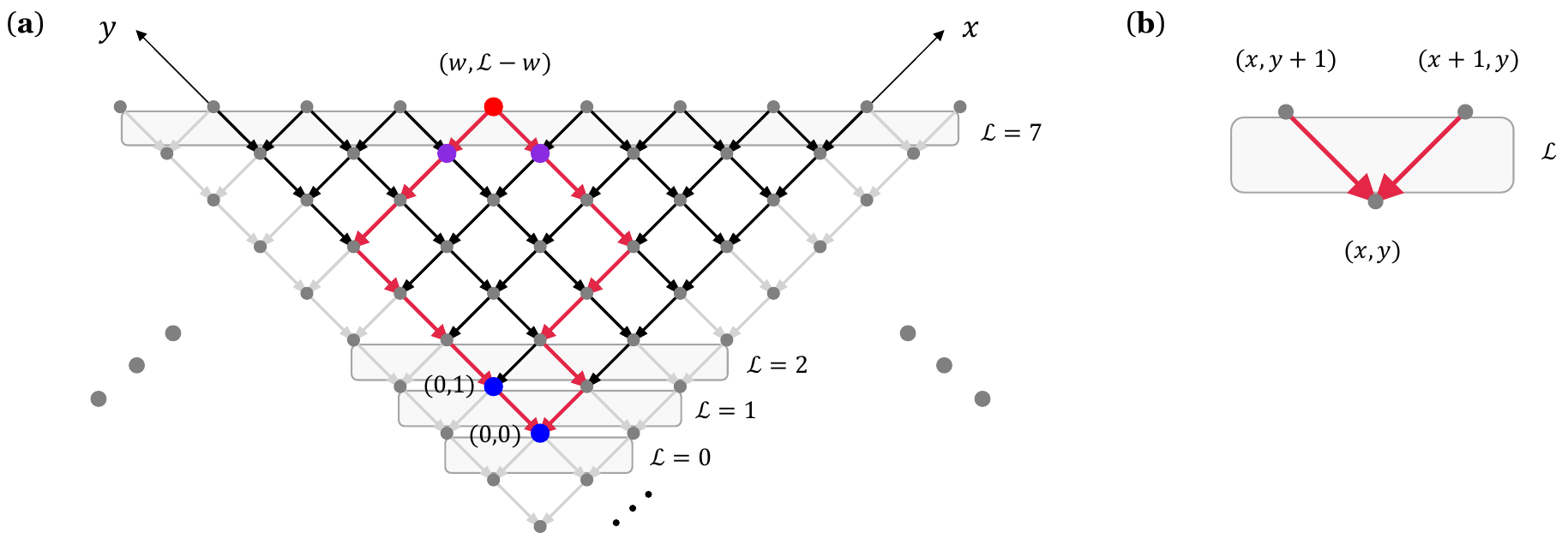}
    \caption{The coordinate system of analytically evaluating the  expectation $\langle Z_{0,0}Z_{0,1}\rangle_{\theta}$. (\textbf{a}) The Pauli propagation path of $Z_{0,0}$ and $Z_{0,1}$ is a non-intersecting oriented cycle (red arrows) with a single source at $(w,\LL-w)$ and a single sink at $(0,0)$. (\textbf{b}) Generation of non-local Pauli paths by the layered structure of the 2-d Ising whip circuit. $\LL$ is the layer index.}
    \label{fig:loops_analytical_ZZ}
\end{figure}

Extending the above argument to $Z_{0,0}Z_{0,1}$, it contains two local $Z$s that are transformed to the linear combination of the product of two other local $Z$s after the $\LL$-layer Pauli propagation. The only exception occurs in the case of two adjacent $Z$s, because they can move towards the same qubit in the next layer, as illustrated by the two purple points and the one red point in Fig.~\ref{fig:loops_analytical_ZZ}(\textbf{a}). At the red point it is $Z_{w,\LL-w}^2=I_{w,\LL-w}$ and the dynamics of this group of Pauli strings stops because the identity commutes with any remaining gate in the circuit. The resulting Pauli strings are the ones that contribute to the expectation value. We denote the two propagation paths that close at $(w,\LL-w)$ as a cycle $\CC$. Thus, all non-local contribution to $Z_{0,0}Z_{0,1}$ expectation value correspond to cycles each with a unique sink at $(0,0)$ and a unique source at $(x_0>0,y_0>0)$, and we refer to this non-local contribution as a \textit{cycle contribution}.

The total expectation value is obtained by a weighted summation over all cycle contributions. The weight is related to the cycle with length $|\CC|=2\LL$ closed after the $\LL$-layer whip circuit. Note that in the second and third term of Eq.\eqref{eq:one-propagation}, each layer for each local $Z$ generates a weight of $-\cos(\theta)\sin(\theta)$. Thus, a cycle closed after $\LL$ layers has a weight $-(\cos(\theta)\sin(\theta))^{2\LL-1}$, where the $-1$ in the exponent appears because $Z_{0,1}$ is propagated by $\LL-1$ layers. 
The number of cycles that close at $(w,\LL-w)$ in Fig.~\ref{fig:loops_analytical_ZZ}(\textbf{a}) can be determined combinatorially by invoking the Lindström-Gessel-Viennot (LGV) lemma. 
LGV lemma can be applied to obtain the number of non-intersecting pairs of paths on a lattice starting from $(w-1,\LL-w)$ and $(w,\LL-w-1)$ [The two purple points in Fig.~\ref{fig:loops_analytical_ZZ}(\textbf{a})], and ending at $(0,0)$ and $(0,1)$. This number of non-intersecting pairs equals to the desired number of cycles and is given by
\begin{equation}
     N_{\CC}(w,\LL)= \begin{pmatrix}
        \LL-2\\w-1
    \end{pmatrix}
    \begin{pmatrix}
        \LL-2\\w-1
    \end{pmatrix}
    -
   \begin{pmatrix}
        \LL-2\\w
    \end{pmatrix}
    \begin{pmatrix}
        \LL-2\\w-2
    \end{pmatrix},
    \label{eq:app-LGV-lemma}
\end{equation}
where the second term is discarded if $w<2$. Since $l:=\LL-w$ and $w$ is the length and width of the rectangle expanding from $(0,0)$ to $(w,\LL-w)$, $N_{\CC}(w,\LL)$ can be rewritten as a function of $l$ and $w$
\begin{equation}
    \begin{aligned}
    N_{\CC}(w,\LL)\rightarrow N_{\CC}(l,w) &= \begin{pmatrix}
        l+w-2\\w-1
    \end{pmatrix}
    \begin{pmatrix}
        l+w-2\\w-1
    \end{pmatrix}
    -
   \begin{pmatrix}
        l+w-2\\w
    \end{pmatrix}
    \begin{pmatrix}
        l+w-2\\w-2
    \end{pmatrix}\\
    &=\frac{lw}{(l+w-1)(l+w)^2}{l+w\choose l}^2,
\end{aligned}
\end{equation}
which gives Eq.~(4) in the main text. 

In summary, the total cycle contribution after $\LL_0$-layers of the whip circuit reads
\begin{align}
    \langle Z_{0,0}Z_{0,1}\rangle_{\theta,\text{cycle}}= -\sum_{\LL=2}^{\LL_0} \sum_{w=1}^{\LL-1} N_{\CC}(w,\LL) (\cos\theta\sin\theta)^{2\LL-1}.
    \label{eq:cycle-contribution}
\end{align}
Here the first sum iterates over the different layers $\LL$, whereas the second sum iterates over all possible endpoints within a given layer $\LL$.

Since the path weight $ (\cos\theta\sin\theta)^{2\LL-1}$ has no dependence on $w$, the following derivation benefits from first processing the summation $\sum_{w=1}^{\LL-1} N_{\CC}(w,\LL)$ in Eq.~\eqref{eq:cycle-contribution}. This summation is simplified by the Chu–Vandermonde identity
\begin{equation}
    \sum_{k=0}^r\begin{pmatrix}
        m\\k
    \end{pmatrix}\begin{pmatrix}
        n\\r-k
    \end{pmatrix}=\begin{pmatrix}
        m+n\\r
    \end{pmatrix},
\end{equation}
for any non-negative integers $r$, $m$, $n$. Thus, the summation over the first term in Eq.~\eqref{eq:app-LGV-lemma} evaluates to
\begin{equation}
\begin{aligned}
    \sum_{w=1}^{\LL-1} \begin{pmatrix}
        \LL-2\\w-1
    \end{pmatrix}
    \begin{pmatrix}
        \LL-2\\w-1
    \end{pmatrix}&= \sum_{w=1}^{\LL-1} \begin{pmatrix}
        \LL-2\\w-1
    \end{pmatrix}
    \begin{pmatrix}
        \LL-2\\ \LL-1-w
    \end{pmatrix}=
    \begin{pmatrix}
        2(\LL-2)\\ \LL-2
    \end{pmatrix}
    \\
    &=
    \frac{(2(\LL-2))!}{((\LL-2)!)^2} = \frac{4^{\LL-2}}{\sqrt{\pi}}\LL(\LL-1)\frac{\Gamma(\LL-\frac{3}{2})}{\LL!},
\end{aligned}
\end{equation}
with the gamma function $\Gamma(z)=\int_0^\infty t^{z-1}e^{-t}dt$. Summation over the second term of Eq.~\eqref{eq:app-LGV-lemma} similarly evaluates to
\begin{equation}
    \begin{aligned}
     \sum_{w=1}^{\LL-1}\begin{pmatrix}
        \LL-2\\w
    \end{pmatrix}
    \begin{pmatrix}
        \LL-2 \\w-2
    \end{pmatrix}
    &=\sum_{w=0}^{\LL-1}\begin{pmatrix}
        \LL-2 \\w
    \end{pmatrix}
    \begin{pmatrix}
        \LL-2 \\\LL-w
    \end{pmatrix}=
    \begin{pmatrix}
        2(\LL-2)\\\LL
    \end{pmatrix}
    \\
    &=\frac{4^{\LL-2}}{\sqrt{\pi}}(\LL-2)(\LL-3)\frac{\Gamma(\LL-\frac{3}{2})}{\LL!}.
\end{aligned}
\end{equation}
In the first equality, the term of $w=0$ is freely added since we define ${\LL-2 \choose w-2}$ to be zero if $w<2$ following the LGV lemma.  Combining these two equations gives
\begin{equation}
    \begin{aligned}
    \sum_{w=1}^{\LL-1} N_{\CC}(w,\LL) &= \frac{4^{\LL-2}\Gamma(\LL-\frac{3}{2})}{\sqrt{\pi}\LL !}[\LL(\LL-1)-(\LL-2)(\LL-3)]= \frac{4^{\LL-1}}{\sqrt{\pi}}\frac{\Gamma(\LL-\frac{1}{2})}{\LL!}.
    \label{eq:cycle-number}
\end{aligned}
\end{equation}


\noindent \textbf{Full contribution}~---~ Combing the local contribution in Eq.~\eqref{eq:local-contribution} with the cycle contribution in Eq.~\eqref{eq:cycle-contribution}, and using the cycle number formula Eq.~\eqref{eq:cycle-number}, we write out the full contribution to the $ Z_{0,0}Z_{0,1}$ expectation value in the infinite volume limit as 
\begin{align}
    \langle Z_{0,0}Z_{0,1}\rangle_{\theta}=-\cos\theta\sin\theta-
    \lim_{{\LL_0}\rightarrow\infty}\sum_{\LL=2}^{\LL_0} \frac{4^{\LL-1}}{\sqrt{\pi}}\frac{\Gamma(\LL-\frac{1}{2})}{\LL!}(\cos\theta\sin\theta)^{2\LL-1}.
\end{align}
Although we did not define Eq.~\eqref{eq:cycle-number} for $\LL=1$, the resulting number of cycles has a well-defined value of $\left.\frac{4^{\LL-1}}{\sqrt{\pi}}\frac{\Gamma(\LL-\frac{1}{2})}{\LL!}\right|_{\LL=1}=1$. Thus, the local and cycle contribution can be combined as
\begin{align}
    \langle Z_{0,0}Z_{0,1}\rangle_{\theta}=-
    \lim_{{\LL_0}\rightarrow\infty}\sum_{\LL=1}^{\LL_0} \frac{4^{\LL-1}}{\sqrt{\pi}}\frac{\Gamma(\LL-\frac{1}{2})}{\LL!}(\cos\theta\sin\theta)^{2\LL-1}.
    \label{eq:ZZ-full-contribution}
\end{align}
By shifting the variable $\LL\to \LL+1$ and taking the definition of gamma function, we have
\begin{equation}
    \begin{aligned}
    \langle Z_{0,0}Z_{0,1}\rangle_{\theta}=-\lim_{{\LL_0}\rightarrow  \infty} \sum_{\LL=0}^{\LL_0} \frac{4^{\LL}}{\sqrt{\pi}}\frac{\int_0^\infty t^{\LL-\frac{1}{2}}e^{-t} dt}{(\LL+1)!} \frac{1}{2^{2\LL+1}}(\sin2\theta)^{2\LL+1}\\
    =-\lim_{{\LL_0}\rightarrow  \infty} \int_0^\infty  t^{-1}\frac{e^{-t} }{2\sqrt{\pi}} \sum_{\LL=0}^{\LL_0}\frac{\sqrt{t}^{2\LL+1}}{(\LL+1)!}(\sin2\theta)^{2\LL+1}dt,
\end{aligned}
\end{equation}
where the infinite summation can be simplified as per the Taylor expansion of an exponential function
\begin{align}
    \lim_{{\LL_0}\rightarrow  \infty}\sum_{\LL=0}^{\LL_0}\frac{\sqrt{t}^{2\LL+1}}{(\LL+1)!}(\sin2\theta)^{2\LL+1} = \frac{1}{\sqrt{t}\sin(2\theta)} (e^{t\sin^2(2\theta)}-1).
\end{align}
Thus we have 
\begin{equation}
    \begin{aligned}
    \langle Z_{0,0}Z_{0,1}\rangle_{\theta}&=-\int_0^\infty \frac{e^{-t} }{2\sqrt{\pi}} \frac{e^{t\sin^2(2\theta)}-1}{t^{\frac{3}{2}}\sin(2\theta)}dt\\
    &=-\frac{1}{\sin(2\theta)} \int_0^\infty \frac{1 }{2\sqrt{\pi}} \frac{e^{-t\cos^2(2\theta)}-e^{-t}}{t^{\frac{3}{2}}}dt\\
    &=-\frac{1}{\sin(2\theta)}\{\underbrace{\left.[ \frac{ -1 }{\sqrt{\pi}} \frac{e^{-t\cos^2(2\theta)}-e^{-t}}{t^{\frac{1}{2}}}]\right|_{0}^\infty}_0 + \int_0^\infty \frac{1 }{\sqrt{\pi}} \frac{-\cos^2(2\theta)e^{-t\cos^2(2\theta)}+e^{-t}}{t^{\frac{1}{2}}}dt\}\\
    &=\frac{1}{\sin(2\theta)} (\sqrt{\cos^2(2\theta)}-1)\\
    &=\frac{|\cos(2\theta)|-1}{\sin(2\theta)},\label{eq:expect}
\end{aligned}
\end{equation}
where we integrated by parts in the third line and applied the definition of the gamma function in the fourth line. 

This final result correctly reproduces $\langle  Z_{0,0}Z_{0,1}\rangle_{\theta}=1$ at $\theta=-\pi/4+m\pi$, and $\langle  Z_{0,0}Z_{0,1}\rangle_{\theta}=-1$ at $\theta=\pi/4+m\pi,m\in\mathbb{Z}$, as required by the generated ferromagnetic and anti-ferromagnetic GHZ state. At these values of $\theta_c\equiv\pi/4+m\pi/2$, $\langle Z_{0,0}Z_{0,1}\rangle_{\theta}$ has a cusp that is continuous but non-differentiable due to the absolute value function $|\cos(2\theta)|$. This cusp is similar to the cusp of the ground state energy curve of the $1$-d transverse-field Ising model at the critical point, which indicates a continuous phase transition.

The derivative of $\langle Z_{0,0}Z_{0,1}\rangle_{\theta}$ with respect to $\theta$ everywhere else of $\theta_c$ gives
\begin{equation}
    \begin{aligned}
    \partial_\theta \langle  Z_{0,0}Z_{0,1}\rangle_{\theta}
    &=2(\cot^2(2\theta)-|\cos(2\theta)|\csc^2(2\theta))\sec(2\theta)\\
    &=2(\frac{\cos(2\theta)-\mathrm{sign}[\cos(2\theta)]}{\sin^2(2\theta)}).
\end{aligned}
\end{equation}
This expression has a clear discontinuity whenever $\cos(2\theta)$ changes its sign at $\theta_c$, which corresponds to a divergence in the 2nd derivative $\partial_{\theta}^2\langle  Z_{0,0}Z_{0,1}\rangle_{\theta}$.
\\\\
\noindent \textbf{Approximating $ZZ$ expectation by $\LL_0$ cut-off}~---~ As shown in Eq.~\eqref{eq:ZZ-full-contribution}, the contribution to the $ZZ$ expectation value as a function of the number of layer $\LL$ reads
\begin{align}
    f(\LL,\theta) = \frac{4^{\LL-1}}{\sqrt{\pi}}\frac{\Gamma(\LL-\frac{1}{2})}{\LL!}(\cos\theta\sin\theta)^{2\LL-1}=\frac{1}{2\sqrt{\pi}} \frac{\Gamma(\LL-\frac{1}{2})}{\Gamma(\LL+1)}(\sin 2\theta)^{2\LL-1}
\end{align}
Its scaling behavior for large $\LL$ is obtained by the Stirling formula, which gives
\begin{align}
    f(\LL,\theta)\sim \frac{1}{\LL^{3/2}} (\sin 2\theta)^{2\LL-1}.
\end{align}
Approximating the $ZZ$ expectation value in the presence of a $\LL_0$ cut-off gives an error in the expectation value
\begin{align}
    \epsilon(\LL_0)=\sum_{\LL=\LL_0}^{\infty}f(\LL,\theta)\sim \int_{\LL_0}^{\infty}\frac{1}{\LL^{3/2}} (\sin 2\theta)^{2\LL-1}d\LL\sim \left\{ \begin{array}{ll}
  1/\sqrt{\LL_0},& \textrm{if $\theta= \theta_c$;}\\
 (\sin 2\theta)^{2\LL_0-1}/(-\LL_0^{3/2}\log|\sin2\theta|), & \textrm{if $\theta\neq \theta_c$,}
  \end{array} \right.
  \label{eq:cut-off-error-scaling}
\end{align}
in the case of large $\LL_0\gg 1$. We see that the error decreases exponentially with respect to $\LL_0$ if $\theta$ is far from the critical point $\theta_c$, but only polynomially at the critical point. This polynomial decreasing behavior leads to the hardness of the classical simulation of the whip circuit compared to quantum computers at the critical point. This is discussed in the following subsection.

\subsection{Time complexity of quantum computers \textit{v.s.} naive Pauli propagation}\label{sec:Time complexity of evaluating $ZZ$ expectation}

For noiseless quantum computers, evaluating one $ZZ$ expectation value requires time complexity given by the following proposition:


\begin{proposition}\label{prop:QC-time-complexity}
    On an $L\times L$ lattice, evaluating the $Z_iZ_j$ expectation of the whip circuit on quantum computers within an error $\epsilon$ requires time complexity
    \begin{align}
        T_{\mathrm{QC}}=\OO\left(\frac{L^2}{\epsilon^2}\right).
        \label{eq:qc-time-complexity}
    \end{align}
\end{proposition}
\begin{proof}[proof]
To estimate the $ZZ$ expectation within an error $\epsilon$, the required measurement shot number should satisfy~\cite{lin2022lecturenotesquantumalgorithms}
\begin{align}
    N_{\text{shot}}\geq \frac{\text{Var}_{\theta}(ZZ) }{\epsilon^2}=\frac{1-\langle ZZ\rangle_{\theta}^2}{\epsilon^2}.
\end{align}
This requirement is satisfied by choosing $N_{\text{shot}} = 1/\epsilon^2$. For each shot of the quantum computer, $\OO(L^2)$ $ZY$ rotations are performed. Thus, the total time complexity of estimating the $ZZ$ expectation reads
\begin{align}
    T_{\mathrm{QC}}=\OO\left(\frac{L^2}{\epsilon^2}\right).
\end{align}
\end{proof}
We see the time complexity of quantum computer polynomially depending on the system size as long as $\epsilon$ decreases at most polynomially to the system size. In comparison, the time complexity of naive Pauli propagation has scaling behavior given by the following proposition:



\begin{proposition}\label{prop:PP-time-complexity}
    Evaluating the $Z_iZ_j$ expectation of the whip circuit using naive Pauli propagation within an error $\epsilon$ requires time complexity
    \begin{align}
        T_{\text{naive PP}}=\OO(\epsilon^34^{1/\epsilon^2})
    \end{align}
\end{proposition}
\begin{proof}[proof]
Naive Pauli propagation method evaluating $Z_iZ_j$ expectation by summing over contributions from all Pauli paths. Thus, it requires time complexity at least proportional to the number of Pauli paths $\NN_{\CC}$ contributing to the $ZZ$ expectation. The Pauli propagation method takes a cut-off of Pauli path weight to keep a moderate number of Pauli paths. For the Ising whip circuit, the Pauli path weight is proportional to the layer index $\LL$ introduced in the previous subsection. Thus, the Pauli path weight truncation is equivalent to a layer truncation taking $\LL$ up to $\LL_0$. According to Eq.~\eqref{eq:cycle-number}, the number of Pauli paths at one of the layer $\LL$ has the scaling behavior 
\begin{align}
    \sum_{w=1}^{\LL-1}N_{\CC}(w,\LL)=\frac{4^{\LL-1}}{\sqrt{\pi}}\frac{\Gamma(\LL-\frac{1}{2})}{\Gamma(\LL+1)}\sim \OO\left(\frac{4^{\LL}}{\LL^{3/2}}\right).
\end{align}
Due to the dominant exponential factor $4^{\LL}$, the total number of Pauli paths up to the truncation layer $\LL_0$ has the same scaling behavior
\begin{align}
    \NN_{\CC}= \sum_{\LL=1}^{\LL_0}\sum_{w=1}^{\LL-1}N_{\CC}(w,\LL)\sim\OO\left(\frac{4^{\LL_0}}{\LL_0^{3/2}}\right).
    \label{eq:total-number-of-Pauli-paths}
\end{align}
Then, according to Eq.~\eqref{eq:cut-off-error-scaling}, for a given error $\epsilon$, in the worst case, i.e., $\theta=\theta_c$, the cut-off $\LL_0$ should grow inverse polynomially to $\epsilon$ as $\LL_0\sim 1/\epsilon^2$. Combined with Eq.~\eqref{eq:total-number-of-Pauli-paths}, the time complexity of the naive Pauli propagation method reads
\begin{align}
        T_{\text{naive PP}}=\OO(\NN_{\CC})  = \OO(\epsilon^34^{1/\epsilon^2}).
        \label{eq:naive-Pauli-scaling}
\end{align}
\end{proof}

Comparing the above two propositions, the time complexity of quantum computers and that of the naive Pauli propagation method are most distinguished if we require the error $\epsilon$ suppressed polynomially with the system size $L$. In this regime, we have the time complexity
\begin{align}
    T_{\text{QC}} \sim\OO(\text{poly}(L))
\end{align}
and 
\begin{align}
    T_{\text{naive PP}}= \OO(4^{\text{poly}(L)}).
\end{align}
On the other hand, the exponential time scaling of $T_{\text{naive PP}}$ is reduced to polynomial for $\theta\neq \theta_c$. Since in this regime, following the derivation from Eq.~\eqref{eq:total-number-of-Pauli-paths} to \eqref{eq:naive-Pauli-scaling}, the time complexity is $T_{\text{naive PP}}\sim \left(\frac{1}{\epsilon \log b}\right)^{\log_{b}2}$ with $b=1/|\sin 2\theta|$, which is polynomial to $L$ for $1/\epsilon \sim \text{poly}(L)$.
Thus, the Ising whip circuit cannot be efficiently simulated by the naive Pauli propagation method at and only at the phase transition point $\theta=\theta_c$. This failure point of the naive Pauli propagation method corresponds to the point out of the $1-\delta$ fraction of the random circuits in Ref.~\cite{Angrisani2025}. However, this point is important since it corresponds to the global minimum of the ferromagnetic and anti-ferromagnetic Ising system.

\subsection{Correlation function}\label{sec:correlation-function}

\begin{figure}
    \centering
    \includegraphics[width=0.5\textwidth]{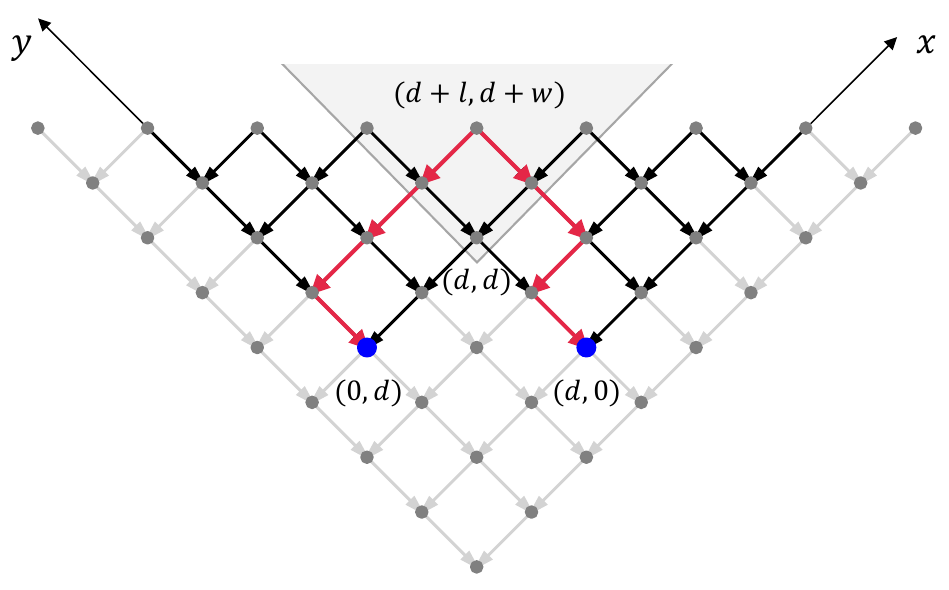}
    \caption{The coordinate system of evaluating the correlation function $\langle Z_{d,0} Z_{0,d}\rangle_{\theta}$. The red arrows denote a non-local Pauli path. The black arrows denote $ZY$ rotations conjugate non-trivially on the observable $Z_{d,0} Z_{0,d}$.}
    \label{fig:correlation-lattice}
\end{figure}

The connected correlation function of two spatially departed local operators in the whip state decays exponentially with their distance
\begin{align}
    C(i,j)\equiv \bra{\phi_{\text{w}}(\theta)}O_i O_j\ket{\phi_{\text{w}}(\theta)}-\bra{\phi_{\text{w}}(\theta)}O_i\ket{\phi_{\text{w}}(\theta)}\bra{\phi_{\text{w}}(\theta)}O_j\ket{\phi_{\text{w}}(\theta)}\sim e^{-|i-j|/\xi}
\end{align}
where $\xi$ is the correlation length, $|i-j|$ is the Manhattan distance of the spatial sites $i$ and $j$. At the critical point $\theta=\theta_c=\pm \pi/4$, the Ising ferromagnetic ground state (with $O_j=Z_j$) and the anti-ferromagnetic one (with $O_j=(-1)^{|j|}Z_j$) in Eq.~\ref{eq:ground-states} have $C(i,j)=1$, corresponding to divergent correlation length $\xi\to \infty$. The critical exponent $\nu$ describe this divergence behavior as $\xi\sim |\theta-\theta_c|^{-\nu}$.


To characterize $\xi$, we evaluate the correlation function $\langle Z_{i}Z_{j}\rangle_{\theta}$ across the 2-d lattice in the limit of $L\rightarrow\infty$. Since the whip circuit is anisotropic, $\langle Z_{i}Z_{j}\rangle_{\theta}$ depends on the horizontal or vertical relative direction of $i$ to $j$.
Here, we focus on the horizontal case $\langle Z_{d,0}Z_{0,d}\rangle$, which is the correlation between sites diagonally separated by a Manhattan distance of $2d$, as illustrated in Fig.~\ref{fig:correlation-lattice}. We can think about these operators as being located in the $d$th cross-section, i.e. $\LL=d$ in Fig.~\ref{fig:loops_analytical_ZZ}, such that the case $d=1$ reproduces the calculation of the previous section.

We now follow the same procedure of calculating the Hamiltonian expectation value. First we note that the earliest point where the operators can meet is $(d,d)$. Then they can meet at one of the sites in the gray-shaded triangle region in Fig.~\ref{fig:correlation-lattice}. Similar to the previous section, the LGV lemma gives the number of non-intersecting paths from the initial configuration to meeting at $(d+l, d+w)$ as
\begin{align}
    N^d(w,l)&=
    \begin{pmatrix}
        d+l+w-1\\w
    \end{pmatrix}
    \begin{pmatrix}
        d+l+w-1\\l
    \end{pmatrix}
    -
    \begin{pmatrix}
        d+l+w-1\\w-1
    \end{pmatrix}
    \begin{pmatrix}
        d+l+w-1\\l-1
    \end{pmatrix}.
    \label{eq:correlation-com-factor}
\end{align}
For the purpose of calculating the expectation value we define $h\equiv l+w$, such that we have the expectation
\begin{equation}
    \begin{aligned}
    \langle Z_{d,0}Z_{0,d} \rangle_\theta 
    &=\sum_{h=0}^{\infty} \sum_{w=0}^{h} N^d(w,h-w) (\cos(\theta)\sin(\theta))^{2h+2d}\\
    &=\sum_{h=0}^{\infty}
    \frac{2d \Gamma(2(d+h))}{h!(2d+h)!}\frac{1}{2^{2h+2d}}
    \sin(2\theta)^{{2h+2d}}\\
    &=\left[\frac{\sin(2\theta)}{1+|\cos(2\theta)|}\right]^{2d}\\
    &=\langle Z_{1,0}Z_{0,1}\rangle^{2d},
    \label{eq:analytic-correlation}
\end{aligned}
\end{equation}
with the detailed derivation following a similar structure to Eq.~(\ref{eq:cycle-contribution}-\ref{eq:expect}). Since the on-site terms $Z_i$ and $Z_j$ anti-commute with the bit-flip operator, their expectation values vanish as $\langle Z_i\rangle_{\theta}=\langle Z_j\rangle_{\theta}=0$. Thus, the connected correlation function reads 
\begin{align}
    C(i,j) = \langle Z_{d,0}Z_{0,d} \rangle_\theta &=\exp\{-2d \log(\frac{1+|\cos(2\theta)|}{\sin(2\theta)})\}=\exp\{-2d \log(\frac{\sin(2\theta)}{1-|\cos(2\theta)|})\}.
\end{align}
Thus, in the vicinity of the critical point $\varphi\equiv \theta-\theta_c \ll 1, \theta_c = \pi/4$, the correlation length is
\begin{align}
    \xi^{-1} =\log(\frac{\cos(2\varphi)}{1-|\sin(2\varphi)|}),
    \label{eq:correlation-length-exact}
\end{align}
whose scaling behavior
\begin{equation}
     \xi^{-1}=  2|\varphi|+\OO(|\varphi|^3)
     \label{eq:analytic-correlation-length}
\end{equation}
corresponds to the critical exponent $\nu=1$.

\begin{figure}
    \centering
    \includegraphics[width=0.9\textwidth]{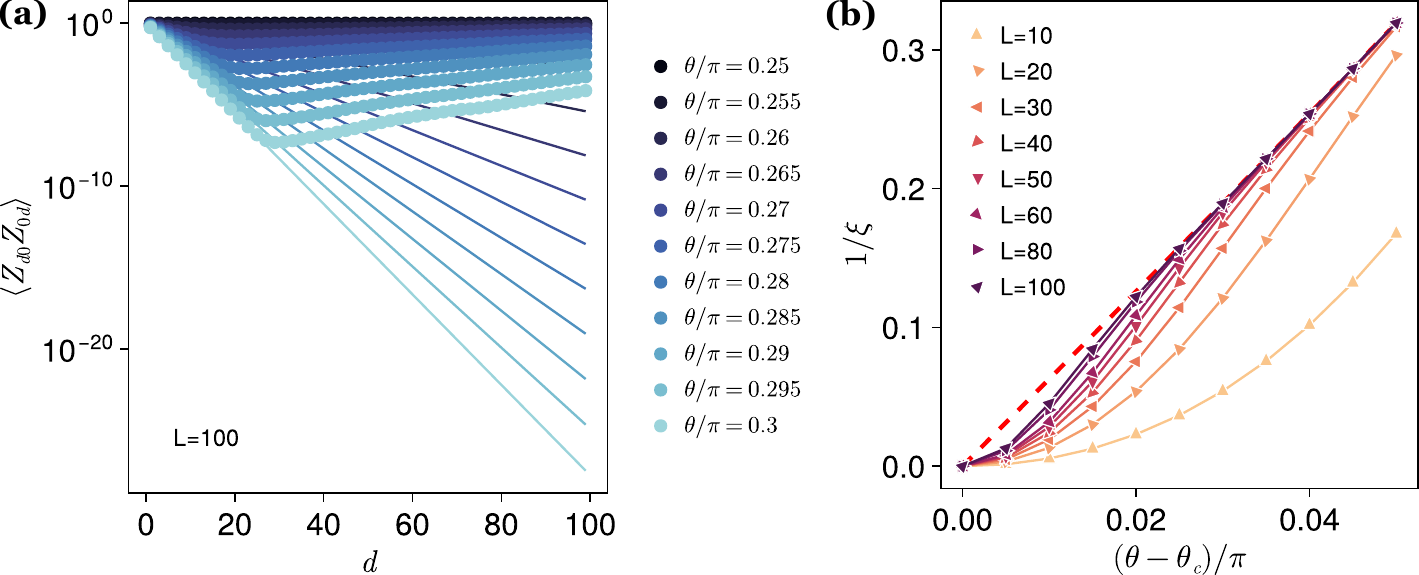}
    \caption{(\textbf{a}) The numerical correlation function $\langle Z_{d,0}Z_{0,d} \rangle$ on a $100\times100$ lattice using PEPS. Solid lines correspond to the analytical results in the $L\to\infty$ limit. (\textbf{b}) Numerical results of the correlation length diverge near the critical point $\theta_c = \pi/4$. The dashed line represents the analytical expression of the correlation length by Eq.~\eqref{eq:correlation-length-exact}.}
    \label{fig:correlation length}
\end{figure}

Figure~\ref{fig:correlation length}(\textbf{a}) plots the numerical results of evaluating the correlation function on a $100\times100$ lattice using PEPS, with different colors corresponding to different values of $\theta$.
Details of the PEPS setup are given in Note \ref{sec:Tensor network time complexity}.
The correlation function exhibits a clear exponential decay at small separations $d$, while its value rises at larger $d$ due to the finite volume effect. The solid lines correspond to the analytical results given by Eq.~\eqref{eq:analytic-correlation}, which agree well with the numerical data in the small-$d$ regime.

Therefore, we compute $\xi$ by extracting the correlation length from the short-distance decay. 
For different system sizes, we plot $\xi$ as a function of $\theta-\theta_c$ in Fig.~\ref{fig:correlation length}(\textbf{b}).
The dashed line indicates the theoretical scaling $\xi^{-1} \sim |\theta-\theta_c|$ in the infinite volume limit. While the finite-size data at $L=10$ deviates from the analytic curve by Eq.~\eqref{eq:correlation-length-exact}, the numerical curves progressively approach the analytic one with the critical exponent $\nu=1$ as the system size increases.
Notably, the deviation is larger as $\theta-\theta_c$ approaches zero due to the strong finite volume effect near the critical point.

\subsection{Order parameter}\label{sec:order-parameter}

\begin{figure}
    \centering
    \includegraphics[width=0.5\textwidth]{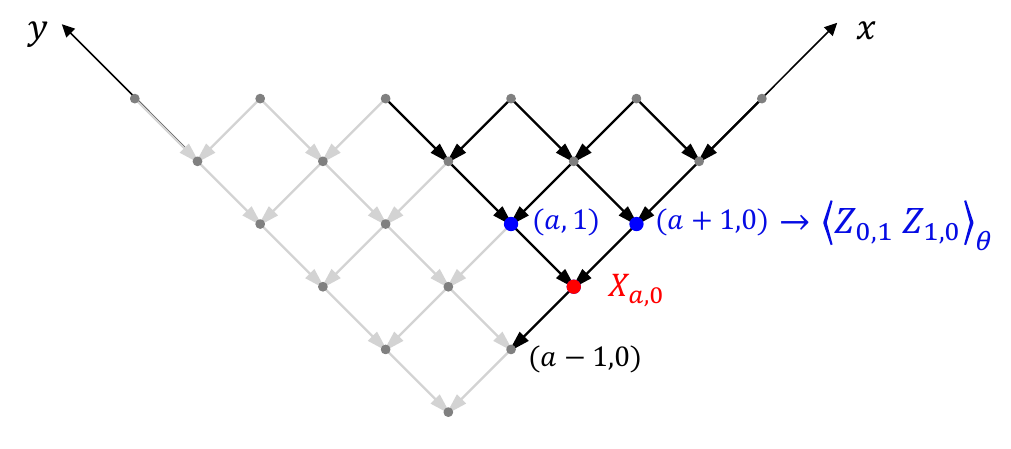}
    \caption{The coordinate system of evaluating the order parameter $\langle X_{a,0}\rangle_{\theta}$. The black arrows denote $ZY$ rotations that conjugate non-trivially on the observable $X_{a,0}$.}
    \label{fig:order-lattice}
\end{figure}

In the main text, we choose the order parameter to be the averaged expectation of local Pauli-$X$s at the lower boundary $\mathcal{B}'$ of the 2-d lattice. In the infinite volume limit, the averaged expectation converges to the expectation of a local $X_{a, 0}$ at the lower boundary of the 2-d lattice, where $a>0$ is an arbitrary positive integer as illustrated by the red dot in Fig.~\ref{fig:order-lattice}. 

Calculating the expectation $\langle X_{a, 0}\rangle_{\theta}$ can be reduced to that of the correlation function $\langle Z_{1,0}Z_{0,1}\rangle_{\theta}$.
To see this, note that the first gate conjugating non-trivially on $X_{a, 0}$ is $e^{-\ii \theta Z_{a, 0}Y_{a-1,0}/2}$, which results in
\begin{equation}
    e^{\ii\theta Z_{a, 0}Y_{a-1, 0}/2}X_{a,0}e^{-\ii\theta Z_{a, 0}Y_{a-1,0}/2} = \cos(\theta)X_{a, 0} - \sin(\theta) Y_{a,0}Y_{a-1,0}.
\end{equation}
The second term proportional to $YY$ does not contribute to the expectation value, since the local $Y_{a-1, 0}$ is unaffected by any future $ZY$ rotations, and vanishes since the initial state gives $\bra{+}Y\ket{+}=0$. Therefore, only the original operator $X_{a, 0}$ remains, with a coefficient of $\cos(\theta)$. The next two gates $e^{-\ii\theta Z_{a+1 , 0}Y_{a, 0}/2}$ and $e^{-\ii\theta Z_{a , 0}Y_{a, 1}/2}$ conjugate on $X_{a, 0}$ as
\begin{equation}
    \cos(\theta) X_{a, 0}\rightarrow \cos^3(\theta) X_{a, 0}+\cos^2(\theta)\sin(\theta)Z_{a+1, 0}Z_{a, 0}+\cos^2(\theta)\sin(\theta)Z_{a, 0}Z_{a, 1}-\cos(\theta)\sin^2(\theta) Z_{a+1, 0} X_{a, 0} Z_{a,1}.
    \label{eq:order-params-first-step}
\end{equation}
In this expression, the first term proportional to $X_{a,0}$ does conjugate by further $ZY$ rotations, leaving the final contribution $\cos^3(\theta)\bra{+}X\ket{+} = \cos^3(\theta)$ to the expectation value. The second and third terms do not contribute to the expectation value due to $\bra{+}Z\ket{+}=0$, and only the fourth term $Z_{a+1, 0} X_{a, 0} Z_{a,1}$ further generates non-local Pauli paths.

For the fourth term $Z_{a+1, 0} X_{a, 0} Z_{a,1}$, further $ZY$ rotations conjugate trivially on the site $(a,0)$, such that $\langle Z_{a+1, 0} X_{a, 0} Z_{a,1}\rangle_{\theta}$ = $\langle Z_{a+1, 0} Z_{a,1}\rangle_{\theta}$ due to $\bra{+}X_{a, 0}\ket{+} = 1$. Thus, we only need to consider the Pauli propagation of $Z_{a+1, 0} Z_{a,1}$, as illustrated by the blue dots in Fig.~\ref{fig:order-lattice}. This Pauli propagation is the same as that of $\langle Z_{d,0}Z_{0,d}\rangle_{\theta}$ with $d=1$ in Eq.~\eqref{eq:analytic-correlation}. Thus, we have
\begin{align}
    \langle Z_{a+1, 0}X_{a, 0} Z_{a,1} \rangle_{\theta} =\left(\frac{\sin(2\theta)}{1+|\cos(2\theta)|}\right)^2.
\end{align}
Inserting this into Eq.~\eqref{eq:order-params-first-step} gives
\begin{align}
    \langle X_{a, 0} \rangle_{\theta} &= \cos^3(\theta) -\cos(\theta) \sin^2(\theta) \langle Z_{a+1, 0} X_{a, 0} Z_{a,1}  \rangle_{\theta} =\frac{\cos(2\theta)+|\cos(2\theta)|}{2\cos(\theta)}.
\end{align}
This gives the analytic formula of the order parameter in the main text.

\subsection{Ising whip circuit without phase transition}\label{sec:Whip circuits without phase transition}

\begin{figure}
    \centering
    \includegraphics[width=0.7\textwidth]{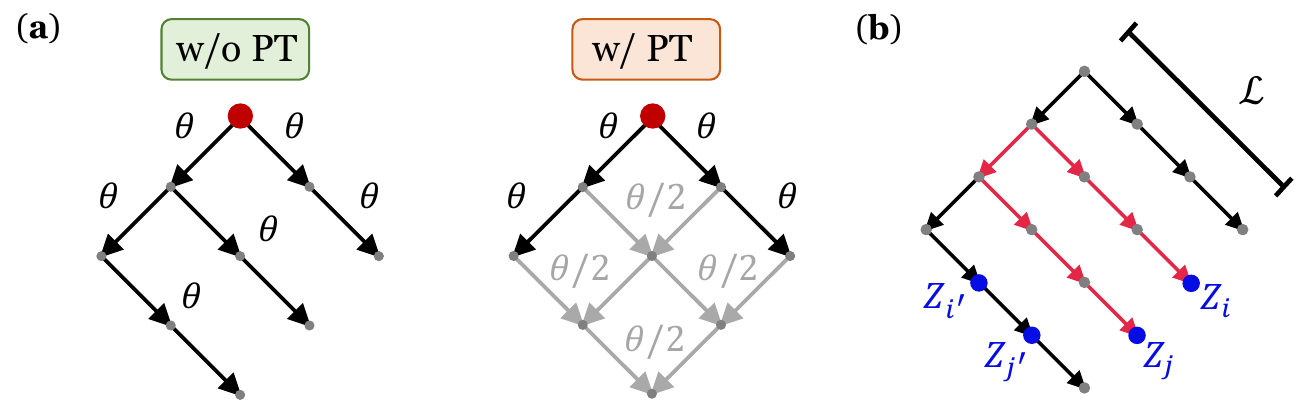}
    \caption{(\textbf{a}) Illustration of the 2-d Ising whip circuit without phase transition (left panel). $ZY$ rotations $e^{-\im\theta Z_iY_j}$ and $e^{-\im\theta Z_iY_j/2}$ are denoted by black and gray arrows, respectively. The circuit has the same sequential structure with the whip circuit with phase transition (right panel) introduced in the main text. The red dot is the DAG source site. (\textbf{b}) Evaluating the expectation of local observable $Z_i Z_j$ in the whip circuit without phase transition. The red-highlighted arrows together with the edge $Z_iZ_j$ consist a cycle contributing to the expectation.} 
    \label{fig:Ising_withoutPT}
\end{figure}

Here we present the 2-d Ising whip circuit without phase transition, which is equivalent to the original sequential circuit proposed in Ref.~\cite{Chen_2024}. The circuit uses the $ZY$ rotations $e^{-\im\theta Z_iY_j}$ as the basic gates, and the rotations are applied sequentially starting from the source site of the DAG, as illustrated in the left panel of Fig.~\ref{fig:Ising_withoutPT}(\textbf{a}). This Ising whip circuit has comparable circuit depth and basic gates to the one with phase transition (right panel of Fig.~\ref{fig:Ising_withoutPT}(\textbf{a})). Additionally, the circuit also exactly prepares the ground state of the 2-d ferromagnetic Ising Hamiltonian $H=-\frac{1}{|\EE|}\sum_{\langle i,j\rangle} Z_iZ_j$, since the GHZ state is generated by $\theta=-\pi/4$~\cite{Chen_2024}
\begin{align}
    \ket{\phi_{\text{w}}^{\text{w/o}}(\theta)}=\left.\prod_{\langle i,j\rangle '} e^{-\im\theta Z_i Y_j}\ket{+}^{\otimes N}\right|_{\theta=-\pi/4}=\frac{1}{\sqrt{2}}(\ket{0}^{\otimes N}+\ket{1}^{\otimes N}).
\end{align}
where $\langle i,j\rangle '$ denotes the nearest-neighbor coupling restricted to the tree-like subgraph in the left panel of Fig.~\ref{fig:Ising_withoutPT}(\textbf{a}).  Beyond $\theta=-\pi/4$, the energy landscapes of $|\phi_{\text{w}}^{\text{w/o}}(\theta)\rangle$ and the one with phase transition $|\phi_{\text{w}}(\theta)\rangle$ are different for general $\theta$ values. In the following, we show that the energy density of $|\phi_{\text{w}}^{\text{w/o}}(\theta)\rangle$ is analytic, indicating that the circuit has no phase transition. 

For an $L\times L$ lattice with $|\EE|=2L(L-1)$ edges, the local $ZZ$ terms in the Hamiltonian are categorized into two kinds: $ZZ$ of the first kind has no $ZY$ rotation applied on the corresponding qubits, and of the second kind has a $ZY$ rotation applied, as illustrated by $Z_{i}Z_{j}$ and $Z_{i'}Z_{j'}$ in Fig.~\ref{fig:Ising_withoutPT}(\textbf{b}), respectively. The $ZZ$ expectation of the first kind only has one cycle contribution as illustrated by the red arrows in Fig.~\ref{fig:Ising_withoutPT}(\textbf{b}). Since the cycle length is $2\LL+2$, the expectation reads
\begin{align}
    \bra{\phi_{\text{w}}^{\text{w/o}}(\theta)} Z_iZ_j
\ket{\phi_{\text{w}}^{\text{w/o}}(\theta)} = (-\sin 2\theta)^{2\LL+1}=-(\sin 2\theta)^{2\LL+1}.
\end{align}
The $ZZ$ expectation of the second kind only has one local contribution, which gives
\begin{align}
    \bra{\phi_{\text{w}}^{\text{w/o}}(\theta)} Z_{i'}Z_{j'}
\ket{\phi_{\text{w}}^{\text{w/o}}(\theta)} = -\sin 2\theta,
\end{align}
and there are $L^2-1$ such $ZZ$ terms in the Ising Hamiltonian. Combining these two kinds of $ZZ$ terms, the Ising energy density is derived from the Hamiltonian expectation
\begin{equation}
    \begin{aligned}
    \bra{\phi_{\text{w}}^{\text{w/o}}(\theta)} H
\ket{\phi_{\text{w}}^{\text{w/o}}(\theta)}=&\frac{1}{2L(L-1)}[(L^2-1)\sin2\theta +(L-1)\sum_{\LL=1}^{L-1}(\sin2\theta)^{2\LL+1} ]\\
=&\frac{1}{2L}[(L+1)\sin2\theta +\sum_{\LL=1}^{L-1}(\sin2\theta)^{2\LL+1}].
\end{aligned}
\end{equation}
In the infinite volume limit $L\to\infty$, the energy density reads 
\begin{align}
    \lim_{L\to\infty}\bra{\phi_{\text{w}}^{\text{w/o}}(\theta)} H
\ket{\phi_{\text{w}}^{\text{w/o}}(\theta)}=\left\{ \begin{array}{ll}
 -1, & \textrm{if $\theta= -\pi/4+m\pi$;}\\
 1, & \textrm{if $\theta= \pi/4+m\pi$;}\\
 -\frac{1}{2}\sin2\theta & \textrm{otherwise,}
  \end{array} \right.
\end{align}
where $m\in\ZZ$. Thus, the energy density is analytic for all $\theta$ with only removable singularities at $\theta=-\pi/4+m\pi/2, m\in\ZZ$. Removable singularities usually indicate a crossover but not a phase transition. Thus, the whip circuit state $\ket{\phi_{\text{w}}^{\text{w/o}}(\theta)}$ has no phase transition.

\subsection{VQE performance}\label{sec:VQE-performance}
\begin{figure}
    \centering
    \includegraphics[width=0.48\textwidth]{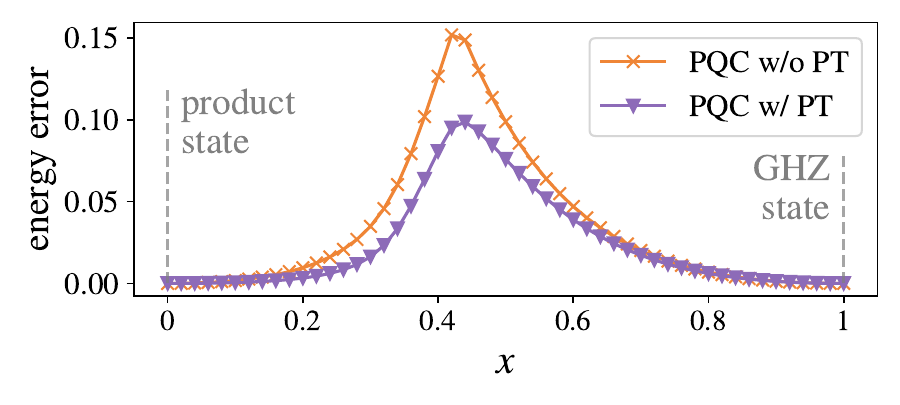}
    \caption{Relative error of the ground state energy $\left|(E_{\text{VQE}}-E_{\text{ED}}\right)/E_{\text{ED}}|$ obtained by VQE and exact diagonalization (ED) as a function of the coupling strength $x$ of the transverse-field Ising model. The errors using PQC with and without phase transition (PT) are marked by lower-triangles and crosses, respectively.}
    \label{fig:VQE_TFIM}
\end{figure}

The Ising whip circuit can be used as a VQE ansatz for high-accuracy state preparation. As a benchmark, we consider the transverse-field Ising model (TFIM), $H_{\text{TFIM}}(x) = -(1-x)\sum_{i\in\Lambda}X_i-x\sum_{\langle i,j\rangle}Z_iZ_j$, on a $4\times 4$ square lattice. 
We compare the performance of the Ising whip circuits with and without phase transition illustrated in Fig.~\ref{fig:Ising_withoutPT}. Both heuristic ans\"atze accurately prepare the ground states of $H_{\text{TFIM}}(x)$ at $x=0$ and $x=1$, i.e., the product state $\ket{+}^{\otimes N}$ and the GHZ state. Beyond these two ends, for each $x\in[0,1]$, variational parameters $\theta$ in both heuristic ans\"atze are optimized to the minimum energy expectation, and the corresponding approximate ground state energy $E_{\text{VQE}}$ is derived. The resulting relative error as a function of $x$ is plotted in Fig.~\ref{fig:VQE_TFIM}. We see that the PQC with phase transition consistently achieves lower energy errors for all $x\in[0,1]$. In particular, the worst-case relative error is reduced from $0.15$ to $0.10$, corresponding to an improvement of 33\%. Thus, a PQC with phase transition signals improved state-preparation accuracy as a VQE ansatz. 

\section{Phase diagram of the 2-d whip circuit}\label{sec:Proof of the equal order}

In this section, we describe the phase diagram of quantum states generated by the $2$-d Ising whip circuit. The phase diagram has a period of $2\pi$ because the $ZY$ rotations in the bulk of the lattice $e^{-\im \theta Z_iY_j/2}$ (with $\deg^-_j=2$)  and at the boundary of the lattice $e^{-\im \theta Z_iY_j}$ (with $\deg^-_j=1$) has the least common period of $2\pi$ (up to a global phase). In the following subsections, we focus on the representative states in two phases: (1) The states with the whip angle $\theta=0,\pi$; (2) The states with $\theta=\pm\pi/2$; and (3) the GHZ-type states at the critical points $\theta=\pm\pi/4, \pm 3\pi/4$. These representative states are summarized in the phase diagram Fig.~(6) of the main text.

\subsection{$\theta=0,\pi$}
The whip states at $\theta=0,\pi$ are product states. Because for $\theta=0$, the state is the initial state $\ket{\phi_{\text{w}}(0)}=\ket{+}^{\otimes N}$. For $\theta=\pi$, $ZY$ rotations at the upper boundary and in the bulk of the 2-d lattice respectively read
\begin{align}
    e^{-\im\pi Z_iY_j} = -1;\quad
    e^{-\im \pi Z_iY_j/2} = -\im Z_iY_j.
\end{align}
Since $Z$ and $Y$ flip the initial product state by $Z\ket{\pm} = \ket{\mp}$ and $Y\ket{\pm} = \pm \im\ket{\mp}$. The resulting whip state $\ket{\phi_{\text{w}}(\pi)}$ is a product of $\ket{+}$ and $\ket{-}$ at each qubit, where $\ket{\pm}$ depends on the parity of the $Z$ and $Y$ number applied on the specific qubit. $\ket{\phi_{\text{w}}(0)}$ and $\ket{\phi_{\text{w}}(\pi)}$ on a $4\times 4$ lattice is illustrated in Fig.~\ref{fig:whip_states}(\textbf{a}), corresponding to the two symmetry breaking phase characterized by the order parameter $\langle \partial X\rangle_{\theta}$.


\subsection{$\theta=\pm\pi/2$}
The two whip states at $\theta=\pm\pi/2$ are Clifford states, since the $ZY$ rotations at $\theta=\pm\pi/2$ are Clifford: the $ZY$ rotations with $\theta=\pm\pi/2$ at the upper boundary of the lattice are Clifford Pauli operators $-\im ZY$, and the rotations in the bulk of the lattice can be decomposed to the standard Clifford gates shown in Fig.~\ref{fig:whip_states}(\textbf{b}). Therefore, since the initial state $\ket{+}^{\otimes N}$ is a Clifford state, the whip states $\ket{\phi_{\text{w}}(\pm\pi/2)}$ are also Clifford. Although $\ket{\phi_{\text{w}}(0)}$ and $\ket{\phi_{\text{w}}(\pi)}$ are also Clifford, their Clifford stabilizers are completely different, leading to the distinguished entanglement and order parameter behavior shown in the main text.

\begin{figure}
    \centering
    \includegraphics[width=0.98\textwidth]{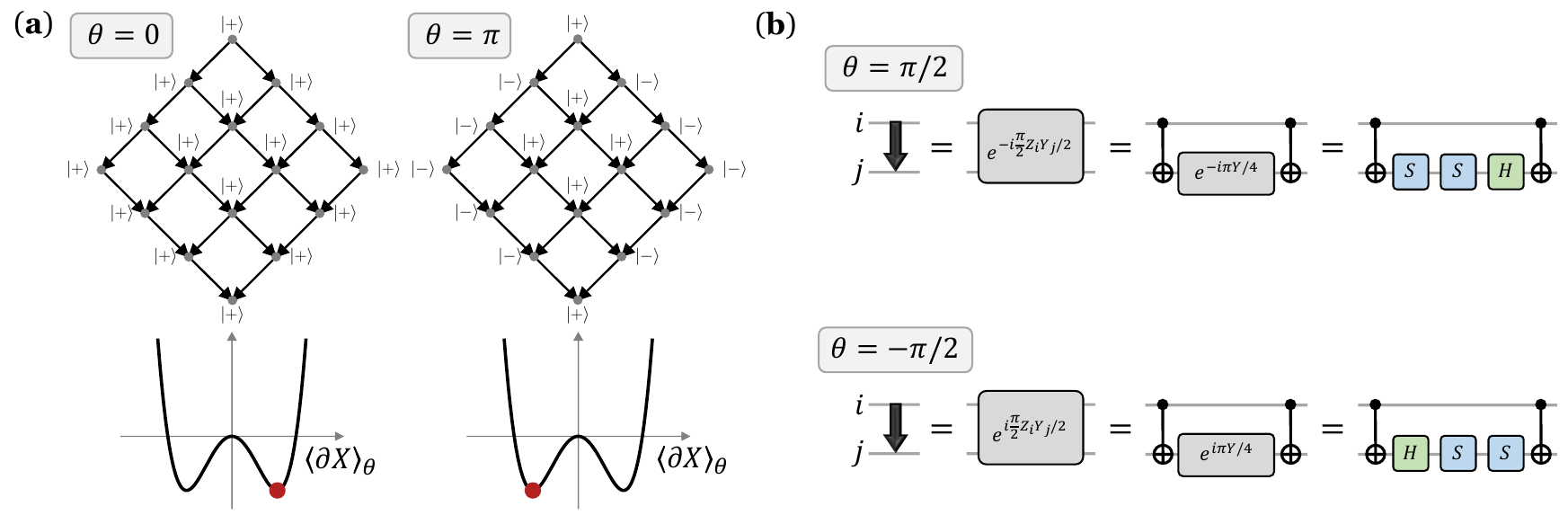}
    \caption{(\textbf{a}) Illustration of the product state $\ket{\phi_{\text{w}}(0)}$ and $\ket{\phi_{\text{w}}(\pi)}$ on the $4\times 4$ lattice, corresponding to the two symmetry breaking phases characterized by the order parameter $\langle \partial X\rangle_{\theta}$. (\textbf{b}) The $ZY$ rotations with the whip angle $\theta=\pm\pi/2$ can be decomposed into the standard Clifford gates $\{H, S, \text{CNOT}\}$. Thus, $\ket{\phi_{\text{w}}(\pm\pi/2)}$ are Clifford states.}
    \label{fig:whip_states}
\end{figure}

\subsection{$\theta=\pm\pi/4,\pm3\pi/4$}
In Supp. Note~\ref{sec:Ising whip circuit on arbitrary dimensions}, we show that $\ket{\phi_{\text{w}}(\pi/4)}$ and $\ket{\phi_{\text{w}}(-\pi/4)}$ are the anti-ferromagnetic Ising ground state and ferromagnetic ground state, respectively. Additionally, according to the symmetry transformation $T:\theta\rightarrow \theta+\pi$, we have 
\begin{equation}
    \begin{aligned}
    \ket{\phi_{\text{w}}(-3\pi/4)}&=e^{i\phi} \hat{T}^{-1}\ket{\phi_{\text{w}}(\pi/4)}=e^{i\phi} \hat{T}\ket{\phi_{\text{w}}(\pi/4)}=e^{i\phi} \ket{\phi_{\text{w}}(\pi/4)}\\
    \ket{\phi_{\text{w}}(3\pi/4)}&=e^{i\phi'}\hat{T}\ket{\phi_{\text{w}}(-\pi/4)}=e^{i\phi'}\ket{\phi_{\text{w}}(-\pi/4)}
\end{aligned}
\end{equation}
where $\phi$ and $\phi'$ are some irrelevant global phases. In the final step of the two formulae, we use the fact that both the anti-ferromagnetic and ferromagnetic GHZ states are invariant by applying the symmetry operator $\hat{T}=\prod_{i\in\BB}Z_i$. Thus, $\ket{\phi_{\text{w}}(-3\pi/4)}$ and $\ket{\phi_{\text{w}}(3\pi/4)}$ are also the anti-ferromagnetic and ferromagnetic ground states. These two states are respectively denoted by the blank and filled dots in the phase diagram Fig.~(6).

\section{Weight-adjustable loop ansatz with and without phase transition}\label{app:WALA}

In this note, we describe the weight-adjustable loop ansatz (WALA) with and without phase transition on the 3-dimensional $L\times L\times 2$ lattice illustrated in Fig.~\ref{fig:WALA-general}(\textbf{b}). The lattice has $(L-1)\times (L-1)$ cubes, and each cube is composed of 6 plaquettes and 12 links. Each link of the lattice is assigned with one qubit (black dot) initialized in the $\ket{0}$ state. For WALA with phase transition (PT), in each plaquette of the lattice with qubits $\{i,j,k,l\}$, we assign a plaquette rotation $e^{-\im\theta X_iX_jX_kY_l/N^Y_l}$ denoted by a circle oriented to the qubit $l$, as shown in Fig.~\ref{fig:WALA-general}(\textbf{c}). The plaquette rotation is realized by one $R_Y(\theta/N^Y_l) = e^{-\im\theta Y_l/N^Y_l}$ and $6$ CNOT gates. $N^Y_l$ is the total number of $R_Y$ rotations applied on the qubit $l$ in the WALA, which is similar to the role of $\deg^{-}_j$ in the Ising whip circuit.

\begin{figure}
    \centering
    \includegraphics[width=0.9\linewidth]{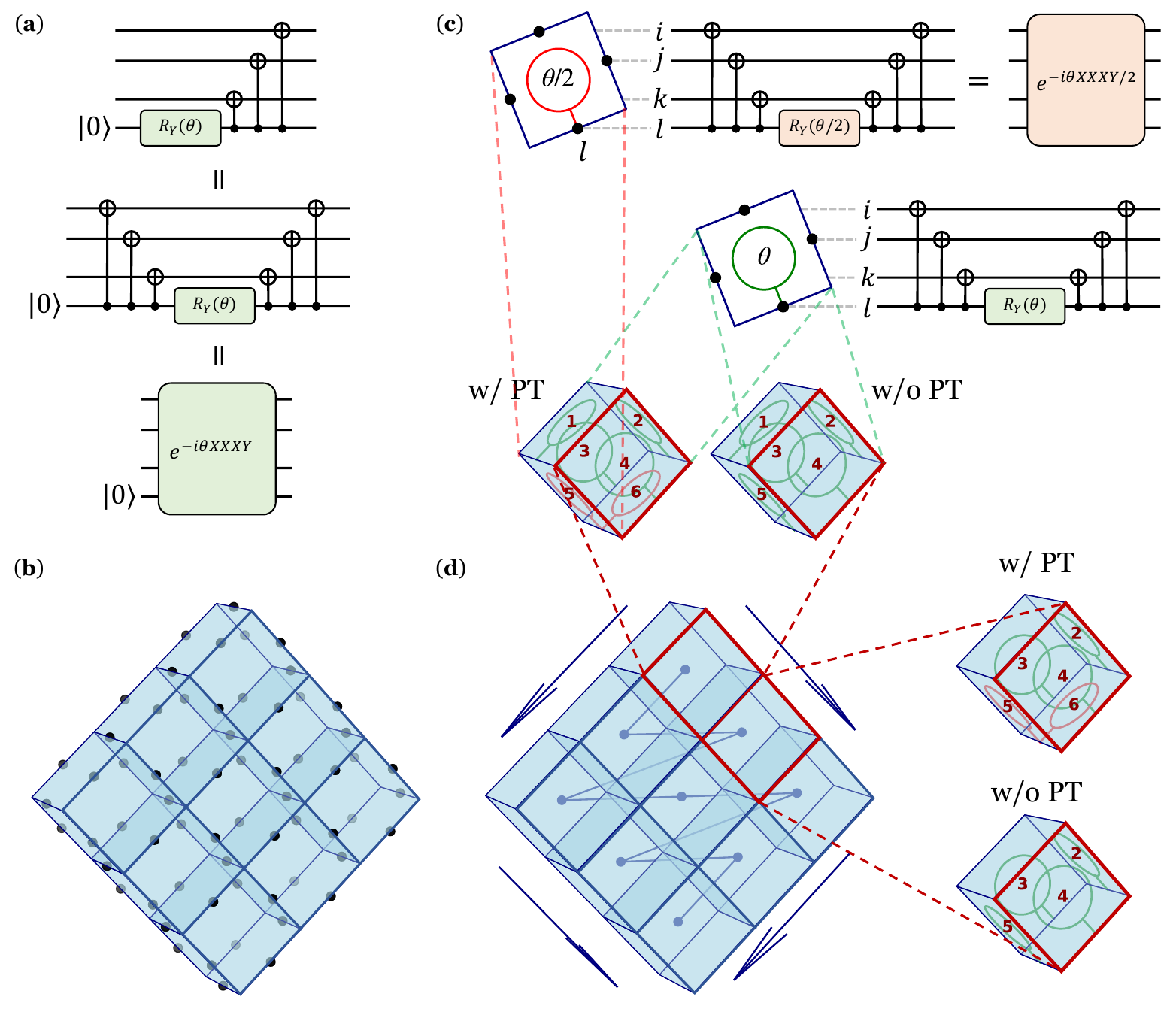}
    \caption{Weight-adjustable loop ansatz (WALA) with and without phase transition (PT). (\textbf{a}) The original WALA~\cite{PhysRevB.107.L041109,Cochran_2025} preparing the toric code state has the basic gate equivalent to the plaquette rotation $e^{-i\theta XXXY}$. (\textbf{b}) An $L\times L\times 2$ lattice of the $\ZZ_2$ gauge theory. Each link of the lattice is assigned by one qubit (black dot). (\textbf{c}) WALA in a single cube of the lattice apply plaquette rotations sequentially following the order \textbf{1}-\textbf{6} and \textbf{1}-\textbf{5} respectively for the WALA with and without PT. The red and green oriented cycles denote plaquette rotations $e^{-i\theta XXXY/2}$ and $e^{-i\theta XXXY}$, respectively. (\textbf{d}) WALA in the whole lattice apply plaquette rotations sequentially from the the top cube to the bottom one explicitly denoted by the broken line inside the lattice.}
    \label{fig:WALA-general}
\end{figure}

The plaquette rotation $e^{-\im\theta X_iX_jX_kY_l}$ is equivalent to the original WALA proposed in Ref.~\cite{PhysRevB.107.L041109}. The original WALA applies a $R_Y$ rotation to the qubit $l$ initialized to the state $\ket{0}$, then three CNOTs are applied to the remaining three qubits of the plaquette, as shown in the first row of Fig.~\ref{fig:WALA-general}(\textbf{a}). Since the qubit $l$ is initialized as $\ket{0}$, the circuit in the first row is equivalent to the second row, and thus the whole circuit is equivalent to the rotation $e^{-\im\theta XXXY}$. Additionally, the plaquette rotation is closely related to the imaginary time evolution of the plaquette operator $X_p^{\square} = XXXX$~\cite{PhysRevA.111.032612}, because
\begin{equation}
    \begin{aligned}
    e^{\tau XXXX}\ket{0000} &= [\cosh(\tau)+\sinh(\tau) XXXX]\ket{0000}=\cosh(\tau)\ket{0000}+\sinh(\tau) \ket{1111}\\
    &\propto [\sin(\theta)-i\cos(\theta) XXXY]\ket{0000}=e^{-\im\theta XXXY}\ket{0000}.
\end{aligned}
\end{equation}
Thus, WALA chooses the basic gates following the same idea as the Ising whip circuit.

The WALA state can be written as
\begin{align}
    \ket{\phi_{\mathbb{Z}_2}^{\text{w/}}(\theta)} \equiv \prod_{p=\{i,j,k,l\}} e^{-\im\theta X_iX_jX_kY_l/N^Y_l}\ket{0}^{\otimes N},
    \label{eq:WALA-state}
\end{align}
and the WALA state without PT is denoted by $\ket{\phi_{\mathbb{Z}_2}^{\text{w/o}}(\theta)}$ which is similar to $\ket{\phi_{\mathbb{Z}_2}^{\text{w/}}(\theta)}$ as we discuss later. In the following content, we describe the order of the product $\prod_{p=\{i,j,k,l\}}$ in Eq.~\eqref{eq:WALA-state}, and show that the exact toric code state 
\begin{align}
    \ket{\Psi_{\text{T}}}\equiv \prod_{p} \left(\frac{1}{\sqrt{2}}\II+\frac{1}{\sqrt{2}}X_p^{\square}\right) \ket{0}^{\otimes N}
\end{align}
is prepared exactly by choosing $\theta=\pi/4$ in $\ket{\phi_{\mathbb{Z}_2}^{\text{w/}}(\theta)}$.

\noindent \textbf{Product order in one cube}~---~The plaquette rotations are applied sequentially following the order $\mathbf{1}$-$\mathbf{6}$ and $\mathbf{1}$-$\mathbf{5}$ for WALA with and without PT respectively, as shown by the two cubes at the bottom of Fig.~\ref{fig:WALA-general}(\textbf{c}). These two orders ensure that the plaquette rotation has $R_Y$ applied directly to the $\ket{0}$ state, such that the toric code state can be correctly prepared at $\theta=\pi/4$~\cite{KITAEV20032,PhysRevB.107.L041109}. The only difference between WALA with and without PT is at the two bottom plaquettes of the two cubes in Fig.~\ref{fig:WALA-general}(\textbf{c}). In WALA without PT, the plaquette rotation with angle $\theta$ is applied on one of the two bottom plaquettes, whereas WALA with PT has two plaquette rotations with angle $\theta/2$ applied on both plaquettes. WALA without PT exactly prepares the toric code state $\ket{\Psi_{\text{T}}}$ at $\theta=\theta_c\equiv\pi/4$, as has been shown in Ref.~\cite{PhysRevB.107.L041109}, and WALA with PT $\ket{\phi_{\mathbb{Z}_2}^{\text{w/}}(\theta_c)}$ can be shown to be equal to $\ket{\phi_{\mathbb{Z}_2}^{\text{w/o}}(\theta_c)}$: Denote the quantum state after applying plaquette rotations $p\in [\mathbf{1},\mathbf{4}]\equiv\{\mathbf{1},\mathbf{2},\mathbf{3},\mathbf{4}\}$ as $\ket{\Psi_{[\mathbf{1},\mathbf{4}]}(\theta_c)}$, and the qubits of plaquette $\mathbf{5}$ and $\mathbf{6}$ are $\{i,j,k,l\}$ and $\{i',j',k',l\}$, respectively. Then we have the sequential plaquette rotations applied to $\ket{\Psi_{[\mathbf{1},\mathbf{4}]}(\theta_c)}$:
\begin{equation}
    \begin{aligned}
    \ket{\phi_{\mathbb{Z}_2}^{\text{w/}}(\theta_c)}&=e^{-\im\theta_c X_{i'}X_{j'}X_{k'}Y_{l}/2}e^{-\im\theta_c X_{i}X_{j}X_{k}Y_{l}/2}\ket{\Psi_{[\mathbf{1},\mathbf{4}]}(\theta_c)}\\
    &=e^{-\im\theta_c X_{i}X_{j}X_{k}Y_{l}/2}e^{-\im\theta_c X_{i'}X_{j'}X_{k'}Y_{l}/2}\ket{\Psi_{[\mathbf{1},\mathbf{4}]}(\theta_c)}\\
    &=e^{-\im\theta_c X_{i}X_{j}X_{k}Y_{l}/2}(\cos\frac{\theta_c}{2}-\ii\sin\frac{\theta_c}{2}X_{i'}X_{j'}X_{k'}Y_{l})\ket{\Psi_{[\mathbf{1},\mathbf{4}]}(\theta_c)}\\
    &=e^{-\im\theta_c X_{i}X_{j}X_{k}Y_{l}/2}(\cos\frac{\theta_c}{2}-\ii\sin\frac{\theta_c}{2}X_{i}X_{j}X_{k}Y_{l})\ket{\Psi_{[\mathbf{1},\mathbf{4}]}(\theta_c)}\\
    &=e^{-\im\theta_c X_{i}X_{j}X_{k}Y_{l}}\ket{\Psi_{[\mathbf{1},\mathbf{4}]}(\theta_c)}\\
    &=\ket{\phi_{\mathbb{Z}_2}^{\text{w/o}}(\theta_c)},
\end{aligned}
\end{equation}
where in the second line, we use the commutation relation $[X_{i'}X_{j'}X_{k'}Y_{l},X_{i}X_{j}X_{k}Y_{l}]=0$ between the two plaquette rotations both oriented to qubit $l$. In the forth line, we use the fact that $\ket{\Psi_{[\mathbf{1},\mathbf{4}]}(\theta_c)}$ is the stabilizer state of the plaquette operators $X_p^{\square}$ with $p\in [\mathbf{1},\mathbf{4}]$, and their product is also stabilizer
\begin{align}
    \ket{\Psi_{[\mathbf{1},\mathbf{4}]}(\theta_c)} = \prod_{p\in [\mathbf{1},\mathbf{4}]} X_p^{\square}\ket{\Psi_{[\mathbf{1},\mathbf{4}]}(\theta_c)}=X_iX_j X_kX_{i'}X_{j'} X_{k'}\ket{\Psi_{[\mathbf{1},\mathbf{4}]}(\theta_c)},
\end{align}
where $\{i,j,k,i',j',k'\}$ is exactly the boundary of the four plaquettes $[\mathbf{1},\mathbf{4}]$. Therefore, both $\ket{\phi_{\mathbb{Z}_2}^{\text{w/}}(\theta_c)}=\ket{\phi_{\mathbb{Z}_2}^{\text{w/o}}(\theta_c)}=\ket{\Psi_{\text{T}}}$ is the toric code state for $\theta_c=\pi/4$, whereas these two states are not equal for general $\theta\neq\theta_c$.
\\\\
\noindent \textbf{Product order in the lattice}~---~For the WALA on the whole $L\times L\times 2$ lattice, quantum gates in unit of cubes are applied sequentially to the initial state from the top cube to the bottom cube as shown by the four arrows in Fig.~\ref{fig:WALA-general}(\textbf{d}). The broken line in the figure connecting the cube centers explicitly denotes the order of applying plaquette rotations, which is similar to the order of applying $ZY$ rotations in the Ising whip circuit. This cube order guarantees that each plaquette rotation oriented to qubit $l$ inside all cubes has the $R_Y$ rotation applied directly to the initial $\ket{0}_l$ state. Additionally, since each plaquette in the lattice should have the plaquette rotation applied only once, a later cube does not apply the plaquette rotation if a previous one has already applied it. For example, two cubes in the right panel of Fig.~\ref{fig:WALA-general}(\textbf{d}) only apply plaquette rotations $[\mathbf{2},\mathbf{6}]$ and $[\mathbf{2},\mathbf{5}]$ for WALA with and without PT, respectively.

Combining the plaquette order in one cube and the cube order in the lattice gives the explicit product order of $\ket{\phi_{\mathbb{Z}_2}^{\text{w/}}(\theta)}$ in Eq.~\eqref{eq:WALA-state}, and $\ket{\phi_{\mathbb{Z}_2}^{\text{w/o}}(\theta)}$ with the similar form also holds a well-defined product order. Our previous argument on the equivalence of $\ket{\phi_{\mathbb{Z}_2}^{\text{w/}}(\theta_c)}$ and $\ket{\phi_{\mathbb{Z}_2}^{\text{w/o}}(\theta_c)}$ in one cube still holds for the $L\times L\times 2$ lattice, i.e, the toric code state on the $L\times L\times 2$ lattice is prepared exactly by
\begin{align}
    \ket{\Psi_{\text{T}}} = \ket{\phi_{\mathbb{Z}_2}^{\text{w/}}(\theta_c)} = \ket{\phi_{\mathbb{Z}_2}^{\text{w/o}}(\theta_c)}.
    \label{eq:WALA-toric-code-state}
\end{align}

The toric code state is the ground state of $H_{\ZZ_2}(x) = -(1-x)\sum_{p} X_{p}^{\square}-x\sum_{l} Z_l-\sum_{v} Z_{v}^{+}$ with $x=0$, since $\ket{\Psi_{\text{T}}}$ is the stabilizer state of both $X_{p}^{\square}$ and $Z_{v}^{+}$ for all plaquettes $p$ and vertex $v$, i.e,
\begin{align}
    \bra{\Psi_{\text{T}}}X_{p}^{\square}\ket{\Psi_{\text{T}}} = \bra{\Psi_{\text{T}}}Z_{v}^{+}\ket{\Psi_{\text{T}}}=1.
\end{align}
Thus, $X_{p}^{\square}$ and $Z_{v}^{+}$ expectation values of the WALA state $\ket{\phi_{\mathbb{Z}_2}^{\text{w/o}}(\theta_c)}$ and $\ket{\phi_{\mathbb{Z}_2}^{\text{w/}}(\theta_c)}$ are $1$ by Eq.~\eqref{eq:WALA-toric-code-state}. In the following subsection, we study the $\theta$ dependence of the plaquette expectation
\begin{align}
    \langle X_{p}^{\square}\rangle ^{\text{w/o}}_{\theta}  \equiv \bra{\phi_{\mathbb{Z}_2}^{\text{w/o}}(\theta)} X_{p}^{\square}\ket{\phi_{\mathbb{Z}_2}^{\text{w/o}}(\theta)}.
\end{align}
The plaquette expectation of WALA with PT $\quad \langle X_{p}^{\square}\rangle_{\theta}  \equiv \bra{\phi_{\mathbb{Z}_2}^{\text{w/}}(\theta)} X_{p}^{\square}\ket{\phi_{\mathbb{Z}_2}^{\text{w/}}(\theta)}$ is evaluated numerically in the End Note. The $\theta$ dependence of the $Z_{v}^{+}$ expectation is trivial, since the plaquette rotations $e^{-\im\theta X_iX_jX_k Y_l}$ commute with $Z_{v}^{+}$, such that $\langle Z_{v}^{+}\rangle ^{\text{w/o}}_{\theta}=\langle Z_{v}^{+}\rangle ^{\text{w/}}_{\theta}=1, \forall \theta\in\mathbb{R}$.

\subsection{No non-analyticity in WALA without phase transition}

\begin{figure}
    \centering
    \includegraphics[width=0.94\linewidth]{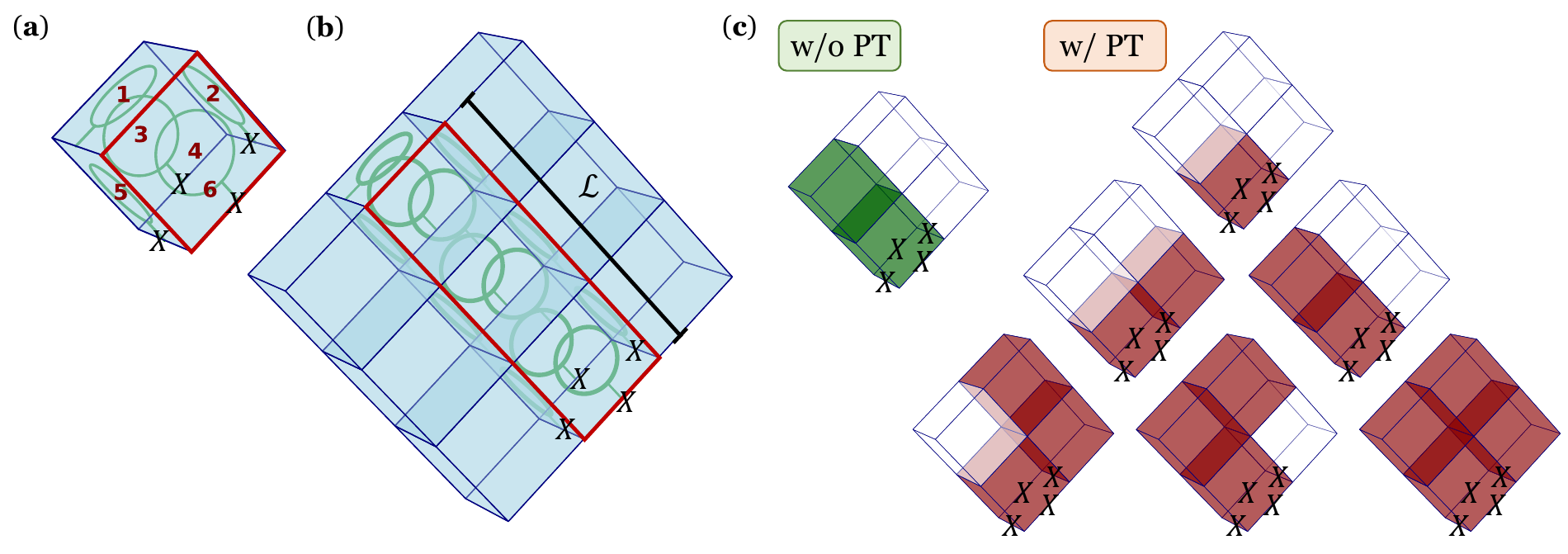}
    \caption{Plaquette expectation of WALA without PT in (\textbf{a}) a single cube; (\textbf{b}) a column of cubes with height $\LL$. Each green oriented cycle denote a plaquette rotation $e^{-i\theta XXXY}$. (\textbf{c}) Non-local sphere contributions to the plaquette expectation of WALA without (left panel) and with (right panel) phase transition.}
    \label{fig:WALA_wo_PT}
\end{figure}

In this section, we analytically evaluate the plaquette expectation of $\ket{\phi_{\mathbb{Z}_2}^{\text{w/o}}(\theta)}$ in the infinite volume limit $L\to \infty$. Similar to the Ising whip circuit without PT studied in Supp. Note~\ref{sec:Whip circuits without phase transition}, we show that the expectation has no non-analyticity, indicating that $\ket{\phi_{\mathbb{Z}_2}^{\text{w/o}}(\theta)}$ has no phase transition. 

We firstly evaluate the expectation of the plaquette operator in the cube illustrated at the bottom of Fig.~\ref{fig:WALA-general}(\textbf{c}). Due to the product order $\mathbf{1}$-$\mathbf{5}$, the plaquette operator on $p\in[\mathbf{1},\mathbf{5}]$ can be derived by
\begin{align}
    \bra{0}^{\otimes N} e^{\im\theta X_i X_j X_k Y_l} X_iX_jX_k X_l e^{-\im\theta X_i X_j X_k Y_l}\ket{0}^{\otimes N}=\sin 2\theta,
\end{align}
which contains only the local contribution. The non-local contribution can be observed for the plaquette operator on $p=\textbf{6}$, as illustrated in Fig~\ref{fig:WALA_wo_PT}(\textbf{a}). The plaquette operator of $\bos{6}$ is transformed by all plaquette rotations of $p\in[\bos1,\bos5]$. Thus, the expectation reads
\begin{align}
     \langle X_{\bos 6}^{\square}\rangle ^{\text{w/o}}_{\theta,1}=(\sin 2\theta)^5.
\end{align}
This non-local contribution comes from the closed sphere covering the whole cube~---~similar to the non-local contribution by the closed loops in the Ising whip circuit. The exponent $5$ can be derived directly from the \textit{surface area} of the sphere, which is $|\SS|=6$ in unit of the plaquette area in this single cube example.

Generalizing this example to the $L\times L\times 2$ lattice, the only non-local contribution comes from  $X_{\bos{6}}^{\square}$ of all cubes, and each non-local contribution corresponds to a closed sphere in the lattice: The closed sphere of $X_{\bos{6}}^{\square}$ is consisted of cubes at the same column of $X_{\bos{6}}^{\square}$, as indicated by the red rectangle in Fig~\ref{fig:WALA_wo_PT}(\textbf{b}). The area of the closed sphere is $|\SS| = 4\LL+2$. Thus, the plaquette expectation reads
\begin{align}
\langle X_{\bos 6}^{\square}\rangle ^{\text{w/o}}_{\theta,\LL}=(\sin 2\theta)^{4\LL+1}.
\label{eq:plaquette-contribution}
\end{align}
The non-analytic part of the $\mathbb{Z}_2$ gauge energy density $\langle H_{\ZZ_2}\rangle_{\theta}$ is proportional to the averaged plaquette expectation within the same column. In the infinite volume limit $L\to\infty$, the averaged plaquette expectation reads
\begin{align}
    \lim_{L\to\infty}\frac{1}{L}\sum_{\LL=1}^{L}\langle X_{\bos 6}^{\square}\rangle ^{\text{w/o}}_{\theta,\LL} =\left\{ \begin{array}{ll}
 \lim_{L\to\infty}\frac{(\sin 2\theta)^{5}}{L}\frac{1-(\sin 2\theta)^{4L}}{1-(\sin 2\theta)^{4}}=0, & \textrm{if $\theta\neq \theta_c$;}\\
 1, & \textrm{if $\theta= \theta_c$,}
  \end{array} \right.
\end{align}
i.e., the energy density only has a removable singularity at $\theta=\theta_c = \pi/4+m\pi/2, m\in\ZZ$. Similar to the Ising whip circuit without PT, $\ket{\phi_{\mathbb{Z}_2}^{\text{w/o}}(\theta_c)}$ has no phase transition, which results from only one non-local sphere contribute to the plaquette expectation in Eq.~\eqref{eq:plaquette-contribution}.

In contrast, WALA with PT introduce more sphere contributions to the plaquette expectation, leading to the non-analyticity in the plaquette expectation $\langle X_{p}^{\square}\rangle_{\theta}$ shown in Appendix A of the End Note. These sphere contributions appear due to the additional plaquette rotation $e^{-\im\theta XXXY/2}$ at $p=\bos 6$ in Fig~\ref{fig:WALA-general}(\textbf{c}), such that the sphere can be closed across the whole lattice. Fig~\ref{fig:WALA_wo_PT}(\textbf{c}) illustrates the six spheres contributing to the plaquette expectation $\langle X_{\bos 6}^{\square}\rangle ^{\text{w/}}_{\theta,\LL}$. The number of these spheres is much more than the single sphere contribution of $\langle X_{\bos 6}^{\square}\rangle ^{\text{w/o}}_{\theta,\LL}$ in the left panel of Fig~\ref{fig:WALA_wo_PT}(\textbf{c}), as required by the non-analyticity mechanism described in the main text.

\section{Classical simulation of the 2-d Ising whip circuit}
In this note, targeting on the 2-d Ising whip circuit, we present the simulation time complexity using the Pauli propagation with the early-evaluation strategy and the tensor network method. We show the efficiency of classically simulating the 2-d Ising whip circuit, and the difficulty of generalizing these methods to the 3-d and higher dimensional Ising whip circuit.

\subsection{Pauli propagation time complexity}

The exponential time complexity of the the naive Pauli propagation method shown in Note~\ref{sec:Time complexity of evaluating $ZZ$ expectation} comes from the exponential number of Pauli paths. Here, we propose an efficient simulation method that avoids evaluating each Pauli path individually. The initial product state $|\phi_0\rangle=|+\rangle^{N}$ has the property that $\langle \phi_0 |\sigma_I|\phi_0\rangle=0$ for any Pauli string $\sigma_I$ containing at least one local $Y$ or $Z$. 
Owing to the sequential structure of the whip circuit, this means we can remove any Pauli strings from the simulation \textit{early} if it contains a local $Y_j$ or a $Z_j$ on any site $j$ that will no longer be transformed by any remaining gates in the rest of the circuit. This procedure amounts to taking the expectation $\bra{+}X_j\ket{+}=\bra{+}\II_j \ket{+}=1$, $\bra{+}Y_j\ket{+}=\bra{+}Z_j \ket{+}=0$ in advance at the site $j$ that we remove, and update the remaining coefficients accordingly. 

Concretely, to evaluate the expectation $\langle Z_iZ_j\rangle_{\theta}$ with nearest-neighbor $i=(0,1),j=(0,0)$ shown in Fig.~\ref{fig:loops_analytical_ZZ}(\textbf{a}), we divide the Ising whip circuit on an $L\times L$ lattice into layers of $ZY$ rotations 
\begin{align}
    \ket{\phi_{\text{w}}(\theta)} = U_{1}U_{2}\cdots U_{\LL} \cdots  U_{2L-2} \ket{+}^{\otimes N},
\end{align}
where $U_{\LL}$ denotes $ZY$ rotations in a horizontal layer in Fig.~\ref{fig:loops_analytical_ZZ}(\textbf{a}). For a local $Z$ operator at a site $(x,y)$, as illustrated in Fig.~\ref{fig:loops_analytical_ZZ}(\textbf{b}), the action of the layer $U_{\LL}$ reads
\begin{align}
    U_{\LL}^{\dagger}Z_{x,y}U_{\LL} = \cos^2(\theta) Z_{x,y} -  \cos(\theta)\sin(\theta) X_{x,y}Z_{x,y+1}  - \cos(\theta)\sin(\theta) X_{x,y}Z_{x+1,y} + \sin(\theta)^2 Z_{x+1,y}Z_{x,y}Z_{x,y+1},
    \label{eq:UZU-conjugation}
\end{align}
After applying $U_{\LL}$, the site $(x,y)$ is no longer acted on by any further $ZY$ rotations. Thus, the first and the fourth term are removed by the early evaluation $\bra{+}Z_{(x,y)} \ket{+}=0$, whereas the second and the third term are retained due to $\bra{+}X_{(x,y)} \ket{+}=1$, leaving a single $Z_{x,y+1}$ and $Z_{x+1,y}$ for further propagation. The effect of this kind of transformation is illustrated as 
\begin{center}
    \includegraphics[width=0.5\textwidth]{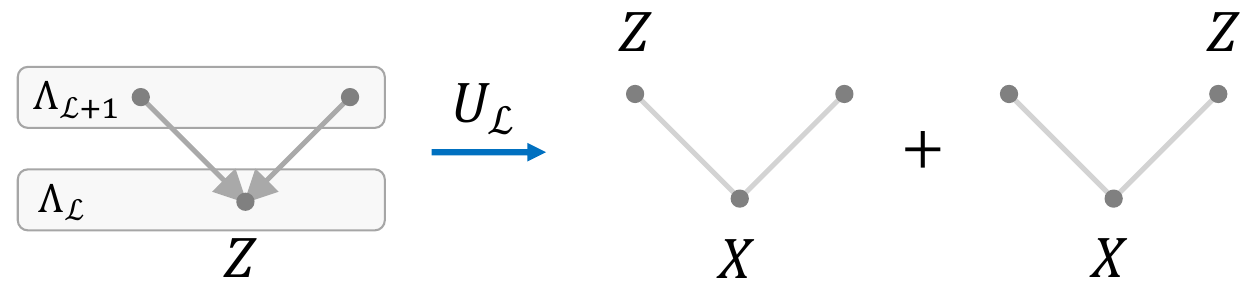}
\end{center}
where the gray arrow denotes the $ZY$ rotation $e^{-\ii\theta ZY/2}$. A key property of this transformation is that one layer $U_{\LL}$ transforms a local $Z_j$ with $j\in \Lambda_{\LL}$ into the linear combination of only \textit{local} $Z_{j'}$s with $j'\in \Lambda_{\LL+1}$, where $\Lambda_{\LL}$ denotes sites at the lower side of the layer $U_{\LL}$, as illustrated by Fig.~\ref{fig:Lambda_lattice}. This property inspires a modified approach of Pauli propagation with the \textit{early-evaluation} strategy to simulate the sequential transformation of $Z_{0,1}Z_{0,0}$. 

\begin{figure}[b]
    \centering
    \includegraphics[width=0.45\linewidth]{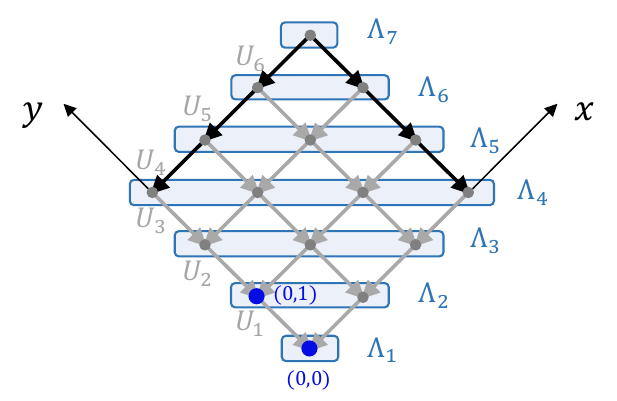}
    \caption{The modified Pauli propagation with the early-evaluation strategy to calculate the expectation value $\langle Z_{0,0}Z_{0,1}\rangle_{\theta}$ on an $L\times L= 4\times 4$ lattice. $\Lambda_\LL$ with $\LL\in\{1,\ldots,7\}$ is the set of sites, and $U_\LL$ with $\LL\in\{1,\ldots,6\}$ is the set of $ZY$ rotations in the $\LL$th layer.}
    \label{fig:Lambda_lattice}
\end{figure}

This modified Pauli propagation maintains a linear combination of two-body Pauli-$Z$ strings
\begin{align}
    \overline{ZZ}(\LL)\equiv \left\{ \begin{array}{ll}
 Z_{0,1}Z_{0,0} &\text{if } \LL=1;\\
\sum_{i,j\in\Lambda_{\LL};~ i \text{ is left to } j} c_{ij}^{(\LL)} Z_i Z_j+c_{\LL}&\text{if } 1<\LL\leq 2L-1,
  \end{array} \right.
  \label{eq:ZZbar-definition}
\end{align}
where $c_{ij}^{(\LL)}, c_{\LL}$ are real coefficients as functions of $\theta$. To start with, the first layer $U_1$ transforms $\overline{ZZ}(1)$ to $\overline{ZZ}(2)$, which reads
\begin{align}
    \overline{ZZ}(2) = Z_{0,1}\times(-\cos\theta\sin\theta Z_{0,1}-\cos\theta\sin\theta Z_{1,0}) = -\cos\theta\sin\theta -\cos\theta\sin\theta Z_{0,1}Z_{1,0},
\end{align}
where the bracket in the second term is derived from $U_{1}^{\dagger}Z_{0,0}U_{1}$, and we have used Eq.~\eqref{eq:UZU-conjugation}. Thus, we have $c_{(0,1);(1,0)}^{(1)}= -\cos\theta\sin\theta$ and $c_1 = -\cos\theta\sin\theta$. Then, for $U_{2}$, both $Z_{1,0}$ and $Z_{0,1}$ are updated using Eq.~\eqref{eq:UZU-conjugation}, and so on. Until the last step, $\overline{ZZ}(2L-2)$ reads 
\begin{align}
    \overline{ZZ}(2L-2) = c^{(2L-2)}_{i'j'}Z_{i'}Z_{j'} +c_{(2L-2)},
\end{align}
where $i'$ and $j'$ denote the site $i' = (L-2, L-1), j'=(L-1, L-2)$. The last layer $U_{2L-2}=e^{-\ii\theta Z_{k'}Y_{j'}} e^{-\ii\theta Z_{k'}Y_{i'}}, k'=(L-1,L-1)$ transforms $Z_{i'}Z_{j'}$ as
\begin{equation}
    \begin{aligned}
    &U_{2L-2}^{\dagger} Z_{i'}Z_{j'}U_{2L-2}= \cos^2(2\theta)Z_{i'} Z_{j'}-\cos(2\theta)\sin(2\theta)Z_{i'}Z_{k'}X_{j'}-\cos(2\theta)\sin(2\theta)Z_{j'}Z_{k'}X_{i'}+\sin^2(2\theta)X_{i'} X_{j'}.\nonumber
\end{aligned}
\end{equation}
Here, only the last term does not vanish by the initial product state $\ket{+}_{i'}\otimes \ket{+}_{j'}$, which can be illustrated by 
\begin{center}
    \includegraphics[width=0.3\textwidth]{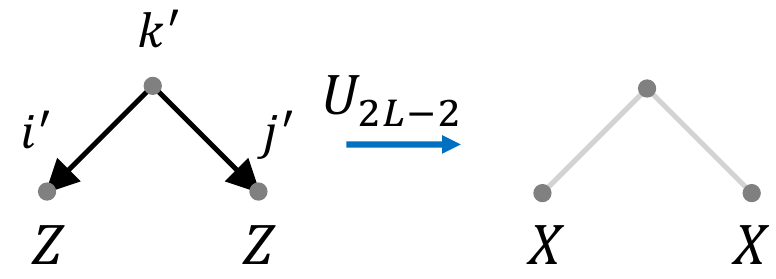}
\end{center}
where the black arrows denote the $ZY$ rotations $e^{-\ii\theta ZY}$. In other words, $\overline{ZZ}(2L-2)$ is transformed by $U_{2L-2}$ as 
\begin{align}
    \overline{ZZ}(2L-2) \xrightarrow{U_{2L-2}} \overline{ZZ}(2L-1) =c_{(2L-1)}= c^{(2L-2)}_{i'j'}\sin^2(2\theta)+c_{(2L-2)} = \langle Z_{0,0}Z_{0,1}\rangle_{\theta}.
\end{align}
We see $\overline{ZZ}(2L-1)$ contains only the coefficient $c_{(2L-1)}$, which is equal to $\langle Z_{0,0}Z_{0,1}\rangle_{\theta}$.

The time complexity of evaluating $\langle Z_{0,0}Z_{0,1}\rangle_{\theta}$ by this procedure grows polynomially with the lattice size $L$. To see this, note that each transformation from $\overline{ZZ}(\LL)$ to $\overline{ZZ}(\LL+1)$ in $U_{\LL}$ requires time proportional to the number of $ZZ$ terms in $\overline{ZZ}(\LL)$. 
According to the definition in Eq.~\eqref{eq:ZZbar-definition}, the number of $ZZ$ term in $\overline{ZZ}(\LL)$ is ${|\Lambda_{\LL}|\choose 2}$, where $|\Lambda_{\LL}|$ is the number of sites in the $\LL$th layer. On an $L\times L$ lattice, as illustrated by Fig.~\ref{fig:Lambda_lattice}, $|\Lambda_{\LL}|$ reads
\begin{align}
    |\Lambda_{\LL}|=\left\{ \begin{array}{ll}
 \LL & \textrm{if $1\leq\LL\leq L$;}\\
 2L-\LL & \textrm{if $L<\LL\leq 2L-2$.}
  \end{array} \right.
\end{align}
This is also proportional to the number of individual transformations in $U_\LL$, such that the time complexity of Pauli propagation with the early-evaluation strategy is
\begin{align}
T_{\text{PP}}\sim\sum_{\LL=1}^{2L-2}|\Lambda_\mathcal{L}|\begin{pmatrix}
        |\Lambda_{\LL}|\\2
    \end{pmatrix}
    = \sum_{\LL=1}^{L}L\begin{pmatrix}
        L\\2
    \end{pmatrix}+\sum_{\LL=L+1}^{2L-2}(2L-\LL)\begin{pmatrix}
        2L-\LL\\2
    \end{pmatrix}
    \sim \sum_{\LL=1}^{L}\LL^3\sim \OO(L^4).
\end{align}
In Fig.~\ref{fig:orqa_scaling} we show the wall-clock time and number of Pauli strings of the above modified Pauli propagation without any truncation performed using the ORQA-framework~\cite{orqa}. We see that the number of Pauli strings grows proportionally to $L^2$, and the corresponding time complexity is consistent with the predicted scaling of $\OO(L^4)$.

\begin{figure}
    \centering
    \includegraphics[width=0.6\linewidth]{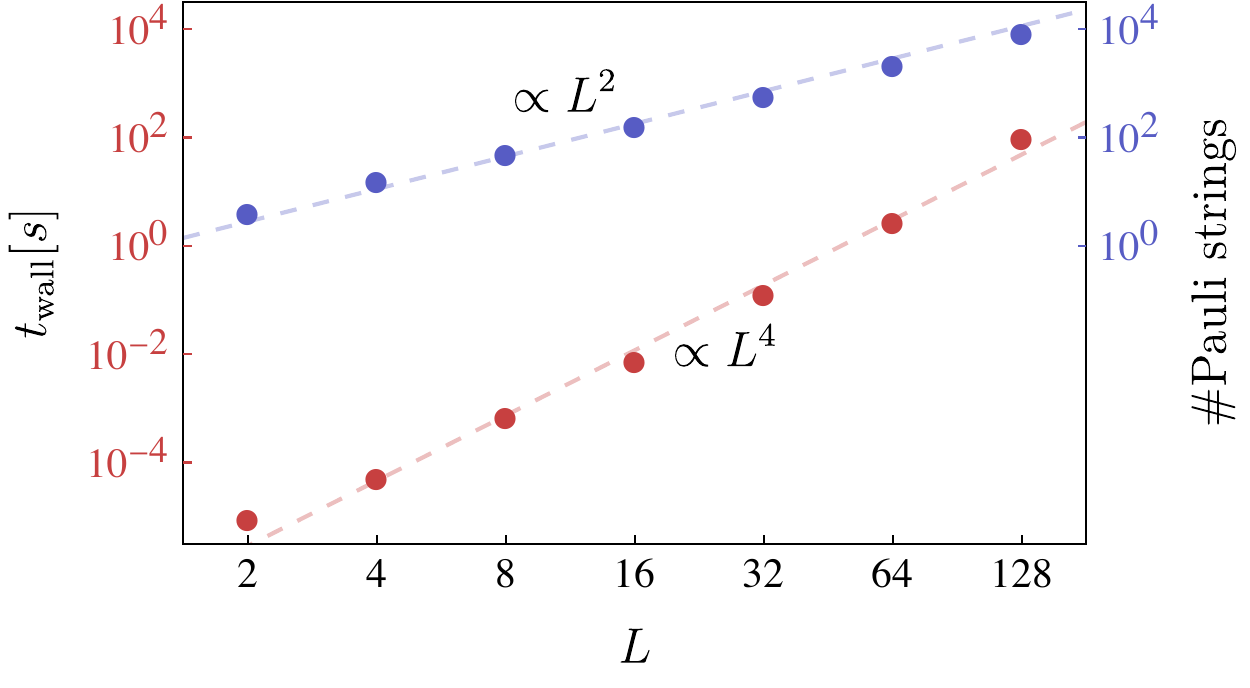}
    \caption{Overall wall-clock time in seconds of --- and the maximum number of Pauli strings retained in --- our Pauli propagation calculations of the $ZZ$ expectation for the 2-d whip circuit on the $L\times L$ lattice. We show these quantities in red and blue, respectively. Here $\theta=\pi/4$, but we find no difference in wall-time or number of Pauli strings for other choices of $\theta$, except at the trivial points $\theta=m\pi/2$, $m\in\mathbb{Z}$. The dashed lines corresponds to the power-laws $\propto L^4$ (red) and $\propto L^2$ (blue).}
    \label{fig:orqa_scaling}
\end{figure}

\vspace{2mm}
\noindent \textbf{Application to the 3-d Ising whip circuit}~---~The efficiency of the modified Pauli propagation method with the early-evaluation strategy can not straightforwardly be generalized to the 3-d Ising whip circuit. In the 2-d case, a local $Z_j$ is transformed by at most two $ZY$ rotations. Thus, according to Eq.~\eqref{eq:UZU-conjugation}, $Z_j$ with $j\in \Lambda_{\LL}$ is transformed to another local $Z_j'$ with $j'\in \Lambda_{\LL+1}$, such that the two-body Pauli-$Z$ linear combination is transformed to another two-body linear combination as
\begin{align}
    \overline{ZZ}(\LL)\xrightarrow{U_{\LL}}\overline{ZZ}(\LL+1).
\end{align}
However, for the 3-d Ising whip circuit, a local $Z_j$ is transformed by at most three $ZY$ rotations. In this case, a calculation similar to Eq.~\eqref{eq:UZU-conjugation} gives the following transformation
\begin{center}
    \includegraphics[width=0.7\textwidth]{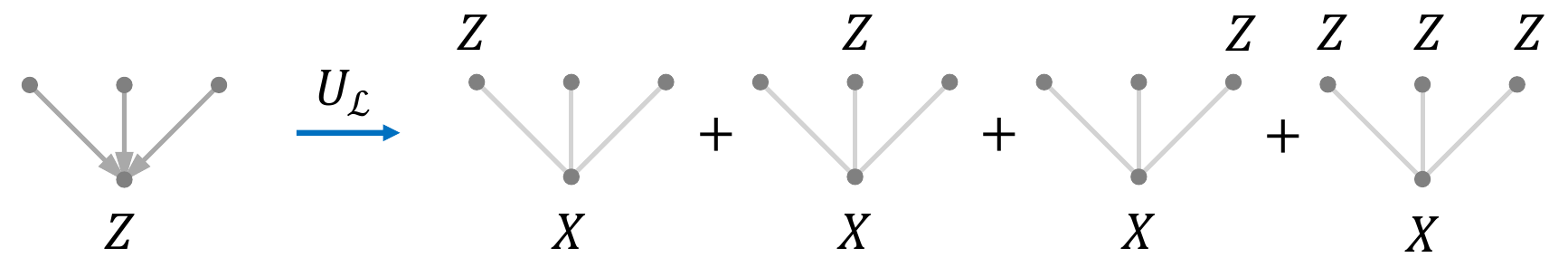}
\end{center}
The first three terms still transform a local $Z$ into another local $Z$. However, the last term extends the local $Z$ into the tensor product of three $Z$s. Thus, for the 3-d Ising whip circuit, starting from a two-body Pauli-$Z$ linear combination, the transformation by $U_{\LL}$ reads
\begin{align}
    \overline{Z^{\otimes 2}}(\LL)\xrightarrow{U_{\LL}}\overline{Z^{\otimes 6}}(\LL+1)\xrightarrow{U_{\LL+1}}\overline{Z^{\otimes 18}}(\LL+2)\cdots,
\end{align}
and the number of terms in these linear combinations grows exponentially with $\LL$. Thus, the Pauli propagation retaining only the two-body Pauli-$Z$ strings does not work for the 3-d Ising whip circuit. Similar difficulties are encountered in the higher dimensional Ising whip circuits with more than three $ZY$ rotations applied to a local $Z$.

\subsection{Tensor network time complexity}\label{sec:Tensor network time complexity}

\begin{figure}
    \centering
    \includegraphics[width=\linewidth]{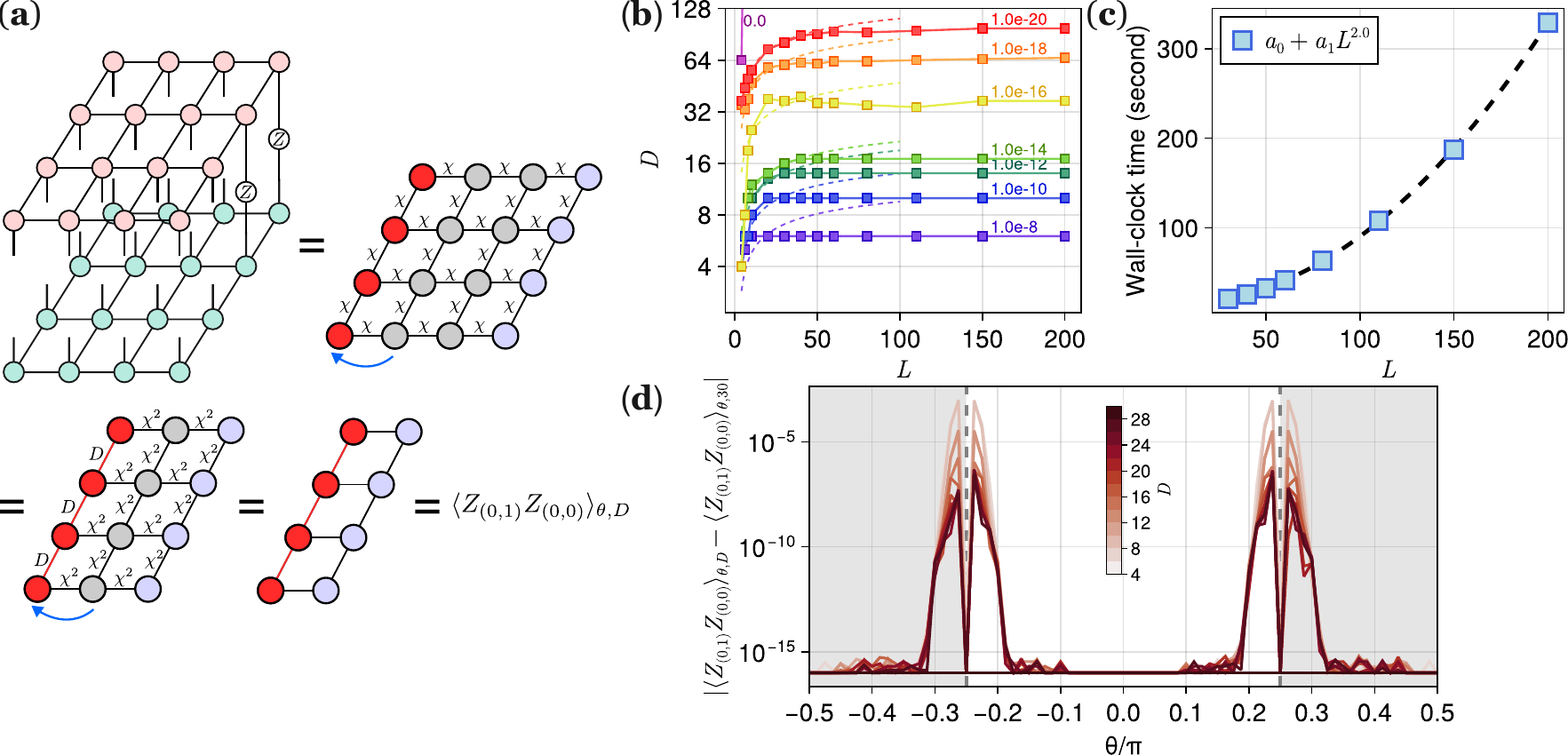}
    \caption{(a) Tensor-network contraction scheme for evaluating $\langle Z_{(0,1)}Z_{(0,0)}\rangle_{\theta,D}$ at arbitrary $\theta$ values, utilizing the MPO--MPS multiplication method.
    (b) Maximum bond dimension $D$ as a function of system size $L$ for different truncation cutoffs during tensor-network contraction; dashed lines represent the curve fitting for $\log L$. Here $D$ takes the maximum value at $\theta=0.2\pi$ and $0.225\pi$.
    (c) Wall-clock time required to evaluate $\langle Z_{(0,1)}Z_{(0,0)}\rangle_{\theta,D}$ near the critical point versus $L$ for $\theta=0.225\pi$ and $0.2375\pi$. Maximum bond dimension is fixed as $D=32$.
    (d) Difference between the expectation values $\langle Z_{(0,1)}Z_{(0,0)}\rangle_{\theta,D}$ with bond dimension $4\leq D<30$ and $\langle Z_{(0,1)}Z_{(0,0)}\rangle_{\theta,30}$. The lattice size is $L\times L=100\times 100$. The difference smaller than $10^{-16}$ is regarded as $10^{-16}$ in avoid of round-off errors at machine precision.}
    \label{fig:tn-contraction}
\end{figure}

We numerically study the time complexity of projected entangled pair states (PEPS) to evaluate the expectation value $\langle Z_{(0,1)}Z_{(0,0)}\rangle_{\theta}$, with $ Z_{(0,1)}Z_{(0,0)}$ on the 2-d coordinate. The initial state $\ket{\phi_0}=\ket{+}^{\otimes N}$ is represented as PEPS with a bond dimension $\chi=1$.
After applying the whip circuit, the resulting state $\ket{\phi_{\text{w}}(\theta)}$ is described by a PEPS with a bond dimension $\chi=2$.
The expectation value $\langle Z_{(0,1)}Z_{(0,0)}\rangle_{\theta}$ can be expressed as two local $Z$ tensors sandwiched by the bra and ket PEPS representations of $\ket{\phi_{\text{w}}(\theta)}$, as illustrated in Fig.~\ref{fig:tn-contraction}(\textbf{a}).
Contracting physical indices on each lattice site leads to a 2-d tensor network with an effective bond dimension $\chi^2=4$, covering the whole square lattice.

To perform the contraction stably and efficiently, we reinterpret the left-most column of tensors in Fig.~\ref{fig:tn-contraction}(\textbf{a}) as a MPS and the right-adjacent column as a MPO .
We then recursively apply the MPO--MPS multiplication using the density-matrix scheme \cite{ITensor} until reaching the final two columns, i.e., the second tensor network in the second row of Fig.~\ref{fig:tn-contraction}(\textbf{a}).
We denote by $D$ the maximum bond dimension of the left-most-column MPS during the contraction, and its right-adjacent column is treated as another MPS. The expectation value $\langle Z_{(0,1)}Z_{(0,0)}\rangle_{\theta}$ is approximated by computing the inner product between these two resulting MPSs. We denote the evaluation results as $\langle Z_{(0,1)}Z_{(0,0)}\rangle_{\theta,D}$ with the given maximum bond dimension $D$.

Fig.~\ref{fig:tn-contraction}(\textbf{b}) plots the maximum bond dimension $D$ as a function of lattice size $L$, where we set different truncation levels from $10^{-8}$ to $10^{-20}$ marked by different colors.
We see that the maximum bond dimension rapidly saturates as $L$ increasing for all truncation levels considered. Thus, for a given truncation level~\cite{ITensor}, the required bond dimension during contraction has no dependence on the system size, and the computational cost of each MPO--MPS multiplication in each column of the $L\times L$ lattice scales as $\order{L}$.
Since the MPO--MPS multiplication is repeated for $\order{L}$ times, the overall time complexity of the tensor-network contraction scales as $\order{L^2}$. This quadratic computational cost is consistent with the numerical time consumption shown in Fig.~\ref{fig:tn-contraction}(\textbf{c}).

To demonstrate the effect of criticality to the PEPS evaluation, Fig.~\ref{fig:tn-contraction}(\textbf{d}) plots the difference between $\langle Z_{(0,1)}Z_{(0,0)}\rangle_{\theta,D}$ with $D<30$ and $\langle Z_{(0,1)}Z_{(0,0)}\rangle_{\theta,30}$ versus $\theta$. Larger deviations are observed in the vicinity of the critical points $\theta=\theta_c=\pm \pi/4$. Therefore, the PEPS evaluation is more sensitive to the bond dimension truncation near the phase transition point of the Ising PQC, again showcasing the increased classical simulation complexity due to the criticality of PQCs.

\end{document}